\documentclass[11pt]{article}
\usepackage[utf8]{inputenc}
\usepackage{amssymb}
\usepackage{comment}
\usepackage{xcolor}
\usepackage{caption}
\usepackage{subcaption}
\usepackage{fullpage}
\usepackage{graphicx}
\newif\ifshow
\showfalse

\title{Translationally Invariant Constraint Optimization Problems}
\author{
  Dorit Aharonov\thanks{{\tt dorit.aharonov@gmail.com}}
  \and
  Sandy Irani\thanks{{\tt irani@ics.uci.edu}}
}
\date{\today}

\include{algo.tex}

\usepackage{graphicx}

\newcommand{\tile}{{\cal{T}}}
\newcommand{\pnexp}{\mbox{P}^{\mbox{NEXP}}}
\newcommand{\fpnexp}{\mbox{FP}^{\mbox{NEXP}}}

\newcommand{\fpnp}{\mbox{FP}^{\mbox{NP}}}
\newcommand{\pqma}{\mbox{P}^{\mbox{QMA}}}
\newcommand{\expnexp}{\mbox{EXP}^{\mbox{NEXP}}}

\newcommand{\nexp}{\mbox{NEXP}}
\newcommand{\qma}{\mbox{QMA}}

\newcommand{\leftb}{\vartriangleleft}
\newcommand{\rightb}{\vartriangleright}
\newcommand{\next}{\mbox{next}}
\newcommand{\inner}{\mbox{IS}}
\newcommand{\cald}{{\cal{D}}}

\newcommand{\barx}{\overline{X}}

\newcommand{\barzero}{\overline{0}}
\newcommand{\barone}{\overline{1}}
\newcommand{\bartwo}{\overline{2}}
\newcommand{\barthree}{\overline{3}}

\newcommand{\barj}{\overline{j}}
\newcommand{\bark}{\overline{k}}
\newcommand{\barn}{\overline{n}}
\newcommand{\cost}{{\cal{C}}}

\newcommand{\redx}{{\color{red} X}}
\newcommand{\redbarx}{{\color{red} \overline{X}}}
\newcommand{\redb}{{\color{red} B}}

\newcommand{\redc}{{\color{red} c}}

\newcommand{\bluex}{{\color{blue} X}}
\newcommand{\bluebarx}{{\color{blue} \overline{X}}}
\newcommand{\blueb}{{\color{blue} B}}

\newcommand{\bluec}{{\color{blue} c}}

\newtheorem{theorem}{Theorem}[section]

\newtheorem{claim}[theorem]{Claim}
\newtheorem{lemma}[theorem]{Lemma}

\newtheorem{fact}[theorem]{Fact}

\newcommand{\qed}{\mbox{\ \ \ }\rule{6pt}{7pt} \medskip}
\newenvironment{proof}{\noindent{\em Proof:}}{\hfill\qed}

\newtheorem{definition}{Definition}

\begin{document}

\maketitle
\thispagestyle{empty}

\begin{abstract}
    We study the complexity of classical constraint satisfaction problems on a 2D grid. Specifically, we consider the computational complexity of function versions of such problems, with the additional restriction that the constraints are {\it translationally invariant}, namely, the variables are located at the vertices of a 2D grid and the constraint between every pair of adjacent variables is the same in each dimension. The only input to the problem is thus the size of the grid. 
    This problem is equivalent to one of the most interesting problems in classical physics, namely, computing the lowest energy of a classical system of particles on the grid. 
    We provide a tight characterization of the complexity of this problem, and show that 
    it is complete for the class $\fpnexp$. 
    Gottesman and Irani (FOCS 2009) also studied classical constraint satisfaction problems using this strong notion of translational-invariance; they show that the problem of deciding whether the cost of the optimal assignment is below a given threshold is $\nexp$-complete. Our result is thus a strengthening of their result from the decision 
    version to the function version of the problem. 
    Our result can also be viewed as a generalization 
    to the translationally invariant setting, of Krentel's famous result from 1988, showing that 
    the function version of SAT is complete for the class $\fpnp$. 
    
    An essential ingredient in the proof is 
    a study of the computational complexity of a gapped variant of the 
    problem. We show that it is $\nexp$-hard to approximate the cost of the optimal assignment
to within an additive error of $\Omega(N^{1/4})$, where the grid size is $N \times N$. To the best of 
our knowledge, no gapped result is known for CSPs on the grid, even in the non-translationally invariant case. 
This might be of independent interest. 
As a byproduct of our results, we also show that 
a decision version of the optimization
problem which asks whether the cost of the optimal assignment is odd or even
is also complete for $\pnexp$.
\end{abstract}

\pagebreak 
\thispagestyle{empty}

\tableofcontents
\thispagestyle{empty}

\pagebreak
\setcounter{page}{1}
\section{Introduction}
More than half a century ago, 
Cook and Levin inaugurated the field of NP-completness. 
The fact that the Constraint Satisfaction Problem (CSP) is NP-complete has been the cornerstone of our understanding
and approach to   important optimization problems arising in countless applications.
The NP completeness of deciding whether a CSP instance is satisfiable, plays an important role also in physics. This is because constraint satisfaction corresponds to a classical local Hamiltonian which expresses the total energy of a system of particles; the energy is the sum of terms,
each of which describes the energy interaction between constant-sized clusters of particles. 
Finding whether the lowest energy of such systems is
below some threshold or above it, is a special class of CSP, and was famously shown to be  NP-complete in many cases by Barahona and others
\cite{B82,I00}. The theory of NP-completeness has a natural generalization in the quantum setting;  
the Cook-Levin theorem was generalized by Kitaev \cite{KSV02} to show that the following problem is
QMA-complete: 
Given a local Hamiltonian with quantum energy interactions 
describing the  energy  in a
quantum many-body system, decide whether the ground energy is above some value or below another; this problem has been intensely studied in recent years \cite{KSV02, Gharibian_2015, OT05, Aharonov_2009}.

A natural variant of the NP-complete decision CSP, is the function version of the CSP. In this version, the question is not whether all constraints can be satisfied, but what is the {\it maximum number} of constraints that can be satisfied by any assignment. 
One could also consider a weighted variant, as we do here, where the goal is compute the cost of the optimal assignment 
which minimizes the weighted sum of violated constraints.
Computing the cost of the optimal
solution is in fact the more natural version of CSP problems in a large number of combinatorial applications (to give just two examples, max-cut and max independent set). Importantly, in classical physics, when considering the  local Hamiltonian corresponding to a CSP, the function version of the problem is in fact the problem of finding the lowest possible energy for the Hamiltonian over all possible states -- one of the most important notions in physics. 
In the quantum case, the function version corresponds to providing an estimate of the 
lowest energy  of the system over all possible quantum states, which is one of the 
main subjects of interest in condensed matter physics. 

What is known from the computational complexity angle, 
about this 
physically motivated question, the function versions of CSPs? In 1988, 
Krentel \cite{Krentel} proved that the function problem for constraint satisfaction is 
$\fpnexp$-complete. Krentel's proof is  significantly more involved technically
that that of the Cook-Levin's theorem which characterizes the complexity of 
the decision variant of CSPs. Furthermore, and in stark contrast to the theory of
decision problems and NP-completeness, the function version of CSP seems to have received less attention in the TCS literature. 

From the physics point of view, an additional point becomes extremely important. By and large, physicists study local Hamiltonians, be them classical or quantum, in a 
{\em translational-invariant} (TI) setting. In this setting the particles are located at the vertices of a geometric lattice
and 
all the terms acting on adjacent pairs of particles along a particular dimension are the same.  
In particular, the model most relevant to physics is TI in a very strong sense: the dimension of the individual particles and the Hamiltonian term acting on each pair of adjacent particles
in a lattice are fixed parameters of the problem. 
When considering finite systems, the only input is  an integer $N$ indicating the size of the system.
This set-up corresponds to the fact that in physics, different Hamiltonians represent completely different physical systems.
Thus, studying the ground energy  (or some other quantity) in the AKLT model is considered to be a completely different problem than studying the same quantity in, say, the Ising model \footnote{
Some of the recent work on the quantum TI local Hamiltonian problem \cite{Bausch_2017} adopt a  weaker notion in which the input also includes the Hamiltonian term  that is applied to each pair of particles, allowing the Hamiltonian to be tuned to the size of the system. This 
model is mainly considered in quantum Hamiltonian complexity, but have not been a topic of study in physics.}.

For decision problems, 
Gottesman and Irani \cite{GI} show hardness results for both quantum and classical TI Hamiltonians. Since the size of the grid can be given 
by logarithmically many bits, and there is no other input, one encounters 
an exponential factor compared to the common CSP problems; thus the results in 
\cite{GI} show NEXP and $QMA_{EXP}$ completeness for the 
classical and quantum variants of the problem, respectively. 
We note that a tightly related line of works studies TI   
infinite systems \cite{AI21, CW21, C15} and considers 
computability and computational complexity of decision problems in that domain, namely in the so called
thermodynamic limit. 
Although the focus here is on finite systems, constructions for finite systems have played an important role for the results in the thermodynamic limit. In particular, all the results in \cite{AI21, CW21, C15} use a finite construction layered on top of a certain type of aperiodic tiling of the infinite grid. 
 
To the best of our knowledge, the 
computational complexity of function CSPs in the TI setting, has remained open.
In this paper we provide a tight characterization of its complexity, and show that the function version of TI CSP on a $2$-dimensional 
grid is complete for $\fpnexp$. This result thus strengthens Krentel's construction for general CSPs
to apply even TI systems for two and higher dimensions. 
The result is also a generalization of Gottesman-Irani who prove hardness for 2D TI systems for the standard decision problem, where one only needs to determine if the ground energy is below a given threshold.
One of the key technical challenges in our result is to effectively create large ($\Theta(N^{\epsilon})$)
costs on an $N \times N$ grid using only two constant-sized terms which apply in the horizontal and vertical directions. Thus, as a stepping stone to the more complex result on the function version of TI 2D CSPs, we show a fault-tolerant result which we believe is of interest on its own, namely that it is $\nexp$-complete to even approximate the ground energy to within an additive $\Theta(N^{1/4})$.

\subsection{Problem Definitions, Results and 
Main Challenges}
It is most convenient to present our results using 
the language of the weighted tiling problem, where we focus here on the two dimensional case\footnote{
Our version of tiling is equivalent to
the more common Wang tiles \cite{Wang60}.}. 
In this tiling problem, one is  asked
to tile an $N \times N$ 2D grid with a set of $1 \times 1$ tiles. The tiles
come in different colors
and only some pairs of colors can be placed next to
each other in either the horizontal or vertical directions. More precisely, 
a set of tiling rules $\tile$ is a triple $(T, \delta_H, \delta_V)$,
where $T$
is a finite set of tile {\em types} $T = \{t_1, \ldots, t_d\}$,
and $\delta_H$ and $\delta_V$ are functions from $T \times T$ to
$\mathbb{Z}$.
For $(t, t') \in T \times T$, $\delta_h(t,t')$ is the cost of putting a tile
of type $t$ immediately to the left of a tile of type $t'$
and $\delta_v(t,t')$ is the cost of putting a tile
of type $t$ immediately above a tile of type $t'$. 
Let $\lambda_0(\tile(N))$ be the minimum cost of tiling an $N \times N$ grid with tiling rules
$\tile$. The goal is to tile the grid with minimal total cost. Note that  this problem is directly analogous
to a classical Hamiltonian in 2D.   
We first define a function version of the problem.


\begin{definition}
{\sc $\tile$-FWT (Function Weighted Tiling)}

\item \textbf{Input:} An integer $N$ specified with $\lfloor \log N + 1 \rfloor$ bits
 
\item \textbf{Output:} $\lambda_0(\tile(N))$
\end{definition}

\begin{theorem}
\label{th-function}
{\bf {(Main)}} There exists a set of tiling rules
$\tile$ such that $\tile$-FWT is $\fpnexp$-complete. 
\end{theorem}

We note that the fact that the function problem is complete for 2D immediately implies that it is complete for any grid of dimension at least $2$ since the 2D construction can be embedded into a higher dimensional grid.
The 1D CSP case is poly-time computable using dynamic programming.

The upper bound in Theorem \ref{th-function} is easy: it can be achieved by binary search with access to
an oracle for the decision problem.
For the lower bound, one encounters a challenge. The reduction must encode in the tiling rules the computation
of a polynomial time Turing Machine with access to a $\nexp$ oracle.
 If an instance given to the oracle is a {\em yes}
instance, the computation of the verifier can be encoded into the tiling rules.
However {\em no} instances cannot be directly verified in this way. 
Krentel's proof  
that the function problem of weighted SAT is $\fpnp$-complete \cite{Krentel} overcomes this challenge; let us recall it and then explain 
the problem in carrying it over to the TI setting. 
Krentel uses an accounting scheme \cite{Krentel,Papa} that applies a cost to every string $z$ representing guesses for the sequence of responses to all the oracle queries made. The accounting scheme needs to ensure that the minimum cost $z$ is equal to the correct sequence of oracle responses, $\tilde{z}$. 
{\em yes} and {\em no} guesses are treated differently, due to the fact that   
the verifier can check {\em yes} instances (and thus incorrect {\em yes} guesses can incur a very high cost), but {\em no} guesses, cannot be directly verified. In Krentel's scheme, {\em no} 
guesses 
incur a more modest cost, whether correct or not, and their cost 
must 
decrease exponentially. This is because the oracle queries are adaptive;  
an incorrect oracle response could potentially change all the oracle queries made in the future
and so it is important that the penalty for an incorrect guess on the $i^{th}$ query is higher than the cost
that could  potentially be saved on all future queries. 
The weights on clauses that implement this accounting scheme are multiplied by a large 
power of two
to ensure that they are the dominant factor in determining the optimal assignment. 

The difficulty in applying Krentel's accounting scheme in the TI setting
is that the costs must grow with the size of the input. Therefore, it is
not possible to apply the costs directly into the tiling rules which are of fixed
constant size. 
%
A natural attempt to circumvent the problem is 
to assign the required large penalty by many tiles, each of which would acquire a constant penalty; however, the problem in 
implementing this approach is that Cook-Levin 
type reductions from computations to tilings 
are very brittle, as a single error can
potentially derail the entire computation.
For example, imagine inserting a row that does not have a Turing Machine head. There will
be a single fault where the head disappears from one row to the next,
but every row thereafter will contain
the unchanging contents of the Turing Machine tape without a head to execute 
a next step. This imposes a challenge since when  
enforcing large costs by using many tiles, or constraints, we need to make sure that many of these constraints are indeed violated in order to incur the required large 
penalty. 

We provide a construction which circumvents this issue 
by exhibiting some {\it fault tolerance} properties. 
We thus prove what can be viewed as a gapped version or a hardness of approximation result, which is then a natural stepping stone to implementing the more intricate function required in Krentel's accounting scheme. 
To this end we define an approximation version of weighted tiling:


\begin{definition}
{\sc $(\tile,f)$-GWT (Gapped Weighted Tiling)}

\item \textbf{Input:} An integer $N$ specified with $\lfloor \log N + 1 \rfloor$ bits. Two integers $a$ and $b$ such that $b-a \ge f(N)$.

\item \textbf{Output:} Determine whether $\lambda_0(\tile(N)) \le a$ or $\lambda_0(\tile(N)) \ge b$.
\end{definition}

\begin{theorem}
\label{th-gapped}
There exists a set of tiling rules
$\tile$ such that $(\tile,f)$-GWT is $\nexp$-complete for a function $f(n) = \Omega(N^{1/4})$.
\end{theorem}

This shows that it is $\nexp$-hard to even approximate the cost of the optimal tiling
to within an additive error that is $\Omega(N^{1/4})$.
This can be viewed as a gapped version of the results of \cite{GI}; the proof constructs a reduction mapping
the computation into a tiling 
such that even in the presence of $O(N^{1/4})$ faults,
the computation encoded by the tiling is able to proceed and produce approximately correct results.

 
Theorem \ref{th-gapped} is of potential interest on its own. It might resemble a PCP type result, 
but the model we consider differs from the standard PCP setting in two ways: the first is that the underlying graph is a grid, rather than a graph with much higher 
connectivity, and the second is translational-invariance.  It is not  possible to obtain a hardness of approximation result
with an additive error that is linear in $N$ (as one has in the PCP theorem) on any finite dimensional lattice because
such graphs do not 
have the necessary expansion properties. For example, in 2D, one could divide the grid into $b \times b$ squares
for $b = \Theta(\sqrt{\log N})$ and solve each square optimally in polynomial time. The resulting solution
would be within an additive $N/ \sqrt{\log N}$ of the optimal solution. To the best of our knowledge, no 
gapped version was proven before for CSP problem set on a constant dimensional grid,  
even without the TI restriction.  


Finally, our results provide tight characterizations of the complexity of the following 
decision problem; 
\begin{definition}
{\sc $\tile$-PWT (Parity Weighted Tiling)}

\item \textbf{Input:} An integer $N$ specified with $\lfloor \log N + 1 \rfloor$ bits

\item \textbf{Output:} Determine whether $\lambda_0(\tile(N))$ is odd or even.

\end{definition}

The proof is very similar to the proof  of Theorem \ref{th-function}. 
The result on Parity Weighted Tiling illustrates that decision problems related to CSP can be complete for an oracle class just like the function problem. The crucial difference between the threshold decision problem (is the cost of the optimal solution less than $t$?) which is $\nexp$-complete and the parity problem which is $\pnexp$-complete
is that the parity problem still seems to require determining the optimal cost. This seems to make the characterization of its complexity as challenging as for the function version of the problem.

{~} 

\noindent{\bf Organization of remainder of the introduction:} We next proceed to an overview of the proofs. We start with the setup of tiling rules and layers in Subsection \ref{sec-buildingBlocks}. The proofs of the Theorems 
are given in subsection \ref{sec-introfunction}.
We end with related work and open 
questions in Subsection \ref{sec-discussion}. 


\subsection{Tiling Rules and Layers}
\label{sec-buildingBlocks}

We assume that there is a special tile
denoted by $\Box$ which must be placed around the perimeter of the grid to be tiled. Moreover, no $\Box$
tile can be placed in the interior of the grid. We will return later to enforcing this condition in the context of the different problems. The tiles on the interior will be composed of multiple layers where each layer has its own set of
tile types. A tile type for an internal tile
in the overall construction is described by a tile type for each
of the layers.

For ease of exposition,
we allow our tiling rules to also
apply to local {\it squares} of four tiles.  
This can easily then be translated to two-local constraints on tiles, as in our definition of the tiling problem. 
This simple transition is described more fully in Section \ref{sec-tiling}.
For the remainder of the paper our tiling rules include constraints on local squares of four tiles, as well as pairs of horizontal tiles.

If the four tiles
in a square are all interior tiles, then each possible
pattern of four square tiles within a layer will
be designated as legal or illegal. The overall cost of placing four interior tiles
in a local square together
will be function of whether the square for each layer is legal or illegal.
For the Gapped Weighted Tiling, the cost will be just the number of layers for which
the square pattern is illegal. For the Function Weighted Tiling and Weighted Tiling Parity, illegal squares at different
layers will contribute  different amounts to the cost.

In general, a no-cost tiling of each Layer represents a computational process
where each row represents the state of a Turing Machine.
The computation reverses direction from one layer to the next.
The rows of a tiling of an $N \times N$ grid will be numbered $r_0$ through
$r_{N-1}$ from bottom to top.
When referring to the rows in a particular layer, we will exclude the border
rows and order the rows according to the computation direction.
So the first row of Layer $1$, which proceeds from bottom to top,
is row $r_1$ and the last row of Layer $1$ is $r_{N-2}$.
Layer $2$ proceeds from top to bottom, so the first row for Layer $2$ is
$r_{N-2}$ and the last row is $r_1$.

 For the most part, the rules governing the tiling apply to the tile
types within
each individual layer. 
The different layers only interact at the lower and upper border of the grid.
This is how the output of one process (on Layer $i$) is translated into the input
for the next process (on Layer $i+1$).
For example, a square may be illegal if the two lower tiles
are $\Box~\Box$, and the two upper tiles violate certain constraints between
the Layer $i$ and Layer $i+1$ types.
Some of the layers will also have additional constraints on which tiles can be next to each other in the horizontal direction.
Each type of violated constraint 
is given a name described below.

\begin{definition} {\bf [Faults in a Tiling]}
An occurrence of any of the illegal patterns described 
in the constructions
is called a {\em fault}.
A tiling with no faults, will correspond to a {\em fault-free} computation.
\end{definition}

There will be some additional costs (described later)
associated with a computation
ending in a rejecting state. These are not considered faults because
they can happen in correct computations. 
Figure \ref{fig-squareTypes} illustrates the different types of tiling constraints.

\begin{description}
\item{\bf Illegal Computation Squares:} For each layer, every pattern of four tile types will
be designated as a {\em legal computation square} or an {\em illegal computation square}.
In general, these rules enforce that the tiling within the layer
represents a consistent execution
of a Turing Machine. We describe in Subsection \ref{sec-TM2Tile} how to translate the rules of a Turing Machine into 
legal and illegal computation squares.
\item{\bf Illegal Pairs:} Some of the layers will have additional constraints on which tiles can be placed next to each other in the horizontal direction. 
Each ordered pair of tiles types for that layer will be designated as a {\em legal pair}
or an {\em illegal pair}. 
\item{\bf Illegal Initialization Squares:}
For each layer, there are also some initialization rules that constrain the initial configuration of the Turing Machine.
If the layer runs bottom to top, then these rules apply to $r_0$,
which consists of all $\Box$ tiles, and the first row of the layer.
For example, if tile $t_1$ can not be immediately to the left of $t_2$
in the first row of Layer $i$, then
 the square with $\Box~\Box$ directly below $t_1~t_2$ is an 
 {\em illegal initialization square} for Layer $i$.
If the Turing Machine for the layer runs top to bottom, then the square with 
$\Box~\Box$ directly above $t_1~t_2$ in Layer $i$ is illegal.
\item{\bf Illegal Translation Squares:}
Finally, we add rules that control how
the last row of Layer $i$ is translated to the first row of Layer $i+1$.
If Layer $i$ runs top to bottom, then the rules apply to rows $r_0$
and $r_1$.
For example, if tile $t$ in Layer $i$ cannot be translated to $t'$ in Layer $i+1$,
then any square with a $\Box$ directly below a tile whose Layer $i$ type is $t$
and whose Layer $i+1$ type is $t'$ would be illegal.
The translation rules can also apply to pairs of adjacent tiles. E.g., 
it could be illegal to have a square whose bottom two tiles are
$\Box~\Box$ and whose top two tiles have $t_1~t_2$ in Layer $i$ and $t_3~t_4$ in Layer $i+1$.
\end{description}

\begin{figure}[ht]
\centering
 \begin{subfigure}{0.14\textwidth}
  \centering
  \includegraphics[width=\textwidth]{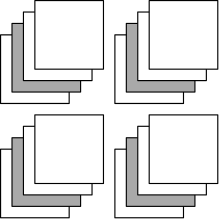}
  \caption{Computation}
  \label{fig-squares}
 \end{subfigure}%
 \hspace{.4in}
 \begin{subfigure}{0.14\textwidth}
  \centering
  \includegraphics[width=\textwidth]{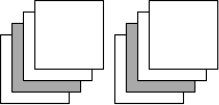}
  \caption{Illegal pairs.}
  \label{fig-pairs}
 \end{subfigure}
 \hspace{.4in}
 \begin{subfigure}{0.14\textwidth}
  \centering
  \includegraphics[width=\textwidth]{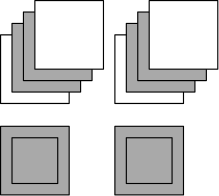}
  \caption{Translation.}
  \label{fig-trans}
 \end{subfigure}%
 \hspace{.4in}
 \begin{subfigure}{0.14\textwidth}
  \centering
  \includegraphics[width=\textwidth]{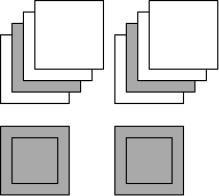}
  \caption{Initialization.}
  \label{fig-init}
 \end{subfigure}
 \caption{\small{Interior tiles have four layers. Border tiles have one layer and are labeled with the $\Box$ symbol. (a) An illegal computation
 square for Layer $2$. The constraint applies
 to the four tile types for Layer $2$ shown in gray. (b) An 
 illegal pair for Layer $2$. The constraint applies
 to the two adjacent tile types for Layer $2$ shown in gray.
  (c) An illegal translation square from Layer $2$ to Layer $3$. The constraint applies to two border tiles and the tile types for Layers $2$ and $3$
 for the other two interior tiles.  (d) An illegal initialization  square for Layer $2$. The constraint applies to two border tiles and the  Layer $2$ tile types
 for the other two interior tiles.}}
 \label{fig-squareTypes}
\end{figure}

\subsection{Proofs Overview}

\noindent {\bf Thoeorem \ref{th-gapped}: Gapped Weighted Tiling}
\label{sec-introgapped}
Recall that the standard encoding of a Turing Machine
into tiling rules is very brittle in that a single fault 
can derail the entire computation. The most straight 
forward way to overcome this is using a construction which embeds many repetitions 
of the computation, so that many faults would be required to derail a large number of those computations. 
Multiple computations thus need to be set up and initiated, using
 a single faulty Turing machine with  TI rules.  In our construction,
this is achieved by a first stage of the 
computation (implemented in Layer $1$, as we describe below), which, roughly, creates intervals in the top row of Layer $1$, such that the independent repetitions of the computations will occur in different strips on the grid; the boundaries of the strips are determined by those intervals. 
The difficulty is how to 
implement the initial set up using a single 
Turing machine, in a fault tolerant way. We now describe the details. 

The tiling rules for the first two layers, as well
as the reduction mapping $x$ to $N$ are independent of the language $L \in \nexp$,
the language we are reducing from.
Let $V$ denote the exponential time verifier for $L$.
Subsection \ref{sec-TM2Tile} describes the details of how a Turing Machine computation
is encoded in a set of tiling rules. In general tiles will be either {\em tape} tiles
which encode a single symbol from the Turing Machine's tape alphabet or {\em head}
tiles which encode both the state of the Turing Machine as well as the current 
tape symbol to which the head is pointing.

The Turing Machine computation represented in Layer $1$ starts with
two non-blank symbols and proceeds to write a sequence of
intervals on the tape, where an interval is a sequence
of $B$ symbols bracketed on either side by a {\em delimiter}
tile from the
set $\{X, \barx, \leftb, \rightb\}$.
The Turing Machine just repeatedly executes a single loop which we refer to as the {\em Outer Loop}. In one iteration of the
Outer Loop, an additional $B$ symbol is inserted into every interval and a new
interval with no $B$'s in the middle is added to the right end of the
non-blank symbols. In an fault-free execution of the Turing Machine,
after $m$ iterations of the loop, there are $m+1$ intervals.
The number of symbols in each interval (including the delimiter tiles on either end)
is $m+2, m+1, m, \ldots, 2$. For $m = 4$, the row should look like:

\vspace{.1in}

\begin{center}
\noindent
\begin{tabular}{|c|c|c|c|c|c|c|c|c|c|c|c|c|c|c|c|c|}
\hline
 $\Box$ & $(q/\leftb)$ & B & B & B & X & B  & B & X & B & X  & $\rightb$ & \# & \# & $\cdots$ & \# & $\Box$  \\
\hline
\end{tabular}
\end{center}

\vspace{.1in}

When the top row of Layer 1 is translated to Layer $2$, the head tile for the Layer $1$
Turing Machine is translated to a tape tile (so the state information is lost)
and a head tile is inserted on the left end of every interval.
For example, an interval $X~B~B~B~\cdots~B~X$ at the end of Layer $1$ is translated to
$X~(q_s/S)~B~B~\cdots~B~T~X$ in the first row of Layer $2$.
In Layers $2$ and $3$ the sizes and locations of the intervals do not change
within a row
unless the interval contains an illegal square.
Thus, a single interval over all the rows of Layer $2$ forms a vertical
strip of tiles,
and a separate, independent computation takes place within each strip.
See Figure \ref{fig-strips} for an example.
Once the intervals are created on Layer $1$,
each computation on Layers $2$ and $3$ is  fault-free unless the strip contains an illegal square. Thus, the number of illegal squares is at least the number of strips
that fail to complete their computation correctly.

In Layer $2$, the computation is just a binary counter Turing Machine that continually
increments a binary counter. All the strips that do not contain an
illegal square will have the same string $x$ represented in the final
row of Layer $2$. 
The string $x$ then serves as the input to the computation in Layer $3$.
The binary counter Turing Machine in Layer $2$ runs for exactly $N-3$
steps. The reduction is the function that maps $x$ to $N$, where 
the string $x$ is written on the tape of a binary counting Turing Machine
after $N-3$ steps. 
Lemma \ref{lem-bctm} gives an exact formula mapping $x$ to $N$
and shows that the value of the number represented by the string $x$ is
$\Theta(N)$, the dimension of the grid. The idea of using a binary counting
Turing Machine to translate  the size of
the grid to a binary input for a computation was used previously
in \cite{GI}. Although since the construction in \cite{GI}  had a gap of $1$,
only a single execution of the verifier was needed. Since we are trying to
produce a gap of $f(N)$, we need at least $f(N)$ separate computations
each of which simulates the verifier on input $x$. 

Each interval $X~(q_s/S)~x~B~\cdots~B~T~X$ is translated unchanged to Layer $3$.
The computation in each strip in Layer $3$ simulates the verifier on
input $x$ using a 
witness that is guessed in the tiling. There is a final cost 
for any rejecting computation.
If $x \in L$, it will be possible to tile each strip at $0$ cost.
If $x \not\in L$,   every strip will  contain
an illegal square or 
will incur a cost for the correct rejecting computation. 
Thus, the gap is essentially created by these parallel computations,
each of which contributes a constant cost if $x\not \in L$.

Since the sizes of the intervals go down to $0$, some of the intervals
will be too narrow to complete the computation in either Layers $2$ or $3$.
If the head ever hits the right end of its interval, it transitions to an
infinite loop, causing no additional cost.
A standard padding argument (Claim \ref{claim-pad}) guarantees that an interval need only 
be $\Theta(N^{1/4})$ wide to complete the computations in Layers $2$ and $3$.
The analysis of Layer $1$ then needs to guarantee that despite the faults,
there will be sufficiently many sufficiently wide intervals. 

The main challenge in the proof is in making the computation in Layer $1$
fault-tolerant, meaning that each illegal pair or square cannot derail the
computation too much. 
The horizontal rules in Layer $1$ are critical for
enforcing that this cannot happen. We show that a row in the tiling that has no illegal
pairs corresponds to a sensible configuration of the Turing Machine. In particular
such a row has
exactly one head tile that lies in between the $\leftb$ and $\rightb$ tiles.
Note that faults can still alter the computation in potentially strange ways. 
Nonetheless, we also show that starting from a row with no illegal pairs, the Layer $1$ Turing Machine will be able
to make progress, and after a sequence of fault-free steps (corresponding to
a sequence of rows containing no illegal squares), the computation
will  perform a complete
 iteration of the loop. Since the number of illegal pairs and squares is
bounded by $O(N^{1/4})$, there are enough complete iterations of the
loop to ensure that the last row of Layer $1$ has enough intervals that are wide enough to complete the
computations in Layers $2,3$. 

By far the most technically involved part of the paper is the analysis of Layer $1$ described in Section
\ref{sec-L1analysis}. All of the results use the final Lemma \ref{lem-analysisL1}, which gives
a  tight characterization of
the difference between the final row in Layer $1$ of a fault-free tiling and the final row of a tiling with faults. 
In fact, the result on  Gapped Weighted Tiling could be established with looser bounds,
but we provide the analysis once in Section \ref{sec-L1analysis} in a form that can be used for
all the results in the paper. Section \ref{sec-introfunction} describes more fully how this tight characterization
is accomplished. A more in-depth overview is then given in the beginning of Section \ref{sec-L1analysis}.

{~}

\noindent{\bf Theorem \ref{th-function}: Weighted Tiling Function}
\label{sec-introfunction}
The hardness reduction for Function Weighted Tiling reduces from an oracle class. The function $f$ is computed by a polynomial time Turing Machine $M$ with access to an oracle for language $L' \in \nexp$. Let $V$ denote the exponential time verifier for $L'$. Using a standard padding argument 
(see for example Lemma 2.30 from \cite{AI21}) we can assume that 
for a constant $c$ of our choice, for every $|x| = n$, there is a $\barn \le cn$,
such that the size of
$f(x)$ is at most $\barn$, and M makes at most $\barn$ oracle calls to $L'$.  Let $z$ denote an $\barn$-bit string denoting the responses to the oracle queries made on input $x$. With $x$ and $z$ fixed, the set of inputs to the oracle $(o_1, \ldots, o_{\barn})$ is also determined.  $V(o_j)$ is an indicator function denoting whether $o_j$ is in $L'$. Note that since $L'$ is in $\nexp$, if 
$V(o_j) = 1$, there exists a witness that will cause the verifier to accept and if $V(o_j) = 0$, V will always reject regardless of the witness.  
Define:

$$\cost(x,z) =  \sum_{j=1}^{\barn} \left[
(1 - z_j)\cdot 2^{\barn - j} + z_j \cdot (1-V(o_j)) \cdot 2^{\barn} \right]$$

Let $f(x,z)$ be the output of Turing Machine $M$ on input $x$ with oracle responses $z$.
Note that since $|f(x,z)| \le \barn$, $f(x,z) \le 2^{\barn}$.
The construction will ensure that  the  minimum cost tiling for a particular $x$ and $z$ will be
$2^{\barn+5} \cost(x,z) + 2^3 \cdot f(x,z).$
Note that $\cost(x,z)$  represented in  the $\barn$ high-order bits of the cost has
the necessary structure where the costs for a {\em no} oracle response decrease exponentially in $j$, the index of the oracle query. 
The cost for a {\em yes} guess  will be $0$ if the input to the
oracle $o_j$ is in fact in $L'$ (i.e., $V(o_j)=1$) and will be a very large cost
of $2^{2\barn+5}$ if $o_j$ is not in $L'$.
This  function will guarantee that the overall cost is minimized
when $z$ is the correct string of oracle responses. In addition, the low order bits encode the output
of the function $f(x,z)$. So if the minimum cost tiling can be computed, this will correspond to $f(x)$,
which is $f(x, \bar{z})$, where $\bar{z}$ is the string of correct oracle responses. The factor of $8$ ensures that
even if the minimum cost is off by $\pm 3$, the value of $f(x)$ can still be recovered.

So far what we have described just implementing the original accounting scheme devised by Krentel. The challenge is to implement this cost function in 2D with TI terms. Note that since the tiling rules are fixed parameters of the problem, it is not possible to encode the cost function  directly into the penalty terms. As with the Gapped Weighted Tiling problem the function is collectively computed by a set of parallel processes within each strip created by the intervals from Layer $1$. However, instead of a threshold function which is either $+f(N)$ or $0$, the parallel processes must collectively compute the more intricate function described above,
which requires that the individual processes have some additional information.

We will describe first what happens in a fault-free computation (with no illegal pairs or computation squares)
and then describe how fault-tolerance is enforced and proven.
The construction for Layers $1$ and $2$ are exactly the same as for the Gapped Weighted Tiling problem. 
Layer $1$ creates a set of intervals. We definite the function $\mu(N)$ to denote the number of intervals on the tape if the Turing Machine for Layer $1$ executes $N-3$ steps. If after $N-3$ steps,
the computation just happens to finish at the end of an execution of the Outer Loop then the intervals have sizes (from left to right) $\mu(N)+1, \mu(N), \ldots, 2$.
If the computation finishes in the middle of an execution of the Outer Loop, the actual sequence of interval sizes
will be close to $\mu(N)+1, \mu(N), \ldots, 2$. The largest interval could have size $\mu(N)+2$
and there may be  a couple missing values in the range where the current interval is being increased. 
Lemma \ref{lem-errorFreeSizes} describes the possible deviations in detail.
$\mu(N)$ is $\Theta(N^{1/4})$ and we show using a standard padding argument that for the constant $c$ of our choice,
all of the computations require at most
$c \mu(N)$ space. This allows us to establish that  at least half of the intervals will be large enough to complete the required computations. 

As in the previous construction, Layer $2$ then executes a binary counting Turing Machine which results
in the string $x$ written to the left of each interval which is large enough to complete the computation.
Note that Layer $1$ is a {\em global} Turing Machine which executes a single process across the entire grid,
while Layer $2$ represents {\em local} computations within each strip.
When $x$ is translated from Layer $2$ to Layer $3$ it is augmented with a guess string $z$ for the oracle queries.
$x$. 
However, there is no guarantee that the
guess for each interval is the same. Note that $z$ can be arbitrary but it must be consistently the same for each
interval. Layer $3$ then executes a global Turing Machine which imposes a high penalty if the  $z$ strings in each
strip are not all the same. This penalty is higher than the cost function  for any $z$, so the lowest
cost tiling will correspond to a configuration in which each strip has the same $x$ and $z$.

Finally, in Layer $4$, there is a local computation in each interval, each
of which makes a $+1$ or $0$ contribution towards the overall cost.
The computation within each interval requires a unique tag in order to determine
which term of the cost it will contribute to.
The tag comes from the size of the interval. The computation begins with counting
the number of locations in the interval. This can be accomplished by having the head shuttle back and forth
between the two ends of the interval implementing both a unary and binary counter until the unary counter
extends across the entire interval. The head returns to the left end of the interval and begins the next
phase of the computation. Since the size of an interval is at most
$O(N^{1/4})$ this phase of the computation will take at most $O(N^{1/2})$ steps.

Now each computation has the same pair $(x,z)$ and a its own integer $r$ indicating
the size of the interval. From $x$, the values of $N$ and $\mu(N)$ can be determined. 
Lemma \ref{lem-errorFreeSizes} shows that in a fault-free computation, the sizes of the 
intervals will decrease from left to right. Moreover, all interval sizes are in the set
$\{ \mu(N)+2, \mu (N)+1, \ldots, 2\}$ with at most one missing value from that set and at most two duplicates.
Thus, the value $\mu(N) - r +2$, will be an almost unique identifier for each interval, starting with $0$ or $1$
on the left and increasing  to the right. Using this tag, each interval determines which portion of the cost
it will contribute to. The number of intervals assigned to compute a particular term in the cost will depend
on the value of the term since each interval can contribute at most $1$ to the overall cost.
If an interval is assigned to check a {\em yes} guess ($z_k = 1$) the computation uses $x$ and $z$ to determine the 
$k^{th}$ input to the oracle $o_k$, guesses a witness and simulates $V$ on input
$o_k$ with the guessed witness. There is a cost of $+1$ if $V$ rejects and $0$ if
$V$ accepts. If $o_k$ is in fact in $L'$, there is a witness which will allow
for a zero cost tiling withing that interval. If $o_k \not\in L$, then every witness
will lead to a $+1$ cost. Thus, the optimal set of witnesses will result
in the minimum value for $2^{\barn+5} \cost(x,z)$. 
In addition, exactly $2^3 f(x,z)$ of the
intervals will just transition to the rejecting state, incurring a cost of $+1$.
The total cost due to those intervals is $2^3 f(x,z)$. For the remaining intervals, no cost is incurred. 

The cost of a computation fault (illegal pair or square) is a constant that is larger than the cost of ending in a rejecting computation. Therefore, for each independent computation (in Layers $2$ and $4$) the optimal tiling will correspond to a correct computation which may or may not incur a cost for ending in a rejecting state.
Technically, the most challenging part of the proof is to show that the process on
Layer $1$ which creates the intervals is fault-tolerant. The proof for Function Weighted  Tiling  requires stronger conditions than for Gapped Weighted Tiling since we not only have to show that there are a large number of  large intervals 
at the end of Layer $1$ but we
need to establish that the sequence of interval sizes is close to what one would have in a fault-free computation. 
To this end, we use a potential function $A$ which captures how much
a sequence of interval sizes $(s_1, s_2, \ldots, s_m)$ deviates
from the expected sequence $(m+1, m, m-1, \ldots ,2)$. The main part of the proof
is to show that each illegal square or pair can cause the value of $A$ to increase
by at most a constant amount. At the end of Layer $1$, the ideal sequence
of interval sizes is $(\mu(N)+1, \mu(N), \ldots, 2)$. Every interval
size that is missing from the actual sequence of interval sizes has caused
$A$ to increase by at least a fixed amount which in turn corresponds to 
faults incurred in the computation. Thus, we show
that it is more cost-effective to complete
the computation correctly (and not incur the higher cost of a fault) and
incur the smaller potential cost of a rejecting computation.

The most important measure of progress of the tiling/computation in Layer $1$ is the number of times
the encoded Turing Machine completes an iteration of the Outer Loop in which the size of every interval
increases by $1$ and a new interval of size $2$ is added. 
Faults can potentially cause an iteration of the Outer Loop to take 
longer as
they may force the head to shuttle back and forth more times which in turn could result in fewer iterations.
Even in a fault-free computation, the number of steps per iteration increases with
each iteration because there are more intervals.
One of the main lemmas in the analysis is 
Lemma \ref{lem-segLB} which lower bounds the number of times the loop is completed in relation to that number in a fault-free computation. 
The proof is a delicate inductive argument which uses the fact that the increase in the running time
of a loop is not accelerated too much with each additional fault.

{~} 

\noindent{\bf Proof Overview for Parity Weighted Tiling}
\label{sec-introparity}
The proof for parity weighted tiling is very similar to the function problem.
Suppose that a language $L \in \pnexp$ is computed by a Turing Machine $M$ with access to an oracle
for $L' \in \nexp$. Let $M(x,z)$ be the indicator function that is $0$ if $M(x,z)$ accepts and $1$
if $M(x,z)$ rejects. The overall cost computed by the collective computations is:
$4 \cost(x,z) + M(x,z)$.
The left-most interval computes $M(x,z)$ and results in a $+1$ cost in the case that $M$ rejects. The remaining intervals which collectively compute Krentel's cost function all impose
costs of $+2$ or $0$. 
Thus the expression $4 \cost(x,z) + M(x,z)$
will guarantee that the minimum $\cost(c,z)$
corresponds to the correct guess $\bar{z}$. Furthermore, the rightmost bit will be $M(x,z)$ which will cause the minimum cost
to be odd or even, depending on whether $M$ accepts.

\subsection{Discussion, Related Work, and Open Problems}
\label{sec-discussion}
 Despite the fact that the function version  of classical local-Hamiltonians
describes  the task of the computational (classical) physicist much more naturally than decision problems, complexity of function problems was hardly studied even in the
non-TI setting, in the literature of classical theory of computer science. 

Recently, related results were discovered in 
the domain of quantum computational complexity. 
In particular, in \cite{AI21}, Aharonov and Irani use a construction for the function version of (finite)
quantum local Hamiltonian as a component for a hardness
result for the infinite 2D grid. More specifically, they prove that the problem of estimating the ground energy 
of a local Hamiltonian on a finite 2D grid, is hard for 
$\fpnp$. 
Importantly, their results do not 
imply the hardness result presented in this paper, and it seems impossible to extend their proof to deduce the classical hardness result of Theorem \ref{th-function}. 
Like \cite{AI21} we implement Krentel's cost function using a fixed Hamiltonian term,   
but since their construction is quantum (as opposed to the classical construction in this paper),
they are able to prove the result using a completely different set of tools which do not carry over 
to the classical case. In quantum constructions, the lowest energy is an eigenvalue of a general Hermitian matrix and
the matrix can be constructed to fine tune the ground energy to an inverse polynomial precision.
In classical constructions, the total energy will be a sum over terms where each term is chosen from a constant-sized set of values determined by the finite horizontal and vertical tiling rules. This allows far 
less control in the classical setting over the precision of the minimum cost tiling.

Incidentally, note that the results for the quantum case proven in \cite{AI21} are not tight, which follows from the fact that they use a quantum construction to obtain hardness for $\fpnexp$, a classical complexity class. It seems challenging to make the characterization tight in the quantum case. 
In contrast to the class NP, the class $\qma$ is a class of promise-problems
and in simulating a $\pqma$ machine, there is no guarantee that the
queries sent to the $\qma$ oracle will be valid queries. The cost/energy applied for a particular query will depend on the probability that a $\qma$ verifier accepts on the provided input. If the input is invalid, then the probability of acceptance can be arbitrary. Thus, Krentel's cost function will potentially be an  uncontrolled quantity. Typically in a reduction where we
want to embed the output of a function into the value of the minimum energy, the low order bits of the energy are used to encode
the output of the function. It's not clear how to do this without being able to control
the binary representation of the minimum energy.
 Note that by embedding a classical computation in the Hamiltonian, the issue of invalid queries is circumvented. 

Both \cite{AI21} and \cite{CW21} study the complexity of computing the ground energy density of infinite TI Hamiltonians to within
a desired precision 
making use of the technique introduced by Cubitt, Prerez-Garcia, and Wolf
which embeds
{\it finite} Hamiltonian constructions
of exponentially increasing sizes, into the 2D infinite lattice, using Robinson tiles.
Robinson tiling rules \cite{R71}  force an aperiodic structure on the tiling of the infinite plane, with squares of exponentially increasing size. 
The quantum construction of \cite{AI21} layers a  TI 1D Hamiltonian  on top of one of the sides of all the squares. The classical construction of \cite{CW21} layers
a classical finite construction on each  square.
Neither work obtains tight results due to the same issue with invalid queries, although the two papers compromise in completely different ways.
The primary technical innovation introduced in \cite{CW21}  is to devise a more robust version
of Robinson tiles which ensures that the lowest energy state corresponds to a correct
Robinson tiling, even though the cost of the classical finite construction layered on top
may introduce a penalty. If it were possible to obtain an even more robust version of Robinson tiles,
one potentially could layer the finite construction from the current
paper on top the more robust constructions  in the hopes of showing that
computing the ground energy density of a classical TI Hamiltonian in the thermodynamic
limit is complete for $\expnexp$ under Karp reductions.

The results in this paper are also related to the work of Ambainis \cite{Amb2014}
which characterizes the complexity of measuring local observables of ground
states of local Hamiltonians (APX-SIM), showing that the problem is complete
for $P^{\qma[\log n]}$.
 $P^{\qma[\log n]}$ contains those problems
that can be solved by a polynomial time classical Turing Machine with
access to $O(\log n)$ queries to a $\qma$ oracle. 
This type of question (determining a property of the ground state) is similar to our classical 
result about determining whether the cost of the optimal tiling is odd or even. 
The results on APX-SIM \cite{Amb2014, GY19, GPY19} are not hindered by
the issue of invalid queries because the quantity being measured is not the actual energy itself.
Note that the important point here is the property that distinguishes the state
to be measured (minimum energy) is different than the local observable
applied to the measured state. By contrast,  computing the
energy of the lowest energy state appears to be more difficult.
The issue of invalid queries appears to be an obstacle, even when the Hamiltonian terms are position-dependent as in the constructions of \cite{GY19, GPY19}, as well as in the TI constructions  in \cite{AI21,WBG20}.

Finally, it was mentioned earlier that the approximation problem considered here differs from the standard PCP
setting in that the underlying graph is a grid and the terms are TI. 
It remains an open question as to whether there is a family of TI instances of constraint
satisfaction on 
general graphs for which it is hard to estimate the optimal solution to within an additive $\Theta(N)$.

\subsection{Paper Outline}

Section \ref{sec-TM2Tile} describes how the rules of a Turing Machine
are translated into tiling rules.
For all the results presented here, Layer $1$ encodes the execution of a Turing
Machine that creates intervals within the grid that marks off where parallel
computations will take place in subsequent layers.
Since the construction and analysis is common to all the results, we present
that first. Section \ref{sec-L1construction} gives the construction
and Section \ref{sec-L1analysis} proves the main lemmas that are needed 
for the analysis of each construction. 
Section \ref{sec-GWT} then describes the rest of the construction
and proof for Gapped Weighted Tiling
which only requires two additional layers.
Section \ref{sec-PWT} describes the rest of the construction and analysis for
Parity Weighted Tiling. Finally, Section \ref{sec-FWT} describes the modifications
to the construction for Parity Weighted Tiling that is requires for Function Weighted Tiling.

\section{Tilings and Turing Machines}

\subsection{Equivalence of Tiling Variants}
\label{sec-tiling}

We defined the translationally-invariant tiling problem to use two constraints: one is applied to each pair of adjacent tiles in the vertical direction and the other is applied to each pair of adjacent tiles in the horizontal
direction. For convenience, the local constraints described throughout the paper
apply to local configurations in a square of four tiles
as well as to pairs of tiles in the horizontal direction.

We sketch here how one can transform a set of a set of tiling rules that applies to squares into an equivalent one that applies to only pairs.
For each pair of tiles $t_1$ and $t_2$ in the original tiling rules,
create a new combined tile type $[t_1, t_2]$.
A large (but constant) constraint can be added to ensure that a tiling with the new tile types is {\em consistent}, meaning
that a tile of the form $[*,t]$ must go to the left of a tile of the form $[t,*]$.
More specifically, an inconsistent tiling can be transformed into a consistent tiling in a way that strictly reduces the cost.

There is now a one-to-one correspondence between consistent tilings with the new tile types and
tilings with the original tile types.  
Moreover, a vertically aligned pair with tile $[t_a, t_b]$ on top
of $[t_c, t_d]$ can enforce the same constraint as a square in the original tiling rules with tiles
$t_a$ and $t_b$ in the top row of the square and with $t_c$ and $t_d$ in the bottom row of the square.
Constraints on the new tiles can be fixed so that the cost of
corresponding tilings are the same.

\subsection{Encoding Turing Machine Computations in Tilings}
\label{sec-TM2Tile}

In this subsection we describe how the rules of a Turing Machine are
translated into legal and illegal {\em computation} squares so that a tiling that contains
no illegal computation squares corresponds to a correct execution of the Turing Machine.
Note that the illegal and legal computation squares described in this subsection
would apply to one particular Layer of the tiling.
We assume here that the Turning Machine proceeds from the bottom to the top.
If the direction were reversed, the the top and bottom rows of the legal and
illegal computation squares would be swapped.

Consider a Turing Machine $M$ with tape alphabet $\Gamma$ and set of states $Q$.
The set of tiles types for the layer corresponds to $\Gamma \cup (\Gamma \times Q)$.
A tile is called a {\em tape} tile if it is labeled with a symbol from the tape alphabet.
A tile is called a {\em head} tile if it is labeled with a state and tape alphabet symbol, e.g. $(q/c)$.

Each configuration of the Turing Machine is represented by a row of tiles.
In general, we will have Turing Machines whose non-blank tape symbols
are bracketed on the left by the symbol $\leftb$ and on the right by symbol
$\rightb$. The tape contents to the right of $\rightb$ is an infinite sequence of blank ($\#$) symbols. 
The corresponding row of tiles will have $\#$ tiles extending to the 
$\Box$ tile at the right side of the grid.
For example, consider a Turing Machine with state $q$ and whose tape symbols include B, X, $\leftb$, and $\rightb$. 
The Turing Machine configuration

\begin{description}

\item

\begin{tabular}{cccccccccccccccccccc}
 & & & &   &  &  &   &  & &  &  $q$ &  &  & & & & &  & \\
 & & & & &   & &    &  &  &   & $\downarrow$ &  & & & &  &  & & \\
$\Box$ & $\leftb$ & B & B & B & B & X  & B & B & B & X  & B &  X & $\rightb$ &  & & & & & \\
\end{tabular}

\end{description}

will be represented by the row

\vspace{.1in}

\noindent
\begin{tabular}{|c|c|c|c|c|c|c|c|c|c|c|c|c|c|c|c|c|c|c|}
\hline
 $\Box$ & $\leftb$ & B & B & B & B & X  & B & B & B & X  & q/B &  X & $\rightb$ & \# & \# & $\cdots$ & \# & $\Box$  \\
\hline
\end{tabular}

\vspace{.1in}

For ease of notation, we will represent this row of tiles by the string:

$$\Box~\leftb~B~B~B~B~X~B~B~B~X~(q/B)~X ~\rightb \#~\cdots \#~ \Box $$

Suppose the Turing Machine has rule  $\delta(q_0, a) \rightarrow (q_1, b, L)$ then the following four
squares would be legal for any tape symbols $x$ and $y$:

\vspace{.1in}
\begin{tabular}{|c|c|}
\hline
$q_1/x$ & b \\
\hline
x & $q_0/a$ \\
\hline
\end{tabular}~~~~~
\begin{tabular}{|c|c|}
\hline
b & y \\
\hline
$q_0/a$ & y \\
\hline
\end{tabular}~~~~~
\begin{tabular}{|c|c|}
\hline
y & $q_1/x$ \\
\hline
y & x \\
\hline
\end{tabular}~~~~~
\begin{tabular}{|c|c|}
\hline
$\Box$ & $q_1/x$ \\
\hline
$\Box$ & x \\
\hline
\end{tabular}

\vspace{.1in}

Thus if two adjacent rows do not contain
any illegal computation squares, then
the next row up reflects the head location and tape contents after the rule has been applied.

Similarly, if the Turing Machine has rule  $\delta(q_0, a) \rightarrow (q_1, b, R)$ then the following four 
squares would be legal for any tape symbols $x$ and $y$:

\vspace{.1in}
\begin{tabular}{|c|c|}
\hline
b & $q_1/x$  \\
\hline
$q_0/a$ & x \\
\hline
\end{tabular}~~~~~
\begin{tabular}{|c|c|}
\hline
y & b \\
\hline
y & $q_0/a$ \\
\hline
\end{tabular}~~~~~
\begin{tabular}{|c|c|}
\hline
$\Box$ & b \\
\hline
$\Box$ & $q_0/a$ \\
\hline
\end{tabular}~~~~~
\begin{tabular}{|c|c|}
\hline
$q_1/x$ & y \\
\hline
x & y \\
\hline
\end{tabular}

\vspace{.1in}

There can also be Turing Machine rules in which the head does not change location:
$\delta(q_0, a) \rightarrow (q_1, b, -)$.
Then the following three 
squares would be legal for any tape symbol $x$:

\vspace{.1in}
\begin{tabular}{|c|c|}
\hline
$q_1/b$ & x  \\
\hline
$q_0/a$ & x \\
\hline
\end{tabular}~~~~~
\begin{tabular}{|c|c|}
\hline
x & $q_1/b$  \\
\hline
x & $q_0/a$\\
\hline
\end{tabular}~~~~~
\begin{tabular}{|c|c|}
\hline
$\Box$ & $q_1/b$  \\
\hline
$\Box$ & $q_0/a$\\
\hline
\end{tabular}

\vspace{.1in}

In the absence of  a head tile, two vertically neighboring tiles must be the same tile type. Thus for any two tape symbols $x$ and $y$
the following square is legal:

\vspace{.1in}
\begin{tabular}{|c|c|}
\hline
x & y \\
\hline
x & y \\
\hline
\end{tabular}
\vspace{.1in}

All squares that have four interior tiles and whose tile types for the layer are not one
of the legal squares described above is an illegal computation square.

In general, we will design TMs where each state $q$ is reached
from a unique direction. In other words, for each state $q$,
the TM can not have rules corresponding to more the one
of the following three forms:
\begin{enumerate}
\item $\delta(*,*) = (q,*,R)$
\item $\delta(*,*) = (q,*,L)$
\item $\delta(*,*) = (q,*,-)$
\end{enumerate}
Any TM that doesn't have this property can be transformed into a TM that does have this property by adding some additional states and rules. This restriction ensures that for each state $q$, 
and any tape symbols $x$, $y$, and $z$,
only one of the squares below can be legal:

\vspace{.1in}
\begin{tabular}{|c|c|}
\hline
y & $q/x$  \\
\hline
y & x \\
\hline
\end{tabular}~~~
\begin{tabular}{|c|c|}
\hline
$q/x$ & y \\
\hline
x & y \\
\hline
\end{tabular}~~~
\begin{tabular}{|c|c|}
\hline
$q/x$ & y \\
\hline
$q'/z$ & y \\
\hline
\end{tabular}

\vspace{.1in}

Thus the only legal way for a head tile to appear in a row,
is for there to be another head tile just below it or to the immediate left or right
in the preceding row.

\begin{definition}
The number of illegal pairs and squares in Layer $i$, denoted by $F_i$
is the number of illegal pairs, illegal computation squares, and illegal
initialization squares in Layer $i$, plus, if $i > 1$, the number of
illegal translation squares from Layer $i-1$ to Layer $i$.
We will refer to $F_i$ as the {\bf cost} of Layer $i$.
\end{definition}

The cost of the whole tiling will be a linear combination of the $F_i$'s.

A  square is called a 
a {\em head} square if one of the two lower tiles is a head tile
and one of the two upper tiles is a head tile.
Note that a head square can be illegal or legal.
The next two facts about the encoding of Turing Machines
in tiling rules will be useful in analyzing Layer 1 of the
constructions. Based on the  method described above
for translating Turing Machine rules into  legal and illegal computation
squares, the following two sets of facts can be easily verified.

\begin{fact}
\label{lem-twoTiles}
\ifshow {\bf (lem:twoTiles)} \else \fi
Consider a tiling where a tile $t'$ is directly above tile $t$.
If any of the following conditions hold, then both squares containing $t'$ and $t$ are illegal
computation squares:
\begin{enumerate}
    \item $t'$ and $t$ are both tape tiles and $t \neq t'$
    \item $t'$ is a head tile $(q/c')$, $t$ is a tape tile $c$ and $c \neq c'$.
    \item $t'$ is a tape tile $c'$, $t$ is a head tile $(q,c)$ and $\delta(q, c) = (*, c', L/R)$ is not a TM rule.
    \item $t'$ is a head tile $(q'/c')$, $t$ is a head tile $(q/c)$ and $\delta(q, c) = (q', c', -)$ is not a TM rule.
\end{enumerate}
\end{fact}

\begin{fact}
\label{lem-twoTiles2}
\ifshow {\bf (lem:twoTiles2)}  \else \fi
Consider a tiling where a tile $t'$ is directly above tile $t$.
If either of the following two conditions hold, then the vertically aligned pair is in a legal head square or an
illegal computation square
\begin{enumerate}
    \item $t'$ is a head tile $(q/c')$, $t$ is a tape tile $c$ and $c = c'$.
    \item $t'$ is a tape tile $c'$, $t$ is a head tile $(q/c)$ and $\delta(q, c) = (*, c', *)$ is  a valid TM rule.
    \item $t'$ is a head tile $(q'/c')$, $t$ is a head tile $(q/c)$ and $\delta(q, c) = (q', c', -)$ is a TM rule.
\end{enumerate}
\end{fact}

\section{Layer 1 Construction: Creating the Intervals}
\label{sec-L1construction}

\subsection{Overview of Layer 1}
\label{sec-L1TMcomponents}

The role of Layer 1 is to create 
what we call "intervals" whose beginning and ends will mark the regions where separate computations which we want to repeat in the next layer, can occur. 

We will refer to an {\em interval} as a sequence of characters that start with a 
heavy tile from the set $\{X, \barx, \rightb, \leftb\}$
and includes all the $B$'s to the right, up to and including the next heavy tile
so that neighboring intervals
overlap in one location. The {\em size} of the
interval is the number of  characters, including the heavy on either end,
so neighboring intervals overlap by one symbol.
In general intervals will begin and end with $X$ or $\barx$, except that
 the left-most  interval begins with a $\leftb$ on the left  and the right-most interval ends with
$\rightb$ on the right side. 

The idea is to design a computation which ends up at the top row 
of layer 1, such that this row can be viewed as a
concatenation of intervals whose sizes begin with the largest one and decreases by one from one interval to the next.  

This is done by running a TM 
with an Outer Loop containing 
and Inner Loop, as follows. 
In each iteration of the outer loop, the size of each interval increases by $1$ and there is a new interval of size
$2$ added to the right end of the tape. This is done by iterating over an 
Inner Loop, which increases the size of just a single interval (and pushes all intervals to its right one site further).

To describe this we will define the Turing Machine that is simulated in Layer $1$. The transition rules of this TM define a set of
legal squares for Layer $1$ as described in Section \ref{sec-TM2Tile}.

The tape symbols are: $\{ X, B, \barx, \leftb, \rightb, \# \}$.
$B$ is the blank symbol that is written by the Turing Machine. The unwritten
tape symbols are all set to $\#$.
The states can be divided into groups:

\begin{enumerate}

\item $q_{OS}$. OS stands for {\em Outer Start}. This computation is in $q_{OS}$ during a set up phase of the outer loop.

\item The inner loop states are: $q_{IS}$, $q_{left}$, $q_{wX}$, $q_{wB}$, $q_{w\barx}$, $q_{w\rightb}$.

\begin{enumerate} 
\item $q_{IS}$ is the starting state for the inner loop.
\item $q_{left}$ just moves left until $\barx$ is reached.
\item $q_{wt}$ stands for {\em write} character $t$. This is how the contents are moved to the right. The state remembers the tape symbol $t$ that it just wrote over.
\end{enumerate}

\item $q_{e1}$ and $q_{e2}$ are special states for the very end of an iteration of the outer loop.
\end{enumerate}

Figure \ref{fig-OuterLoop} shows the steps of the Outer Loop in pseudo-code.

\begin{figure}[ht]
\noindent\fbox{\begin{minipage}{\textwidth}
\begin{tabbing}
(1) {\sc OuterLoop}:  \\
(2)  ~~~~~\= {\sc Set Up Phase}\\
(3)  \> ~~~~~ \= Sweep left in state $q_{OS}$ until $\leftb$ is reached\\
(4)  \> \>Transition to $q_{IS}$ and move right\\
(5)  \> Start of {\sc Inner Loop}\\
(6)   \> \>  Move right in state $q_{IS}$ until an $X$ is reached\\
(7)   \> \> ~~~~~\= If $\rightb$ is reached before $X$, go to  (14)\\
(8)   \> \> Replace $X$ with $B$ and move right, transition to $q_{w\barx}$\\
(9)   \> \> Insert an $\barx$:\\
(10) \>\>\> Sweep right, moving every symbol to the right ($q_{wt}$ for $t \in \{B, X, \rightb\}$)\\
(11)   \> \> When state $q_{w\rightb}$ is reached, replace $\#$ with $\rightb$, transition to $q_{left}$\\
(12)   \> \> Move left in state $q_{left}$  until $\barx$ is reached\\
(13)  \> \> Replace $\barx$ with $X$, transition to $q_{IS}$ and move right. Go to (5).\\
(14) \> {\sc Wind Down Phase}\\
(15)  \> \> Replace $\rightb$ with $B, X, \rightb$ (states $q_{e1}$ and $q_{e2}$). \\
(16)  \> \> Transition to $q_{OS}$. Go to (1).
\end{tabbing}
\end{minipage}}
\caption{Pseudo-code for an integration of the Outer Loop.}
\label{fig-OuterLoop}
\end{figure}

We now describe the exact transition 
rules of the TM which is simulated in 
layer 1.

\subsection{The Turing Machine for Layer $1$} 

\noindent{\bf The Outer Loop}

\begin{description}

\item {\bf The start of an Outer Loop:}

\begin{tabular}{cccccccccccccccccccc}
& & & & & & & &  &  $q_{OS}$ & & & & & & & & & & \\
& & & & & & & &  &  $\downarrow$ &  & & & & & & & & & \\
$\leftb$ & B & B & B & X & B & B &X  & B &  X & $\rightb$ & & & & & & & & \\
\end{tabular}

\item {\bf The start of  the next Outer  Loop:}

\begin{tabular}{cccccccccccccccccccc}
& & & & & & & &  &  &  & & & & $q_{OS}$ &  & & & \\
& & & & & & & &  &  &  & & & &  $\downarrow$ & & & & \\
$\leftb$ & B & B & B & B & X & B & B & B  &X  & B & B &  X & B & X & $\rightb$ & & &  \\
\end{tabular}

\end{description}

\noindent{\bf Set Up Phase}

Each iteration of the outer loop starts with a set-up phase in which the the head moves to the far left:

\begin{description}

\item

\begin{tabular}{cccccccccccccccccccc}
$q_{OS}$ & & & & & & & &  &  &  & & & & & & & & & \\
$\downarrow$ & & & & & & & &  &  &  & & & & & & & & & \\
$\leftb$ & B & B & B & X & B & B & X  & B &  X & $\rightb$ & & & & & & & & \\
\end{tabular}

\end{description}

This is  achieved by the rule $\delta(q_{OS}, t) = (q_{OS}, t, L)$, for any tape symbol $t \neq \leftb$.

This continues until the head reaches $\leftb$. Then the state transitions to $q_{IS}$ and moves to the right. 
\begin{description}

\item

\begin{tabular}{cccccccccccccccccccc}
 & $q_{IS}$ & & &  & & & &  &  &  & & & & & & & & & \\
 & $\downarrow$ & & &  & & & &  &  &  & & & & & & & & & \\
$\leftb$ & B & B & B & X & B & B & X  & B &  X & $\rightb$ & & & & & & & & \\
\end{tabular}

\end{description}

Rules:  $\delta(q_{OS}, \leftb) = (q_{IS}, \leftb, R)$. This configuration will be the start of an iteration of the inner loop.

\noindent{\bf The Inner Loop}

The state $q_{IS}$ will move right past any $B$'s until it reaches a $X$:

\begin{description}

\item

\begin{tabular}{cccccccccccccccccccc}
 & & & &   $q_{IS}$ &  & & & &  &  & & & & & & & & & \\
 & & & & $\downarrow$ &  & & &  &  &  & & & & & & & & & \\
$\leftb$ & B & B & B & X & B & B & X  & B &  X & $\rightb$ & & & & & & & & \\
\end{tabular}

\end{description}

Rules:  $\delta(q_{IS}, B) = (q_{IS}, B, R)$.

Then it replaces the $X$ with $B$
and inserts a $\barx$ to the right of the $B$.
This has the effect of increasing the size of the current interval.
The $\barx$ symbol tells the head where to return to in the next inner loop iteration.

\begin{description}

\item

\begin{tabular}{cccccccccccccccccccc}
 & & & &   & $q_{w \barx}$ &   & & &  &  & & & & & & & & & \\
 & & & & &  $\downarrow$ &   & &  &  &  & & & & & & & & & \\
$\leftb$ & B & B & B & B & B  & B & X  & B &  X & $\rightb$ & & & & & & & & \\
\end{tabular}

\item

\begin{tabular}{cccccccccccccccccccc}
 & & & &   & &   $q_{wB}$& & &  &  & & & & & & & & & \\
 & & & & &  &  $\downarrow$ & &  &  &  & & & & & & & & & \\
$\leftb$ & B & B & B & B & $\barx$  & B & X  & B &  X & $\rightb$ & & & & & & & & \\
\end{tabular}

\end{description}

Rules:  $\delta(q_{IS}, X) = (q_{w\barx}, B, R)$, and $\delta(q_{w\barx}, t) = (q_{wt}, \barx, R)$,
for $t \in \{B, X, \rightb\}$.

The head moves to the right, moving each character over by one space.
The state remembers, the last symbol that was overwritten. This continues until the head reaches $\rightb$.

\begin{description}

\item

\begin{tabular}{cccccccccccccccccccc}
 & & & &   & &   & & &  &  & $q_{w \rightb}$ &   & & & & & & & \\
 & & & & &  &   & &  &  &   & $\downarrow$ &  & & &  & & & & \\
$\leftb$ & B & B & B & B & $\barx$  & B & B & X & B & X  & \# &  & & & & & & & \\
\end{tabular}

\end{description}

Rules:  $\delta(q_{wb}, t) = (q_{wt}, b, R)$, for $b \in \{B, X\}$ and $t \in \{ B, X, \rightb \}$.

The state $q_{w \rightb}$ writes a $\rightb$, transitions to $q_{left}$ and moves left:

\begin{description}

\item

\begin{tabular}{cccccccccccccccccccc}
 & & & &   & &   & & &  &  $q_{left}$ &  &   & & & & & & & \\
 & & & & &  &   & &  &  & $\downarrow$  &  &  & & &  & & & & \\
$\leftb$ & B & B & B & B & $\barx$  & B & B & X  & B &  X & $\rightb$ &  & & & & & & & \\
\end{tabular}

\end{description}

Rules:  $\delta(q_{w \rightb}, t) = (q_{left}, \rightb , L)$.

 Then the
 head then moves all the way to the left until it reaches $\barx$:

\begin{description}

\item

\begin{tabular}{cccccccccccccccccccc}
 & & & &   & $q_{left}$ &   &  & &  &  &  &  & & & & & & & \\
 & & & & &  $\downarrow$ &   &  &  &  &   &  & & & &  & & & & \\
$\leftb$ & B & B & B & B & $\barx$  & B & B & X  & B &  X & $\rightb$ & & & & & & & & \\
\end{tabular}

\end{description}

Rules:   $\delta(q_{left}, b) = (q_{left}, b, L)$, for $b \in \{X, B\}$.

When $\barx$ is reached, the $\barx$ is replaced with a $X$, the state transitions to $q_{IS}$  and a new inner loop begins.

\begin{description}

\item

\begin{tabular}{cccccccccccccccccccc}
 & & & &   &  &  $q_{IS}$ &  & &  &  &  &  & & & & & & & \\
 & & & & &   &  $\downarrow$ &  &  &  &   &  & & & &  & & & & \\
$\leftb$ & B & B & B & B & X  & B & B & X  & B &  X & $\rightb$  & & & & & & & & \\
\end{tabular}

\end{description}

Rules:  $\delta(q_{left}, \barx) = (q_{IS}, X, R)$.

\vspace{.1in}

At the beginning of each iteration of the inner loop, the location of the head has moved over by one interval.
The intervals to the left of the head have all been increased. The current interval and the ones to the right
have yet to be increased. Working through our example, after the next iteration of the inner loop we have:

\begin{description}

\item

\begin{tabular}{cccccccccccccccccccc}
 & & & &   &  &   &  & &  &  $q_{IS}$ &  &  & & & & & & & \\
 & & & & &   &   &  &  &   & $\downarrow$ &  & & & &  & & & & \\
$\leftb$ & B & B & B & B & X  & B & B & B & X  & B &  X & $\rightb$ & & & & & & & \\
\end{tabular}

\end{description}

Then after the next inner loop:

\begin{description}

\item

\begin{tabular}{cccccccccccccccccccc}
 & & & &   &  &   &  & &  &  &  &  & $q_{IS}$ &  & & & & &\\
 & & & & &   &   &  &  &  &   &  & & $\downarrow$ & & & & &  & \\
$\leftb$ & B & B & B & B & X  & B & B & B & X  & B &  B &  X & $\rightb$  &  & & & & & \\
\end{tabular}

\end{description}

\noindent{\bf Wind Down Phase}

In the last iteration of the inner loop, there is no $X$ in between $q_{IS}$ and $\rightb$. This situation is detected when the TM is in the state $q_{IS}$ and it encounters a $\rightb$ instead of an X.
In this case, the last interval is increased and a new $X$ is added. So $\rightb$ is replaced by $B~ X ~ \rightb$.

\begin{description}

\item

\begin{tabular}{cccccccccccccccccccc}
 & & & &   &  &   &  & &  &  &  &  & $q_{IS}$ &  & & & & &\\
 & & & & &   &   &  &  &  &   &  & & $\downarrow$ & & & & &  & \\
$\leftb$  & B & B & B & B & X  & B & B & B & X  & B &  B &  X & $\rightb$  &  & & & & & \\
\end{tabular}

\end{description}

\begin{description}

\item

\begin{tabular}{cccccccccccccccccccc}
 & & & &   &  &   &  & &  &  &  &  & &  $q_{e1}$ & & & & &\\
 & & & & &   &   &  &  &  &   &  & &  & $\downarrow$& & & &  & \\
$\leftb$ & B & B & B & B & X  & B & B & B & X  & B &  B &  X & B  &  \# & & & & & \\
\end{tabular}

\end{description}

\begin{description}

\item

\begin{tabular}{cccccccccccccccccccc}
 & & & &   &  &   &  & &  &  &  &  & &   & $q_{e2}$ & & & &\\
 & & & & &   &   &  &  &  &   &  & &  & & $\downarrow$ & & &  & \\
$\leftb$ & B & B & B & B & X  & B & B & B & X  & B &  B &  X & B  &  X  & \# & & & & \\
\end{tabular}

\end{description}

\begin{description}

\item

\begin{tabular}{cccccccccccccccccccc}
 & & & &   &  &   &  & &  &  &  &  & & $q_{OS}$  &  & & & &\\
 & & & & &   &   &  &  &  &   &  & &  & $\downarrow$  &   & & &  & \\
$\leftb$ & B & B & B & B & X  & B & B & B & X  & B &  B &  X & B  &  X  & $\rightb$ & & & & \\
\end{tabular}

\end{description}

Rules:  $\delta(q_{IS}, \rightb ) = (q_{e1}, B, R)$, $\delta(q_{e1}, \#) = (q_{e2}, X, R)$, $\delta(q_{e2}, \#) = (q_{OS}, \rightb, L)$.

The following notion will play
an important role in the analysis.

\begin{definition} {\bf End configuration} The configuration $\leftb~\{B, X, \barx\}^*~(q_{e2}/\#)$ is called an
{\em end} configuration
\end{definition} 
An end configuration
is the final configuration in an iteration of the Outer Loop.
Even if there have been some faults in previous steps of the computation,
if the Turing Machine is in an end configuration and computes future
steps without faults, the Turing Machine will begin a new iteration
of the Outer Loop.

Figure \ref{fig-TMrules} summarizes the Turing Machine transition rules. Note that not
every state/input symbol combination will occur in a fault-free computation.
However, we are defining the transition function on a wider set of inputs
so that the Turing Machine can recover from errors due to faults that occurred
earlier 
in the computation.

\begin{figure}[ht]
\centering
\begin{tabular}{|c|c|c|c|c|c|c|}
\hline
 & $\leftb$ & $X$ & $B$ & $\barx$ & $\rightb$ & $\#$ \\
 \hline
 \hline
 $q_{OS}$ & $(q_{IS}, \leftb, R)$ & Left & Left & $(q_{OS}, X, L)$ & Left & * \\
 \hline
 $q_{left}$ & $(q_{IS}, \leftb, R)$ & Left & Left & $(q_{IS}, X, R)$ & Left & * \\
 \hline
 $q_{IS}$ & $(q_{IS}, \leftb, R)$ & $(q_{w \barx}, B, R)$ & Right & Right & $(q_{e1}, B, R)$ & * \\
 \hline
 $q_{w \barx}$ & Right & Insert & Insert & Insert & Insert & * \\
 \hline
 $q_{w X}$ & Right & Insert & Insert & Insert & Insert & * \\
 \hline
 $q_{w B}$ & Right & Insert & Insert & Insert & Insert & * \\
 \hline
 $q_{w \rightb}$ & * & * & * & * & * & $(q_{left}, \rightb, L)$ \\
 \hline
 $q_{e1}$ & * & * & * & * & * & $(q_{e2}, X, R)$ \\
 \hline
 $q_{e2}$ & * & * & * & * & * & $(q_{OS}, \rightb, L)$ \\
 \hline
\end{tabular}
\caption{A summary of the rules for the Layer 1 Turing Machine. The word
{\bf Left} stands for $\delta(q,c) = (q,c,L)$. The word
{\bf Right} stands for $\delta(q,c) = (q,c,R)$. The word
{\bf Insert} stands for $\delta(q_{wt},c) = (q_{wc},t,R)$.
A {\bf *} indicates that the Turing Machine does not have a legal transition
on that state/tape symbol combination.}
\label{fig-TMrules}
\end{figure}


{~}

\noindent{\bf Initialization Rules for Layer $1$}

The rules that constrain the initial
configuration of the Layer $1$ Turing Machine shown in Figure \ref{fig-bottomRowRules}.
A square with $\Box~\Box$ in the lower row and $t_1~t_2$ in Layer $1$ of the 
upper row is legal if and only if there is an edge from a vertex with $t_1$ to
a vertex with $t_2$ in the graph.
If a tile does not appear in the graph, then there is no legal square
in which that tile appears in Layer $1$ directly above a $\Box$ tile.

\begin{figure}[ht]
  \centering
  \includegraphics[width=3.0in]{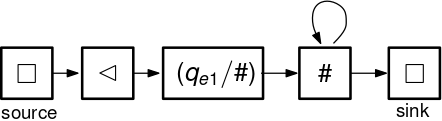}
\caption{These rules constraint the Layer $1$ contents of the bottom
row of the tiling.}
\label{fig-bottomRowRules}
\end{figure}

If a  tiling  does not have an illegal initialization square for Layer $1$,
then
the first row for Layer $1$ must correspond
the to the Turing Machine configuration shown below:

\vspace{.1in}

\begin{tabular}{cccccccccccccccccccc}
 & $q_{e2}$ &  & &   & &   & & &  &  &  &  & & & & & & & \\
 & $\downarrow$ & & & &  &   & &  &  &   &  & & & &  & & & & \\
$\leftb$ & $\#$ & & & &  &   & &  &  &   &  & & & &  & & & & \\
\end{tabular}

\vspace{.1in}

\subsection{Layer $1$: Additional Constraints for Fault Tolerance}
\label{sec-validConfigs}

The TM definition above still allows 
for a single fault to completely halt the computation.
For example, suppose that a row doesn't have a head tile at all.
There would be a local fault where the head tile disappears from one
row to the next. All rows thereafter would just be a replication of the
same tape contents from the row before and would not contain any illegal computation
squares.

To this end we expand the tile types we have so far, so that the tiles $\barx$, $X$ and $B$  come in two different colors: red and blue. Note that
there are now six different tiles corresponding to $\barx$, $X$ and $B$: $\{ \redbarx, \redx, \redb, \bluebarx, \bluex, \blueb\}$.

We can now partition the ordered pairs of Layer $1$ tile types into illegal and legal pairs. If two tiles are an illegal pair for Layer $1$
and those two tiles are adjacent in the horizontal direction in a tiling,
then that pair of tiles will
contribute to the overall cost of the tiling. The set of legal
pairs is best illustrated as a directed graph as show in Figure \ref{figure-ValidConfigGraph}.

\begin{figure}[ht]
  \centering
  \includegraphics[width=\linewidth]{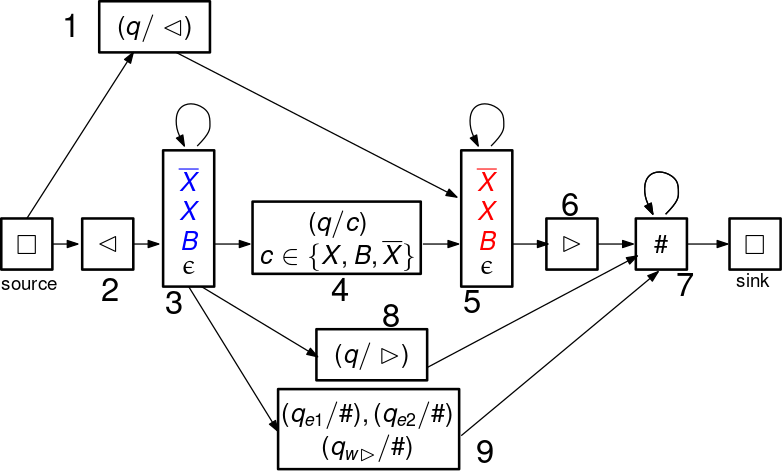}
  \caption{Graphical representation of the set of valid rows. The state $q$ in the diagram
  can be any state that is not in $\{q_{e1}, q_{e2}, q_{w \rightb}\}$}
  \label{figure-ValidConfigGraph}
\end{figure}

\begin{definition}{\bf Legal pairs}
The pair $(t_1, t_2)$ is a legal pair if one
of the following two conditions hold:
\begin{enumerate}
\item There is an edge $(v_1, v_2)$ such that $t_1$ is in $v_1$'s set and $t_2$ is in $v_2$'s set.
    \item There is a vertex $w$ such that $(v_1, w)$ and $(w, v_2)$ are edges, $t_1$ is in $v_1$'s set, $t_2$ is in $v_2$'s set and $\epsilon$ is in $w$'s set.
    \end{enumerate}
All other pairs of tiles are illegal.
\end{definition}

The $\epsilon$ option is included only to reduce the number of edges in the graph.
Suppose instead we added an edge $(v_1, v_2)$ whenever there is a vertex $w$
that contains $\epsilon$ such that $(v_1, w)$ and $(w, v_2)$ are also edges.
Then the $\epsilon$'s could be removed from the graph
and the resulting graph would define the same set of legal pairs.

\begin{definition}
A row of tiles is said to be {\em valid} if no two adjacent tiles
are an illegal pair. Otherwise the row is {\em invalid}.
\end{definition}

The  set of valid rows corresponds to
TM configurations that are a superset of  the configurations
that the TM would reach in any fault-free\footnote{i.e. all rules are obeyed, zero cost} computation. 
We will prove below that
each valid row corresponds to a unique Turing Machine configuration to which a  rule can be applied.
Thus, if a row $r$ is valid, then there is exactly one way
to tile the next highest row without any illegal computation squares
so that the resulting row is also valid.

Since every row must begin and end with a $\Box$ tile and contains no other
$\Box$ tiles, 
a  row of tiles is valid if and only if it can be generated by the following process:
Follow any path
in the graph from the source  to the sink. As each vertex is reached,
select any tile type from the 
current vertex and
place a tile of the type next in the row. $\epsilon$ denotes
the option of not selecting a tile for the current vertex.

{~}

\noindent{\bf Augmenting tiling rules with colors }

{~}

The definition of legal/illegal computation 
squares must be augmented to address the fact that now $\barx$, $X$ and $B$ tiles
come in two different colors. 
In a head square, a tile to the left of a head tile must be blue and a tile to the
right of a head tile must be red.
For example, consider the following computation legal square below:

\vspace{.1in}

\begin{tabular}{|c|c|}
\hline
$(q_{left}/X)$  & $B$ \\
\hline
$X$ & $(q_{left}/B)$ \\
\hline
\end{tabular}

\vspace{.1in}

The $B$ in the upper-right corner must be red 
and the $X$ in the lower-left corner must be blue
in order for the square to be a legal computation square.



It will als be important 
to enforce that the only way for a 
tile to change color from one row to another will be in the
presence of the head. so any square of the form

\vspace{.1in}

\begin{tabular}{|c|c|}
\hline
$B$  & * \\
\hline
$B$ & * \\
\hline
\end{tabular}~~~~~~~
\begin{tabular}{|c|c|}
\hline
$X$  & * \\
\hline
$X$ & * \\
\hline
\end{tabular}~~~~~~~
\begin{tabular}{|c|c|}
\hline
$\barx$  & * \\
\hline
$\barx$ & * \\
\hline
\end{tabular}

\vspace{.1in}

where the two vertically aligned $\barx$, $X$'s or $B$'s have different colors will be an illegal computation square.

To summarize, a square pattern of four Layer $1$ tile types is a legal computation square
if the square is legal according to the rules translating a Turing Machine to legal and illegal
computation squares outlined in Section \ref{sec-TM2Tile} and:
\begin{itemize}
    \item if the square is a head square, then any $\barx/X/B$ tile the left of a head tile
    is blue and any   $\barx/X/B$ tile the left of a head tile
    is red
    \item  if
 a tile $t \in \{X, \barx, B\}$ is directly above a tile $t' \in \{X, \barx, B\}$,
then $t$ and $t'$ have the same color.
\end{itemize}
Otherwise, the square is an illegal computation square.

{~}

\noindent{\bf Properties of Valid Rows }

{~}

The graphical representation of valid rows given in 
Figure \ref{figure-ValidConfigGraph} makes it clear that local constraints can
be used to enforce that a row is valid. In particular,
any row that is not valid must contain
at least one illegal pair. 
However, the lemma below is a more useful description of valid rows which will allow us
to establish that a valid row corresponds to a configuration of the Turing Machine to
which one of the transition rules can be applied.
It will be convenient to refer to the {\em tape contents} of a row of tiles.
This is the row that would result from replacing every tile of the form
$(q/c)$ with $c$. 

\begin{lemma}
\label{lem-validRow}
\ifshow {\bf (lem:validRow)}  \else \fi
A row is valid if and only if the row has the following properties:
\begin{enumerate}
    \item The tape contents of the row has the form:
    $$\leftb~\{X,B,\barx\}^*~\{\rightb, \epsilon\}~\#^*$$
    \item Exactly one tile is a head tile.
    \item Any $\barx/X/B$ tiles to the left of the head tile are blue. Any $\barx/X/B$ tiles to the right of the head tile are red.
    \item If there is a $\rightb$ or $(q/\rightb)$ tile, then the state is not
    in $\{ q_{e1}, q_{e2}, q_{w \rightb}\}$ and the head
    is pointing to one of the tiles from the $\leftb$ to the $\rightb$ (inclusive).
    \item If there is no $\rightb$ or $(q/\rightb)$ tile, then the head is pointing to the leftmost $\#$ tile
    and is in state $q_{e1}$, $q_{e2}$ or $q_{w \rightb}$.
    \end{enumerate}
\end{lemma}

\begin{proof}
We first establish that any valid row must have the five properties outlined in the lemma,
starting with property $1$.
If a row is valid, then it corresponds to a path from the source to the sink in 
the graph in Figure \ref{figure-ValidConfigGraph}.
Note that every path from the source to the sink must first go through vertex $1$ or $2$,
so the tape contents of the first tile must be $\leftb$. There is no path back to
vertices $1$ or $2$, so the first tile is the only $\leftb$ tile.
If there is a $\rightb$ or $(q/\rightb)$ tile, then the path must pass through 
vertices $6$ or $8$. There is no path back to vertices $6$ or $8$,
so $\rightb$ or $(q/\rightb)$ only appear once. 
The only vertices that can come after vertices $6$ or $8$ is vertex $7$,
so only $\#$ tiles can appear after a $\rightb$ or $(q/\rightb)$ tile.
Once vertex $7$ is reached, only vertex $7$ can be reached before the sink.
So after the first $\#$, there can only be $\#$ tiles until the final $\Box$
tile at the sink.

The head tiles are all contained in vertices $1$, $4$, $8$, or $9$.
To establish property $2$, note that every path from the source to the sink passes through
one of the vertices $1$, $4$, $8$, or $9$ 
exactly once.

If vertex $3$ is reached, it occurs before vertices $1$, $4$, $8$, or $9$ in the path.
Therefore only blue  $\barx/X/B$ tiles can be to the left of the head tile. 
If vertex $5$ is reached, it occurs after vertices $1$, $4$, $8$, or $9$ in the path.
Therefore only red   $\barx/X/B$ tiles can be to the right of the head tile. 

If there is no $\rightb$, then the path must pass through vertex $9$.
The state must be
$q_{e1}$, $q_{e2}$,
or $q_{w \rightb}$, and the head points to the leftmost $\#$.
If there is a $\rightb$ symbol, then the path passes through vertices $1$, $4$, or
$8$, in which case the state is not $q_{e1}$, $q_{e2}$,
or $q_{w \rightb}$ and the head points to one of the symbols from the $\leftb$ to the
$\rightb$.

For the converse, the location of the head determines which vertex from 
$1$, $4$, $8$, or $9$ the path goes through. 
Any row in which the head points to the $\leftb$ symbol that also satisfies all five 
properties from the lemma can be generated by a path: 
$1 \rightarrow 5^* \rightarrow 6 \rightarrow 7^* $.
Any row in which the head points to the $\rightb$ symbol that also satisfies all five 
properties from the lemma can be generated by a path: 
$2 \rightarrow 3^* \rightarrow 8 \rightarrow 7^* $.
Any row in which the head points to a $\barx/X/B$ symbol that also satisfies all five 
properties from the lemma can be generated by a path: 
$2 \rightarrow 3^* \rightarrow 4 \rightarrow 5^* \rightarrow 6 \rightarrow 7^* $.
Finally any row in which the head points to a $\#$ symbol that also satisfies all five 
properties from the lemma can be generated by a path: 
$2 \rightarrow 3^* \rightarrow 9 \rightarrow 7^* $.

\end{proof}

\begin{lemma}
\label{lem-nextRow}
\ifshow {\bf (lem:nextRow)}  \else \fi
If $r$ is a valid row, then there is exactly one  row $r'$ that can
be placed above $r$ such that there are no illegal computation squares that span the
two rows $r$ and $r'$. Row $r'$ is valid. Moreover, row $r$ corresponds to a Turing Machine configuration
and $r'$ represents the configuration resulting from executing one step
in configuration $r$.
\end{lemma}

\begin{proof}
Consider a valid row $r$. We will first argue that if there is a  row $r'$
such that if $r'$ is placed directly above $r$,
there are no illegal computation squares, then $r'$ must be unique.

Since $r$  is valid, by Lemma \ref{lem-validRow},
it has exactly one head tile. 
Since each state
in the Turing Machine is reached from a well defined direction, a head tile
in row $r'$ must come from a specific location in $r$. Moreover, each head tile
in $r$ moves to a well-defined location in $r'$. This defines a one-to-one
correspondence between head tiles in row $r$ and head tiles in row $r'$.
For example, if $(q'/c')$ is a head tile in $r'$ and the state $q'$ is reached
from the left, then there must be a head tile $(q/c)$ in row $r$ one location
to the left such that $\delta(q,c) = (q', *, R)$.
Since there is a matching between head tiles in row $r$ and head tiles in row $r'$ 
and there is exactly one head tile in row $r$, then there is exactly one head tile in row $r'$.

Since there are no rules in the Layer $1$ Turing Machine in which the head stays in the same
location, there will be exactly one head square spanning rows $r$ and $r'$ with a head tile
in one of the bottom two tiles and a head tile in one of the top two tiles. The square will
have one of the following two forms, depending on whether the corresponding rule moves the head
right or left:

\vspace{.1in}
\begin{tabular}{|c|c|}
\hline
b & $(q'/y)$  \\
\hline
$q/c$ & y \\
\hline
\end{tabular}~~~~~
\begin{tabular}{|c|c|}
\hline
$(q'/y)$ & b \\
\hline
y & $q/c$  \\
\hline
\end{tabular}
\vspace{.1in}

Note that the contents of the top two tiles are completely determined
by the contents of the lower two tiles and the output of the transition function
on input $(q, c)$. 
If either $y$ or $b$ are a $\barx/X/B$ tiles, then their color is also determined
by the rules for legal/illegal computation squares and the fact that  the head
square is a legal computation square.

All other locations in the row have a tape tile in both  $r$ and $r'$.
If two vertically aligned
tape tiles do not have the same tape symbol and color, then they
must be contained in an illegal computation square.
If tile $b$ in the head square is a $B/X/\barx$ tile and is to the left
of the head then $b$ must be blue. If tile $b$ is a $B/X/\barx$ tile and is to the right
of the head then $b$ must be red.
Therefore if $r$ and $r'$ are both valid and there are no illegal computation squares spanning
rows $r$ and $r'$, then the contents of $r'$ are completely determined.

Next we need to show that if $r$ is valid, then there is a valid $r'$
such that there are no illegal computation squares spanning
rows $r$ and $r'$
Since $r$ is valid, there is exactly one head tile in $r$ and therefore
$r$ corresponds uniquely to a configuration of the Turing Machine.
If the head is pointing
to a $\#$ tile in row $r$, then the state is $q_{e1}$, $q_{e2}$ or $q_{w \rightb}$.
If the head is pointing
to a non-$\#$ tile in row $r$, then the state is not $q_{e1}$, $q_{e2}$ or $q_{w \rightb}$.
Therefore row $r$ corresponds to a configuration of the Turing Machine
and there is a transition rule from Figure \ref{fig-TMrules} that
applies to this configuration.
Let $r'$ be the row resulting from applying one step of the Turing Machine
to the configuration represented by row $r$. 
Color all the $B/X/\barx$ tiles to the left of the head tile blue and
all the $B/X/\barx$ tiles to the right of the head tile red.

We will first establish that there are no illegal squares spanning rows $r$ and $r'$.
The head square must be legal because it represents one
correctly executed step of the Turing Machine.
All other tiles outside of the head square are tape tiles
and are the same in $r$ and $r'$ because they did not change in the computation
step. Moreover since $r$ is valid, all $B/X/\barx$ tiles to the left of the
head square are blue and all $B/X/\barx$ tiles to the right of the
head square are red. Therefore any $B/X/\barx$ tiles outside of the head square
have the same color in $r$ and $r'$.

The final step is argue that the resulting row $r'$ is valid.
Property $3$ is satisfied by construction.  
Row $r'$  also satisfies Property $2$ because it corresponds to a valid Turing
Machine configuration and therefore only has one head tile.
It remains to establish that $r'$ satisfies properties $1$, $4$, and $5$.
According to the Turing Machine rules (shown in Figure \ref{fig-TMrules}) if the head is pointing to a $\leftb$ symbol,
it will write a $\leftb$ symbol and move right into state $q_{IS}$ or $q_{wt}$, where $t \neq \rightb$. The tape contents
remain the same and the head is still in between the $\leftb$ and
$\rightb$ symbols.
If the head is pointing to a $B/X/\barx$ symbol,
it writes a $B/X/\barx$ symbol and moves left or right.
The new state will not be $q_{e1}$, $q_{e2}$, or 
$q_{w \rightb}$.
The tape contents
still have the form $\leftb~\{B, X, \barx \}^*~\rightb~\#^*$
and the head is still in between the $\leftb$ and
$\rightb$ symbols.
If the head is pointing to the $\rightb$
or the leftmost $\#$, the Turing Machine will either
write $\rightb$ and move left into a state that is not in 
$\{q_{e1}, q_{e2}, q_{w \rightb}\}$
or the Turing Machine will write a $B/X/\barx$ symbol
symbol and move right into state $q_{e1}$, $q_{e2}$, or 
$q_{w \rightb}$.
In the former case, the tape contents will be of the form
$\leftb~\{B, X, \barx \}^*~\rightb~\#^*$ and the head is in between the $\leftb$ and
$\rightb$ symbols. In the latter case,
the tape contents will be of the form
$\leftb~\{B, X, \barx \}^*~\#^*$ and the head will point
to the leftmost $\#$ symbol.
\end{proof}




\section{Analysis of Layer $1$: Proving Fault Tolerance}
\label{sec-L1analysis}

The goal of the analysis of Layer $1$ is to show that   as long as the number of faults is not too large ($O(N^{1/4})$) then the end result will approximate the result of a fault free computation reasonably well.

Section \ref{sec-rowcost} associates each illegal pair or square with a particular row and bounds how much a tiling can change from one row to the next as a function of the number illegal configurations associated with that row.
Then in order to show that the computation encoded in the tiling makes progress, even in the presence of faults, we define a notion of a {\em complete segment}
which corresponds to a sequence of rows that represent a complete
and fault-free iteration of the Outer Loop of the Layer $1$ Turing Machine.
In a complete segment, each interval increases in size by $1$ and a new interval
of length $2$ is added to the right end of the non-blank tape symbols.
The main goal of Subsection \ref{sec-segLB} is to prove
Lemma \ref{lem-segLB} which says that the
number of complete segments in Layer $1$ is at least $\mu(N) - O(F)$, where
$F$ is the number of faults and $\mu(N)$ the number of intervals in the last row of a fault-free tiling. 

In general, the number of steps in an iteration of the Outer Loop will depend
on
the number of intervals as well as the total length of all the intervals.
Thus, it's possible for faults to create spurious intervals  which can
cause an iteration of the Outer Loop to take
more steps. 
We  define the {\em weight} of a row to be the number of intervals
and the {\em length} of a row to be the number of non-blank tiles, corresponding to non-blank symbols on the Turing Machine tape. In order to lower bound the number of complete segments, we need to bound the effect of faults on the
weight and length of a row and prove that they do not slow down the computation by too much. These bounds are given in Section \ref{sec-intervals}.

The analysis on the number of complete segments is made more precise in Section \ref{sec-segLB} where we define
the function $X$ which takes as input a sequence of positive integers $(s_1, \ldots, s_m)$ and outputs
the exact number of steps in one iteration of the Outer Loop if the sequence of interval sizes (from left to right)
at the beginning of the iteration  is $(s_1, \ldots, s_m)$.
Section \ref{sec-segLB} gives upper bounds on how much the function $X$ can change in a sequence of correct steps
as well as how much $X$ can change as a result of faults. While the function $X$ can be used to bound the number of rows in a complete segment, it is also important to bound the number of rows outside of complete segments. For example a fault could cause the computation to pop into the middle of an iteration of the Outer Loop by having the head move to a completely arbitrary location. Alternatively, a sequence of correct steps could end in a fault
before the end of the iteration of the Outer Loop is reached. These bounds are put together to prove
Lemma \ref{lem-segLB} which says that the
number of complete segments in Layer $1$ is at least $\mu(N) - O(F)$.

While each complete segment (i.e. iteration of the Outer Loop) results in the creation of a new interval,
it is also necessary to argue that the collection of sizes of the the intervals roughly corresponds
to that in a correct faulty-free tiling. To that end, we introduce a means of identifying and 
tracking the intervals as they grow in size and move to the right. 
Section \ref{sec-clean} also defines the notion of {\em clean} and {\em corrupt} intervals. Intuitively, clean intervals are those that have not been affected by a  fault lower down in the tiling.
The clean intervals are given a unique tag which they keep for the duration of the computation.
The analysis only provides a guarantee for the collection of interval sizes for the clean intervals.
The lower bound on the number of complete segments proven in Section \ref{sec-segLB} is used
to prove a lower bound on the number of clean intervals in Section \ref{sec-clean}.
Suppose that the sizes of the clean intervals, from left to right,
in the last row of Layer
$1$ is $\vec{s} = (s_1, s_2, \ldots, s_m)$.
Lemma \ref{lem-cleanLB} states that $m$, the number of clean intervals, in any tiling for Layer $1$ is at least $\mu(N) - O(F)$.

Section \ref{sec-potential} defines a potential function
which captures how much the sequence $\vec{s}$ differs
from the idealized sequence $(m+1, m, \ldots, 3, 2)$.
Lemma \ref{lem-potential} shows that the value of the potential function
is at most $O(F)$. 
The final lemma required for analyzing the rest of the constructions
is Lemma \ref{lem-analysisL1} which combines
Lemmas \ref{lem-cleanLB} and \ref{lem-potential} and says that if the sequence
of clean intervals at the end of Layer $1$ is $\vec{s}$,
then the number of integers in the range $2, \ldots, \mu(N)+2$
that do not appear as an entry in $\vec{s}$ is bounded by 
$44 F_1 + 3$.

Layer $3$ of the construction for Parity Weighted Tiling and Function Weighted Tiling
uses a Turing Machine that sweeps across all of the non-blank symbols.
In order to establish that this Turing Machine completes its work in $N$ steps, we need
an upper bound on the length of a row as a function of $N$ and the number of faults in Layer $1$. This analysis
is given in Section \ref{sec-lengthUB}.
Note that this section is not used in the proof of Gapped Weighted Tiling.
We present these bounds just after Section \ref{sec-segLB} because they use the definition for the function $X$ which is developed there. 

Finally, in order to compare a tiling with faults to an fault-free tiling,
we will eventually need to characterize the sequence of interval sizes in an fault-free
tiling. This characterization is given in Section \ref{sec-FF}.

\subsection{The Cost of a Row}
\label{sec-rowcost}\ifshow {\bf (sec:rowcost)}  \else \fi

To begin the analysis, it will be convenient to associate each illegal square or pair 
to a particular row in the tiling. This will be important for
accounting for the discrepancy between a faulty tiling
and a fault-free tiling with specific occurrences
of illegal pairs or squares. 
The number of such illegal configurations attributed to a row is the cost of a row.
Note that when analyzing the overall construction, we will refer to the number of illegal pairs and
squares in Layer $i$ by $F_i$. In this subsection, which focuses on Layer $1$, we will drop
the subscript and use $F$ to denote the number of illegal pairs and squares in Layer $1$. 

\begin{definition}{\bf [Cost of a Row]} 
Fix a tiling of the $N \times N$ grid.
The rows of the grid will be numbered from bottom to
top $r_0, \ldots, r_{N-1}$.
Let $h(r_t)$ denote the number of
illegal pairs in row $r_t$. (Note that since $r_0$ is assumed to consist only of $\Box$ tiles,
$h(r_0) = 0$.)
If $t \ge 1$, then $v(r_t)$ is defined to be the number of illegal computation
squares contained in
rows $r_t$ and $r_{t+1}$.
$v(r_0)$ is defined to be the number of illegal initialization squares
(which are by definition contained in rows $r_0$ and $r_1$).
We denote by $F$
the total cost of Layer $1$ of
the tiling, namely the total number of illegal pairs and squares, or the sum of the costs of all rows: $F=\sum_{i=0}^{N-1} h(r_i)+v(r_i)=\sum_i c(r_i)$. 
\end{definition}

The following claim which is used throughout the analysis shows  that 
the number of illegal pairs or squares attributed to a row can be used to bound
the number of locations where the row differs from the row above it.

\begin{claim}
\label{cl-distUB}
\ifshow {\bf (cl:distUB)}  \else \fi
{\bf [Upper Bound on the Distance Between Consecutive Rows]}
If $r_{t-1}$ is an invalid row,
then $d(r_{t-1}, r_{t}) \le
4h(r_{t-1}) + 2v(r_{t-1})$,
where $d(r_{t-1}, r_{t})$ is the number of locations where $r_{t-1}$ and $r_t$ differ.
\end{claim}

\begin{proof}
Suppose that row $r_{t-1}$ is an invalid row
(i.e. $h(r_{t-1})  > 0$).
Let $Q$ be the number of legal head squares that span rows $r_{t-1}$ and $r_t$.
The number of differences between $r_{t-1}$ and $r_t$ inside the legal head
squares is at most $2Q$.

Now consider a pair of  vertically aligned tiles, $t'$ and $t$, 
outside of any legal head square such that $t \neq t'$.
The two tiles must satisfy at least one of the conditions in Facts \ref{lem-twoTiles} and \ref{lem-twoTiles2}
and are therefore contained in at least one illegal computation square.
Since each illegal computation
square contains two pairs of vertically aligned tiles, the number of differences
between $r_{t-1}$ and $r_t$ outside of legal head squares is at most $2 v(r_{t-1})$.

$d(r_{t-1}, r_{t}) \le
2Q + 2v(r_{t-1})$.
We have that the number of head tiles in the row is at least $Q$
(since two legal head squares cannot overlap, and each 
head squares contains a head tile in its bottom row). 
 There is no path in the graph shown 
in Figure \ref{figure-ValidConfigGraph} from a vertex with a
head tile back to another vertex with a head  tile. So 
$h(r_{t-1}) \ge Q-1$. 
Furthermore, since $r_{t-1}$ is not a valid row, $h(r_{t-1}) \ge 1$.
The two bounds on $Q$ together imply that $Q \le 2h(r_{t-1})$. 
Therefore $d(r_{t-1}, r_{t}) \le
4h(r_{t-1}) + 2v(r_{t-1})$.
\end{proof}

\subsection{Intervals and their Properties}
\label{sec-intervals}
\ifshow {\bf (sec:intervals)}  \else \fi

We would like to argue roughly that 
in any tiling that does not contain a large number of faults,  the intervals are organized "more or less" like in the correct 
computation. To be able to make this kind of statement precise, we will need  to expand the definition for intervals in a row
so that intervals are well defined during all points in a valid computation
as well as for invalid rows. 

We start by defining a {\it weight function} on tape symbols,
TM states, and tile types.

\begin{definition} {\bf [Weights of Tape Symbols and TM States]}
If $c$ is a tape symbol, then the weight of $c$, denoted $w(c)$,
equals $0$ if $c = B$ or $c = \#$,  and is $1$ otherwise.
If $q$ is a TM state, then $w(q) = 1$
if $q \in \{ q_{wX}, q_{w \rightb}, q_{w \barx }, q_{e1}, q_{e2}\}$ and $w(q) = 0$ otherwise. 
\end{definition}

Note that the weight $1$ states  are all states that write a weight-$1$ tile, so they
encode the presence of a weight-$1$ tile in the state.
These definitions allow us to assign weights to the tile types.

\begin{definition} {\bf [Tile Weights]}
If $t$ is a tape tile corresponding to tape symbol $c$,
then $w(t) = w(c)$.
If $t$ is a head tile corresponding to $(q/c)$,
then $w(t) = w(c) + w(q)$. If a tile has weight greater than $0$, it is referred to as a {\em heavy} tile.
Otherwise, the tile is called a $0$-weight tile.
\end{definition}

We can now formally describe the intervals in a way that is well-defined for an arbitrary row,
in particular, the definition is precise, even for a row that represents a Turing Machine
configuration in the middle of an execution of the Outer Loop.

\begin{definition} {\bf [Intervals and their Sizes]} 
An {\em interval} is a sequence of more than one tiles in a row,
that begin and end
with heavy tiles and otherwise contain only $0$-weight tiles. A single tile whose weight is $2$ is an interval as well. 
If $I$ is an interval, then $s(I)$ denotes the size of $I$,
which is the number of tiles in $I$.
\end{definition} 
For example, the sequence of two tiles $X~(q_{wX}/B)$ is an interval
of size $2$ and the
single tile $(q_{wX}/X)$ is an interval of size $1$.
The tile $(q_{wX}/X)$ would also be the right end of the interval
extending to the left and the left end of the interval extending to the right. Therefore, a single tile can potentially be contained in three different intervals.
In any row, the sequence of tiles from the leftmost to the rightmost
heavy tile define a sequence of intervals. Consecutive intervals
overlap by one tile. 
See Figure \ref{figure-Intervals} for an example.

\begin{figure}[ht]
  \centering
  \includegraphics[width=\linewidth]{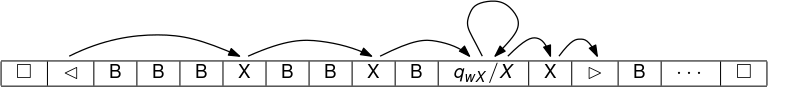}
  \caption{Each interval in the row of tiles is represented by an arrow. The interval begins at the tile
  under the tail of the edge and includes all the  tiles to the right up to and
  including the tile under the head of the edge. The tile
  representing $(q_{wX}/X)$ is part of three intervals.}
  \label{figure-Intervals}
\end{figure}

{~}

The {\em length} of row $r$ is the number of tiles that are not $\#$ or $\Box$ and is denoted by $l(r)$.
The {\em weight} of a  row of tiles $r$ is the sum of the weights
of the tiles in that row and is denoted bey $w(r)$.

\begin{fact}\label{fact:WeightsInts}
The number of intervals in a row $r$ is $w(r)-1$.
\end{fact} 
\begin{proof}
We prove by induction that the number of intervals 
contained in the row up to the $i^{th}$ heavy tile 
is the weight of the tiles up to and including that tile, 
minus $1$. This is true for the first heavy tile, and 
assuming it is true up to the $i^{th}$ tile, 
the next heavy tile adds one interval if its weight is $1$, 
and adds two intervals if its weight is $2$. 
\end{proof}

The next two lemmas upper bound how much the length and the width of the rows change.
The first lemma upper bounds the change over a sequence of rows that do not contain
illegal pairs or squares, corresponding to correct Turing Machine steps.
The second lemma bounds the change from one row to the next in the presence of 
faults. 

\begin{lemma}
\label{lem-valid2}
\ifshow {\bf (lem:valid2)}  \else \fi
{\bf [Change in Length and Weight over Correct Computation Steps]}
Consider a sequence of rows from $r_{s}$ to $r_{t}$
such that the sequence of rows does not contain any illegal pairs or illegal squares. 
If rows $r_{s+1}$ through $r_{t-1}$ are not end rows, then
 $l(r_t) \le l(r_s) + w(r_s) + 1$ and
$w(r_t) \le w(r_s) + 1$.
If $r_s$ is an end row, then $l(r_t) \le l(r_s) + w(r_s)$.
\end{lemma}

\begin{proof}
The only Turing Machine rule that increases the weight of the configuration
is the rule 
$\delta(q_{e1}, \#) = (q_{e2}, X, R)$ which leads to an end row.
Therefore all the rows $r_s$ through $r_{t-1}$ have the same weight and
$w(r_t) \le w(r_s)+1$. 

Each iteration of the inner loop (from an $IS[k]$ configuration to an $IS[j]$ configuration
where $j > k$)
increases the size of one interval by one. So each iteration of the inner loop increases
the total length by $1$. There are at most $m = w(r_s)-1$ iterations of the inner loop.
Then the last step $\delta(q_{e1}, \#) = (q_{e2}, X, R)$ also increase the length by $1$.
If $r_s$ is in state $q_{wt}$ then there can be an additional increase of one
for the symbol $t$ that is inserted. Thus the length will increase by at most
$w(r_s)+1$ and $l(r_t) \le l(r_s) + w(r_s)+1$.
If $r_s$ is an end configuration, the state is $q_{e2}$ and
there is not the additional $+1$ increase
to the length and $l(r_t) \le l(r_s) + w(r_s)$.
\end{proof}

\begin{claim}
\label{cl-invalidBound}
\ifshow {\bf (cl:invalidBound)}  \else \fi
{\bf [Change in Length and Weight in the Presence of Faults]}
Consider two consecutive rows $r_p$ and $r_{p+1}$ in a tiling. Then
\begin{eqnarray*}
l(r_{p+1}) & \le & l(r_p) + 2v(r_p) + h(r_p) + 1\\
w(r_{p+1}) & \le & w(r_p) + 2v(r_p) + h(r_p) + 1\\
\end{eqnarray*}
Moreover, if $r_{p+1}$ is a valid row, then $l(r_{p+1}) \le l(r_p) + \max \{1, v(r_p)\}$.
If $r_{p+1}$ is a valid row that is not an end row, then $w(r_{p+1})  \le  w(r_p) + 2v(r_p)$.
\end{claim}

\begin{proof}
We first account for any increase to the width and length from $r_p$ to $r_{p+1}$
that occurs inside legal head squares.
Consider each legal head square in rows $r_p$ and $r_{p+1}$.
Suppose there are $Q$ such squares.
Each of these squares contains a head tile in the lower
level. Also, the total increase in length or weight  from $r_p$
to $r_{p+1}$ inside these squares is at most $Q$ since
a valid step of the computation will cause the length
or weight to increase by at most $1$.
There is no path in the graph shown 
in Figure \ref{figure-ValidConfigGraph} from a vertex with a
head tile  to another vertex with a head  tile. So 
$h(r_{p}) \ge Q-1$.  Therefore the total increase in
length or weight inside these legal head squares
is at most $h(r_p) + 1$.

If $r_{p+1}$ is valid, then there is only one head tile in $r_{p+1}$
and therefore at most one legal head square. The only TM step that increases
the weight is $\delta(q_{e1},\#) = (q_{e2}, X, R)$. If $r_{p+1}$ is not an
end row (i.e. a valid row that contains $q_{e2}$), then the weight does not increase in
this square.

Now we will show that any increase to the length or
weight from $r_p$ to $r_{p+1}$ that occurs outside
the legal head squares will be at most $2v(r_p)$.
Each illegal square will be given two tokens that can be used to compensate for
any increase to the width or length that occurs inside that square.
Consider a tile $t'$ on top of tile $t$. If any of the conditions from
Fact \ref{lem-twoTiles} apply, then both squares that include 
$t$ and $t'$ are illegal. The width and length can increase by at most $2$.
The increase is paid for by a token from each of the two illegal squares
that contain $t'$ and $t$.

If any of the  cases from Fact \ref{lem-twoTiles2}
apply to $t'$ and $t$, the weight and length can increase by at most $1$.
According to Fact \ref{lem-twoTiles2},  $t'$ and $t$ are in
a legal head square or an illegal  square. In the latter case,
the increase in weight is accounted for by a token from the  illegal square that contains $t'$ and $t$.

The only case not covered by Facts \ref{lem-twoTiles} and \ref{lem-twoTiles2}
is when $t'$ and $t$ are both tape tiles and $t' = t$, in which case the length
and weight does not increase in that location. 
This covers the general upper bounds that apply regardless of whether $r_{p+1}$ is valid
as well as the upper bound on the increase in weight when $r_{p+1}$ is valid.

Finally, to bound the increase in length for the special case where $r_{p+1}$ is valid,
we use again the fact that 
 $r_{p+1}$  has exactly one head tile. Suppose that tile is in
location $j$. The total increase in length from $r_p$ to $r_{p+1}$ at location
$j$ is at most $1$. 
Any other increase in length is due to a tape tile $t$ in row $r_{p+1}$
on top of
a $\#$ tile in row $r_p$. If there are no such occurrences,
then $ l(r_{p+1}) \le l(r_p) + 1$.
If there is at least one such vertically aligned pair, then each one
can increase the length by at most $1$.
Also, by Fact \ref{lem-twoTiles}, each such vertical pair participates in two
illegal squares, one to the left and one to the right. 
Although these squares can overlap, the number of illegal squares is at least
the number
of such vertical pairs plus $1$. In other words, the total increase in the lengths
due to tiles at locations other than $j$ is at most $v(r_p)-1$.
In this case $ l(r_{p+1}) \le l(r_p) + v(r_p)$.
\end{proof}

\subsection{Segments and Complete Segments}
\label{sec-segments}
\ifshow {\bf (sec:segments)}  \else \fi

In a complete iteration of the Outer Loop of the Turing Machine, each interval
increases in size by $1$ and a new interval of size $2$ is added. Thus, we want to show 
a lower bound for the number of fault-free iterations of the Outer Loop
represented in the tiling (Lemma \ref{lem-segLB}).
We partition the rows into segments. The end of a segment is reach
whenever a row has non-zero cost or is a valid end row.
A segment is said to be complete if the last row in the segment as well as the
last row in the previous segment are zero-cost rows. Note that this implies
that the last row in the two consecutive segments are end rows. Since there
are no illegal computation squares contained in the rows of the segment,
the segment corresponds to a sequence of fault-free steps of the Turing Machine
from one end configuration to the next, which is a complete iteration
of the Outer Loop of the Turing Machine.

\begin{definition}{\bf [Segments and Complete Segments]} 
We will partition the sequence of rows a tiling
in to segments going from bottom to top. The first segment
consists of the first two rows $r_0$ and $r_1$. 
Each new segment begins at
the row above the previous segment and
ends at the next row  $r$ that is either
an end row or has $v(r)+h(r)>0$. 
Let $r$ be the last row of a segment and $r'$ be the last row of the previous
segment. A segment is complete if $v(r)+h(r) = v(r')+h(r') = 0$.
\end{definition}

We can now combine the bounds from the previous sections to upper bound the weight and length of a row $r$
as a function of the number of segments and the cumulative cost below row $r$ in the tiling.

\begin{lemma}
\label{lem-phasebounds}
\ifshow {\bf (lem:phasebounds)}  \else \fi
{\bf [Upper Bounds on Length and Weight as a Function of Number of Segments and Number of Faults]}
Let $r$ be a row in segment $t$.  Then 
$$w(r) \le 1 + t + 2  C_t$$
$$l(r) \le 1 + 2 t C_t  + \sum_{j=1}^t j,$$
where $C_t = \sum_{j=0}^{t-1} [v(r_j) + h(r_j)]$.
\end{lemma}

\begin{proof}
Proof by induction on $t$.

{\em Base Case:} $t=1$.  The first segment consists of $r_0$ and $r_1$.
$C_1 = v(r_0)$, which is the number of illegal initialization squares.
We will number the locations of tiles in $r_1$ from left to right so that tiles
$1$ through $N-2$ are the tiles between the $\Box$ tiles.
Square $j$ will be the tiles in locations $j-1$ and $j$ in rows $r_0$ and $r_1$.
For any $S \subseteq \{2, \ldots, N-2\}$ let $v(S)$ be the number of $j \in S$
such that square $j$ is an illegal initialization square.

Each tile in $r_1$ can contribute at most $1$ to the  length and at most $2$ to the
 weight. Let $S_1$ be the locations in $r_1$ that have a tile that is not $\leftb$, $(q_{e2}/\#)$ or $\#$.
For every $j \in S_1$, square $j$ is illegal, so $v(S_1) = |S_1|$.
The total contribution to the weight or length of $r_1$ from tiles that are
in locations in $S_1$ is at most $2 v(S_1)$.
Let $S_2$ be the set of locations with a $\leftb$ tile. 
Square $j$ is illegal for every $j \in S_2$, except $j=1$, 
so $|S_2| \le v(S_2)+1$. Each $\leftb$ tile contributes $+1$ to the weight
and length of $r_1$, so the total contribution to the weight
or length from tiles in $S_2$ is at most $v(S_2)+1$.
If a tile in location $j$ is  $(q_{e2}/\#)$, then either square $j$ is illegal
or location $j-1$ has a $\leftb$ tile.
Let $S_3$ be the set of locations with a $(q_{e2}/\#)$ tile. 
The total contribution to the weight
or length from tiles in $S_3$ is at most $v(S_3) + |S_2| \le v(S_3) + v(S_2) + 1$.
A $\#$ tile does not contribute to the weight or length of $r_1$.
The total weight or length of $r_1$ is at most
$$2 v(S_1) + [v(S_2) + 1] + [v(S_3) + v(S_2) + 1] \le 2 + 2[v(S_1) + v(S_2) + v(S_3)]
= 2 + 2 \cdot C_1.$$

\vspace{.1in}

{\em Induction step:}
We will bound the length and weight of the row at  the end of segment $t > 1$.
Note that if a row $r$ is not the last row in the segment, then it must have zero cost.
The sequence of rows from $r$ to the end of the segment do not contain any illegal pairs
or squares and thus correspond to a set of correctly executed Turing Machine steps
applied to valid row. Note that the length and width of the row do not decrease
with any step of the Turing Machine applied to a valid row. Therefore, it is enough
to bound the length and width of the last row in the segment.

Let $r'$ be the row at the end of previous segment. We consider three cases:

{\bf Case 1:} $v(r') + h(r') = 0$.
Then $C_t = C_{t-1}$.
Rows $r'$ through $r$ do not contain
any illegal squares or pairs. Row $r'$ must be an end row, 
so by Lemma \ref{lem-valid2},
$$w(r) \le 1 + w(r') \le 1 + 1 + (t-1) +  2C_{t-1} \le 1 + t + 2C_t$$  
$$l(r) \le l(r') + w(r') \le \left[ 1 + 2(t-1)C_{t-1} + \sum_{j=1}^{t-1} j \right]  + 
\left[ 1 + (t-1) + 2C_{t-1} \right] = 1 +  2tC_{t} + \sum_{j=1}^t j .$$

{\bf Case 2:} $v(r') + h(r') > 0$ and $r$ is  the row right after $r'$.
Let $\Delta$ denote $C_t - C_{t-1} = v(r') + h(r')$
By Claim \ref{cl-invalidBound},
$$w(r) \le w(r') + 2v(r') + h(r') + 1 \le 1 + (t-1) + 2 C_{t-1}  + 2\Delta + 1
\le 1 + t + 2 C_t$$
\begin{eqnarray*}
l(r) & \le & l(r') + 2v(r') + h(r') + 1\\
&\le & 1 + \sum_{j=1}^{t-1} j + 2(t-1)C_{t-1} +  2 \Delta + 1\\
& \le & 1 + \sum_{j=1}^{t-1} j + 2tC_{t-1} + 2t \Delta + t\\
 & \le  & 1 + \sum_{j=1}^{t} j + 2t C_t\\
\end{eqnarray*}

{\bf Case 3:} $v(r') + h(r') > 0$ and $r$ is not the row right after $r'$.
Then the intervening rows between $r'$ and $r$ are valid.
Let $\bar{r}$ be the row right after $r'$. The row $\bar{r}$ is valid and not
an end row. Otherwise, the segment would have ended at $\bar{r}$ instead of $r$.
By Claim \ref{cl-invalidBound},
$w(\bar{r}) \le w(r') + 2v(r')$. By Lemma \ref{lem-valid2}, 
$w(r) \le 1 + w(r')$. Therefore, 
$$w(r) \le w(\bar{r}) + 1 \le w(r') + 2v(r') + 1 \le 1 + (t-1) + 2C_{t-1} +2v(r') + 1 \le 1 + t + 2C_t$$

By Claim \ref{cl-invalidBound},
and $l(\bar{r}) \le l(r') + \max \{ v(r'), 1\}$.
Let $\Delta = C_t - C_{t-1} = v(r') + h(r')$. Since $\Delta > 0$,
$l(\bar{r}) \le l(r') + \Delta$.
The rows $\bar{r}$ through $r$ do not contain any illegal pairs and are valid. Therefore Lemma \ref{lem-valid2},

\begin{eqnarray*}
l(r) & \le & l(\bar{r}) + w(\bar{r}) + 1\\
&\le & [l(r') +  \Delta] + [ w(r') + 2v(r')] + 1\\
& \le & \left[ 1 + \sum_{j=1}^{t-1} j + 2(t-1)C_{t-1} +   \Delta \right] + \left[ 1 + (t-1) + 2C_{t-1} + 2 \Delta \right] + 1\\
& \le & 1 + \sum_{j=1}^{t} j + 2t C_{t-1} + 2t \Delta = 1 + \sum_{j=1}^{t} j + 2t C_{t}\\
\end{eqnarray*}
The last inequality follows from the facts that $\Delta \ge 1$ and
$t \ge 2$.
\end{proof}

\subsection{Lower Bound on the Number of Segments}
\label{sec-segLB}
\ifshow {\bf (sec:degLB)}  \else \fi

In a fault-free tiling the number of intervals is equal to the number of complete segments
because each iteration of the Outer Loop adds one segment. 
The goal of this section is to show that faults do not change the number of complete
segments in the final row of Layer $1$ by too much.

\begin{definition}{\bf[The function $\mu$]}
Define $\mu(N)$ to be the number of intervals in the last row of Layer $1$
in an fault-free tiling of an $N \times N$ grid.
\end{definition}

The goal in this Section is to prove Lemma \ref{lem-segLB}
which says that the number of complete segments in any tiling in Layer $1$
is at least $\mu(N) - 14 F_1$,  where $F_1$ is the number of faults in the Layer $1$ tiling.
Note that when $F$ is $o(\mu(N))$, which is the case we will be interested in for all the problems we consider, we are characterizing the number of complete segments tightly, up to low order terms.
In order to achieve this tight characterization,  we need to specify the exact number
of rows in a  complete segment which will depend on the sequence of interval lengths.
To this end, we define the following function which characterizes the number
of steps in a complete segment that starts with a  row $r$.
Let $s_1, \ldots, s_m$ be the sizes of all 
the  intervals in row $r$ from left to right. 
Define the function $X$ as:
\begin{equation}
\label{eq-Xdef}
X(r) = \sum_{j=1}^m \left[ 2j ( s_j -1) + 1 \right].
\end{equation}

If the first row $r$ of a segment is a no-cost row, then the segment consists of one or more valid computation steps ending with either a fault or an end row. 
Lemma \ref{lem-valid} provides an upper bound on the number of rows in the segment in this case. The upper bound is the function $X$ applied to the start row plus some additional terms in the cases where the segment begins or ends with a row with non-zero cost. The proof is a somewhat tedious case analysis because if the previous segment ended with a fault, the computation represented in the segment could begin at any point in the execution of the Outer Loop and it is necessary to consider each possible starting point separately. For example, the fault could cause the head to appear in any location and in any state that is consistent with the horizontal constraints.

The next three lemmas (Lemmas \ref{lem-bigSeg2} through \ref{lem-firstBig}) then upper bound the amount by which $X$ can increase in a sequence of rows. Lemma \ref{lem-bigSeg} considers the situation described above where the rows are contained within a segment and first row of the segment is a no-cost row. Lemma \ref{lem-bigSeg} considers a maximal sequence of rows with non-zero cost.  Lemma \ref{lem-firstBig} considers the initial sequence of non-zero cost rows. 

Finally Lemma \ref{lem-segLB} puts all these bounds together to lower bound the number of complete segments. 

\begin{lemma}
\label{lem-valid}
\ifshow {\bf (lem:valid)}  \else \fi
{\bf [Upper Bound on the Length of a Segment]}
Consider a sequence of rows from $r_{s}$ to $r_{t}$
such that the sequence of rows does not contain any illegal pairs or illegal squares. 
Suppose that the sizes of the intervals from left to right (both clean and corrupt) in row $r_s$
are $s_1, \ldots, s_m$.
If rows $r_{s+1}$ through $r_{t-1}$ are not end rows, then
$$t - s \le  l(r_s) + 2w(r_s) - 2 + X(r_s) $$
If $r_s$ is an end row,  then
$t-s \le X(r_s)$.
If $r_s$ and $r_t$ are both end rows, then $t-s = X(r_s)$.
\end{lemma}

\begin{proof}
We will show that if the Turing Machine starts in a valid configuration,
then the Turing Machine will reach an end row within 
$S =  l(r_s) + 2w(r_s) - 2 + \sum_{j=1}^m \left[ 2j(s_j -1) + 1 \right] $ 
correctly executed steps.
 
 Consider a configuration of the Turing Machine that corresponds to a valid row.
The configuration is an
$IS[k]$ configuration if
the state of the Turing Machine is $q_{IS}$ and the head is pointing to one of the internal weight-$0$
symbols in interval $k$ or the weight-$1$ symbol at the right end of interval $k$.

We will establish that the Turing Machine will reach
an end configuration or an $IS[j]$ configuration for $j > k$ within a certain
number of steps.
The head moves right in state $q_{IS}$ until it reaches a $X$ or $\rightb$ symbol.
If it reaches a $\rightb$ symbol first, then it transitions to $q_{e1}$ and then $q_{e2}$
pointing to a $\#$ symbol. This is an end state. The number of steps has been
$$s_k-2 + \sum_{i=k+1}^m (s_i-1) + 2 = 1 + \sum_{j=k}^m (s_j-1) \le S.$$
If the head  reaches an $X$ symbol first in state $q_{IS}$, it writes a $B$, and then inserts an $\barx$ by moving all
the non $\#$ symbols over by one. Interval $k+1$ now has a $\barx$ symbol on its left end.
The number of steps has been at most $\sum_{i=k}^m (s_i-1)$.
When it reaches the $q_{w \rightb}$ state, it replaces the $\#$ with a $\rightb$, transitions to
$q_{left}$ and moves left. The head moves left in state $q_{left}$ until it reaches a 
$\barx$ symbol. 
We are guaranteed that the left end of interval $k+1$ has $\barx$ symbol
but there could be a $\barx$ that comes earlier. Let $j$ be the largest index such that 
the left end of interval $j$ is $\barx$. When the head reaches the left end of interval $j$,
it replaces the $\barx$ with $X$, transitions to $q_{IS}$ and moves right.
This has been an additional $1 + \sum_{i=j}^m (s_i-1)$ steps.
The cycle from the $IS[k]$ configuration to the $IS[j]$ configuration
is an iteration of the Inner Loop.
The total number of steps has been
$$1 + \sum_{i=k}^{m} (s_i - 1) + \sum_{i=j}^{m} (s_i - 1) =
1 + \sum_{i=k}^{j-1} (s_i - 1) + 2 \sum_{i=j}^{m} (s_i - 1) .$$
The number of steps to get to an end configuration from an
$IS[k]$ configuration is maximized if we start
in an $IS[1]$ configuration
and for each $k$, after the Turing Machine leaves the $IS[k]$
configurations, the next time it reaches state $q_{IS}$ is in interval $j = k+1$
(an $IS[k+1]$ configuration). 
After reaching an $IS[m]$  configuration, the Turing Machine
is guaranteed to reach an end configuration in $1 + (s_m - 1)$ steps.
The maximum number of steps required to reach an end configuration when starting in any
$IS[k]$ state is:
\begin{eqnarray}
& & \sum_{k=1}^{m-1} \left[ 1 +  (s_k - 1) + \sum_{i=k+1}^{m} 2 (s_i - 1)  \right] + 1 + (s_m - 1)\\
& = & m + \sum_{i = 1}^m (s_i - 1) + 2 \sum_{j=2}^m \sum_{i = j}^m (s_i-1)
\label{eq-rowBound}
\end{eqnarray}

Next we will argue that as long as the initial configuration of the Turing Machine
is valid, it will reach the state $q_{IS}$ or an end configuration within a certain
number of steps.

If the initial configuration is in state $q_{OS}$ or $q_{left}$,
the head moves left until a $\leftb$ or $\barx$ symbol is reached.
Then the state transitions to $q_{IS}$ and moves right. 
This will be at most an additional $1 + \sum_{i=1}^m (s_i -1)$ steps.
Adding this value to the bound from Expression (\ref{eq-rowBound}):
\begin{equation}
\label{eq-rowBound2}
m + 1 + 2 \sum_{j=1}^m \sum_{i = j}^m (s_i-1) = 1 + \sum_{i=1}^m [2j (s_i-1) + 1]
\end{equation}

If the initial configuration is in state $q_{e1}$
then the current symbol
is $\#$ and the configuration is one  step
away from an end configuration.

If the initial configuration is in state $q_{e2}$ or $q_{w \rightb}$, then
the Turing Machine writes $\rightb$ and moves left into state $q_{OS}$ or $q_{left}$.
The head then moves left until a $\leftb$ or $\barx$ symbol is reached.
Then the state transitions to $q_{IS}$ and moves right. 
This will be at most an additional $1 + \sum_{i=1}^m (s_i -1)$ steps and an $IS[k]$
configuration has been reached.
The same bound from (\ref{eq-rowBound2}) applies.

Finally, suppose that if the Turing Machine is in a state $q_{wt}$. If the current symbol
is $\leftb$, it writes a $\leftb$ symbol, and moves right into state $q_{IS}$.
The current configuration is a $IS[1]$ configuration. In this case,
the number of steps is the bound from Expression (\ref{eq-rowBound}) plus $1$.

If the Turing Machine starts in a state $q_{wt}$
and the current symbol is not $\leftb$, then let $j$ be the index of
the interval that the head starts in.
The Turing Machine will insert a $t$ symbol
and shift the contents of the tape to the right. When the Turing Machine reaches the
$q_{w \rightb}$ state, it writes a $\rightb$, and moves left into state $q_{left}$.
The number of steps has been at most $\sum_{i=1}^m (s_i -1)$.
Interval $j$ has now increased in size by $1$.
The bound from (\ref{eq-rowBound2}) applies, except that the size of 
interval $j$ is now $s_j+1$ instead of $j$. The total number of steps
to reach an end row is at most
$$\sum_{i=1}^m (s_i -1) + 1 + 2j + \sum_{i=1}^m [2j (s_i-1) + 1].$$
The $2j$ term comes from the fact that the size of interval $j$ is $s_j+1$
instead of $s_j$.
Plugging in the expressions $\sum_{i=1}^m (s_i -1) = l(r_s)-1$ and
$j \le m = w(r_s) - 1$:
\begin{eqnarray*}
& & \sum_{i=1}^m (s_i -1) + 1 + 2j + \sum_{i=1}^m [2j (s_i-1) + 1]\\
& \le &
l(r_s) + 2(w(r_s) - 1) + \sum_{i=1}^m [2j (s_i-1) + 1]\\
& \le & l(r_s) + 2(w(r_s) - 1) + m + 2m \sum_{i=1}^m (s_i-1)\\
& \le & l(r_s) + 2(w(r_s) - 1) + w(r_s) - 1 + 2(w(r_s) - 1) (l(r_s)-1)\\
& \le & 2w(r_s)l(r_s)
\end{eqnarray*}
The simplification in the last line of the equations above
use the fact that $l(r_s) \ge w(r_s)$.
The second line gives the first upper bound  for $t-s$ stated in the lemma
and the last line gives the second upper bound for  $t-s$.

If the initial configuration is an end row, then
the Turing Machine writes $\rightb$ and moves left into state $q_{OS}$
until it reaches the $\leftb$ symbol. The head then moves right into
state $q_{IS}$, which is an $IS[1]$ configuration. The number of
steps so far is exactly $1 + \sum_{i=1}^m (s_i -1)$.
The $q_{OS}$ state changes all $\barx$
symbols to $X$ symbols.  The Turing Machine maintains the invariant
that when in state $q_{left}$ or $q_{wt}$, the
$\barx$ symbol will mark the right end point of the
last interval that has been increased. Therefore,
if  an iterations of the inner
loop that starts in an $IS[k]$ configuration, the next iteration will start 
in an $IS[k+1]$ configuration. Therefore the number of remaining
steps to reach an end configuration is exactly the bound from
(\ref{eq-rowBound}). The total number of steps is exactly
the bound given in (\ref{eq-rowBound2}). If the sequence ends before an end row
is reached, then the number of rows is less than the bound given in (\ref{eq-rowBound2}). 
\end{proof}

The following three lemmas will be used
to bound  the growth of the function $X$.
Lemma \ref{lem-bigSeg2} bounds how much $X$ can increase as the result
of fault-free steps of the Turing Machine applied to valid configurations.
Lemmas \ref{lem-bigSeg} and \ref{lem-firstBig} bound how
much illegal pairs or squares can cause $X$ to increase.

\begin{lemma}
\label{lem-bigSeg2}
\ifshow {\bf (lem:bigSeg2)}  \else \fi
{\bf [Upper Bound on Change in $X$ in Fault-Free Steps]}
Consider a sequence of rows $r_a$ through $r_b$ that do not contain
any illegal pairs or squares and such that rows $r_{a+1}$ through $r_{b-1}$
are not end rows. Then
$$X(r_b) - X(r_a) \le 1 + (w(r_b)-1) w(r_b) + 2w(r_b).$$
If $r_a$ is an end row, then $X(r_b) - X(r_a) \le 1 + (w(r_b)-1) w(r_b)$.
\end{lemma}

\begin{proof}
Since rows $r_a$ through $r_b$  do not contain
any illegal pairs or squares, they correspond to $b-a$ steps applied to the valid
configuration represented in  row $r_a$. Moreover, the computation does not
reach the end of an iteration of the Outer Loop until possibly the last step.
Let $m = w(r_a)-1$ be the number of intervals in row $r_a$.
If configuration $r_a$ starts in a state $q_{wt}$, then some interval $j$ may increase
by one as the head sweeps right, moving all the symbols over by $1$. After this point,
each interval can increase by at most $1$ before the end of the Outer Loop is reached.
Thus, the summation in (\ref{eq-Xdef}) can increase by at most $2j + \sum_{j=1}^m 2j$.
If $r_a$ is an end row, then the summation increases by at most $\sum_{j=1}^m 2j$.
If $r_b$ is not an end row, then $w(r_b) = w(r_a)$ and
$$\sum_{j=1}^m 2j = (w(r_b)-1) w(r_b).$$
$$2j + \sum_{j=1}^m 2j \le 2m + (w(r_b)-1) w(r_b) \le (w(r_b)-1) w(r_b) + 2w(r_b).$$
The first expression above bounds $X(r_b) - X(r_a)$ in the case that $r_a$ is an end row
and $r_b$ is not an end row. The second expression bounds $X(r_b) - X(r_a)$ in the case that neither $r_a$ nor $r_b$ are end rows. 

If $r_b$ is an end row, then $w(r_b) = w(r_a)+1$ and a new interval of size $2$ is added
to to the right end, increasing $X$ by an additional $2m+3$.
$$2m+3 + \sum_{j=1}^m 2j = 1 + \sum_{j=1}^{m+1} 2j = 1 + (w(r_b)-1) w(r_b).$$
$$2m+3 + 2j + \sum_{j=1}^m 2j \le 1 + (w(r_b)-1) w(r_b) + 2w(r_b).$$
The first expression above bounds $X(r_b) - X(r_a)$ in the case that $r_a$ 
and $r_b$ are both end rows. The second expression bounds $X(r_b) - X(r_a)$ in the case that $r_a$
is not an end row, but $r_b$ is an end row.
\end{proof}

\begin{lemma}
\label{lem-bigSeg}
\ifshow {\bf (lem:bigSeg)}  \else \fi
{\bf [Upper Bound on Change in $X$ in a Sequence of Rows with Non-Zero Cost]}
Consider a no-cost row $r_{a-1}$ such that row $r_{a}$ has a positive cost.
Let $r_b$ be the next highest no-cost row after $r_a$. 
Let $f$ be the number of illegal pairs or squares contained in the rows
$r_a$ through $r_b$. Let $F$ be the total number of illegal pairs or squares
in Layer $1$. If row $r_b$ is in the $t^{th}$ segment, then
$$X(r_b) - X(r_a) \le 8f[ t^2 + 4tF + t + 2F + 8f].$$
\end{lemma}

\begin{proof}
By Claim \ref{cl-distUB}, $d(r_b, r_a) \le 4 \sum_{j=a}^{b-1}
[v(r_j) + h(r_j)] = 4f$.
We will let $d$ denote $d(r_b, r_a)$.
Note that $r_a$ must be a valid row because the row before it, $r_{a-1}$ is
a no-cost row.
Let $s_1, \ldots, s_m$ be the sizes of the intervals in $r_a$.
We start by re-expressing the summation in the definition of $X(r)$
from (\ref{eq-Xdef}):
$$\sum_{j=1}^m [2j(s_j - 1) + 1] = \sum_{j=1}^m \left[ 1 + \sum_{k=j}^m 
2(s_k -1)\right].$$
We will  account for how the summation in (\ref{eq-Xdef}) changes from $r_a$ to $r_b$.
In the worst case, all of the heavy tiles in row $r_a$ are unchanged
in $r_b$ and all the locations where the two rows differ result
in new heavy tiles in row $r_b$.

Since $r_a$ and $r_b$ are both valid rows, they have exactly one $\rightb$ tile
which is the right end of the right-most interval.
The sum $\sum_{k=j}^m (s_k-1)$ is exactly the distance from the left end
of interval $j$ to the $\rightb$ tile. The position of the $\rightb$ tile
is at most $d$ spaces to the right in $r_b$ than it is in $r_a$, so
the value of the summation in (\ref{eq-Xdef}) increases  from $r_a$ to $r_b$
by at most $2dm$ as a result of the
existing heavy tiles in $r_a$.
There are at most $d$ new heavy tiles in $r_b$ that were not in
$r_a$, each of which  can create at most
two new intervals (since the weight of every tile is at most $2$).
Each new interval can increase the sum by at most $1 + 2d + \sum_{k=1}^m 2(s_k-1)$
which is at most $2 (l(r_a) + d)$.
The increase due to the new intervals is therefore at most
$4d(l(r_a)+d)$. Therefore
$X(r_b) - X(r_a) \le 2dm + 4d(l(r_a)+ d) $.

Since $r_{a}$ has non-zero cost, it is the last row in some segment $t'$, where $t' < t$.
The number of intervals $m$ in row $r_a$ is $w(r_a) - 1$,
which according to Lemma \ref{lem-phasebounds},
is at most $2F+ t' < 2F+t$.
Using the same lemma, we can bound $l(r_a)$ by $2t'F + 1 + t'(t'+1)/2$ which is
at most $2tF + t^2/2$ since $t \ge 2$. 
Plugging the bounds in, and using the fact that $d \le 4f$, we get
$$X(r_b) - X(r_a) \le  8f[ t^2 + 4tF + t + 2F + 8f].$$
\end{proof}

\begin{lemma}
\label{lem-firstBig}
\ifshow {\bf (lem:firstBig)}  \else \fi
{\bf [Upper Bound on Change in $X$ the Initial Sequence of Rows with Non-Zero Cost]}
Let $t$ be the smallest index such that
$r_t$ is a no-cost row.
Let $f$ be the number of illegal pairs and squares contained in
rows $r_0$ through $r_t$.
Then $X(r_t) \le  18f^2 + 18f + 3$.
\end{lemma}

\begin{proof}
By Lemma \ref{lem-phasebounds},
the weight and length of $r_1$ is at most $2 + 2c_1$, where
$c_1$ is the number of illegal initialization squares.
By Claim \ref{cl-invalidBound}, the weight and length can increase from $r_1$
to $r_t$ by at most $2c_2 + (t-1)$, where $c_2$
is the sum of the costs of rows $r_1$ through $r_{t-1}$. 
$c_2 \ge t-1$ because the first $t-1$ rows all have non-zero cost.
Also $c_1 + c_2 = f$. Therefore the weight and length of
$r_t$ is at most
$$[2 + 2c_1] + [2c_2 + (t-1)] \le 2 + 2c_1 + 3c_2 \le 2 + 3f.$$
The summation in the definition of $X(r_t)$
is maximized if the contents of the row consists of $2 + 3f$
heavy tiles, in which case there are $3f + 1$ intervals
all of size $2$. Then $X(r_t)$ is $(3f+2)^2 - 1$
which is at most $9f^2 + 12f + 3$.
The additional cost of $l(r_t) + 2 w(r_t)$ is incurred only if
$r_1$ is not an end row, in which case $f > 0$.
Therefore
$$X(r_t) \le  [9f^2 + 12f + 3] + f(l(r_t) + 2 w(r_t)) \le [9f^2 + 12f + 3] + 3f(2 + 3f)
\le 18f^2 + 18f + 3.$$
\end{proof}

\begin{lemma}
\label{lem-segLB}
\ifshow {\bf (lem:segLB)}  \else \fi
{\bf [Lower Bound on the Number of Complete Segments]}
Let $F$ be the number of illegal squares or pairs in Layer $1$ of a tiling of the
$N \times N$ grid. Then the number of  complete segments in Layer $1$ is at least
$\mu(N) - 14F$. 
\end{lemma}

\begin{proof}
The main work of the proof is to show that the number of segments (complete or not)
in Layer $1$ is at least
$\mu(N) - 12 F_1$. Each segment ends on a row that incurs a cost
or ends on an {\em end row}, corresponding to the
last row of an iteration of the Outer Loop.
Each illegal square or pair can cause the end of at most one segment.
Therefore, if there are at least $\mu(N) - 12 F_1$ segments, then
there must be at least $\mu(N) - 12 F_1 - 2 F_1 = \mu(N) - 14 F_1$ segments
that end with a zero cost row and such that the segment before also
ends on a zero cost row. These are, be definition, complete segments.

Consider a tiling of the $N \times N$ grid that has a total of $F$
illegal pairs and squares.
Number the segments $S_1, S_2, \ldots$ from bottom up.
We will denote the number of row in segment $S_t$ by
$|S_t|$.
Let $\bar{S}_1, \bar{S}_2, \ldots$ be the segments in the
unique tiling in Layer $1$ that has no illegal squares or pairs,
which we call the {\em no-cost} tiling.
There is one interval after the first segment (which just consists of the initial row).
Then in a fault-free tiling every segment corresponds to one iteration of the Outer Loop in which
exactly one interval is added in the last step. 
Therefore in a fault-free tiling the number of complete segments is equal to the
number of intervals in the last row, which is $\mu(N)$.
We will show that 
$$\sum_{t=1}^{\mu(N) - 12F} |S_t| \le \sum_{t=1}^{\mu(N)} | \bar{S}_t|.$$
Thus if $N$ is large enough to accommodate $\mu(N)$ segments in an fault-free tiling, then $N$ is large enough to accommodate $\mu(N) - 12F$
segments in a tiling with $F$ illegal pairs or squares.

A segment will be called {\em big}  if it has more
than one row and otherwise it is called {\em small} .
In a big segment, all the rows except possibly
the last row of the segment have $v(r) + h(r) = 0$.
If a row is valid ($h(r) = 0$) and there are no illegal squares
contained in the row and the row above it ($v(r) = 0$),
then according to Lemma \ref{lem-nextRow},
the row above $r$ is the unique row representing one step of the
Turing Machine applied to the configuration corresponding to $r$, which
is also a valid row. Therefore, if a segment has more than one row,
all the rows in the segment are valid and there are no illegal pairs
or squares contained in the rows of the segment.

Finally, we will partition the illegal pairs and squares in Layer $1$
according to where they are located with respect to the big segments.
If $S_t$ is small, then $f_t = 0$.
Otherwise if $S_t$ is big, let $t'$ be the largest $2 \le t' < t$
such that $S_{t'}$ is also big.
The quantity $f_t$ is defined to be the number of illegal squares contained in the
set of rows from the last row in $S_{t'}$ through the first row in
$S_t$.
If $S_t$ is the first big segment then,
$f_t$ is the number of illegal pairs and squares contained in the rows from
$r_0$ through the first row in $S_t$,
including the illegal initialization squares. Note that $\sum_t f_t = F$.

Note that in any tiling, the first segment consists of rows $r_0$ and
$r_1$, so we will only be concerned with bounding the number of rows
in all the segments, except the first segment.
We will define a quantity $X_t$ for each segment indexed by $t$.
If the $t^{th}$ segment is small, then $X_t = 1$. If the
$t^{th}$ segment is big, then $X_t = X(r)$, where $r$
is the first row in the segment. 
If the last row of the previous segment is a no-cost row, then
it must be an end row. According to Lemma
\ref{lem-valid}, $X_t$ is an upper bound on the number of rows in the
segment.
If the last row of the previous segment has non-zero cost, then $f_t > 0$
and the number of rows of the segment is at most
$X_t + l(r) + 2w(r)$, where $r$ is the first row in the segment.
Therefore
$$\sum_{t=1}^m |S_t| \le \sum_{t=1}^m X_t + \max_r\{l(r) + 2w(r)\} \sum_t{f_t} =
\sum_{t=1}^m X_t + F \cdot \max_r\{l(r) + 2w(r)\}$$
Using the bounds from Lemma \ref{lem-phasebounds} and simplifying,
$l(r) \le 3mF + m^2$ and $w(r) \le 2(m+F)$,
\begin{equation}
\label{eq-Xbound}
\sum_{t=1}^m |S_t| \le 3mF^2 + Fm^2 + 4mF + 4F^2 + \sum_{t=1}^m X_t 
\le Fm^2 + 11 mF^2  + \sum_{t=1}^m X_t 
\end{equation}

We would like to bound how much the parameter $X_t$ can grow from one
big segment to the next. To that end, we define $\Delta_t$
to be the increase in $X_t$. 
If $S_t$ is small, then $\Delta_t = 1$.
Otherwise $\Delta_t = X_t - X_{t'}$, where $t'$ is the largest $2 \le t' < t$
such that $S_{t'}$ is also big.
If $S_t$ is the first big segment then,
$\Delta_t = X_t$. For any $m$:
$$\sum_{t=2}^m X_t \le \sum_{t=2}^m (m-t+1) \Delta_t.$$

The goal then is to bound $\Delta_t$. 
Suppose that $S_t$ is a big segment that is not the first
big segment. Let $S_{t'}$ be the last
big segment before $S_t$.
Let $r$ be the first row of $S_{t'}$, $r'$ be the last row of $S_{t'}$ and
$r''$ the first row of $S_t$.
$\Delta_t = X(r'') - X(r)$.
We will bound $X(r'') - X(r')$ and $X(r') - X(r)$ separately.

We first bound $X(r'') - X(r')$. There are two cases:

{\bf Case 1:}
$r'$ has zero cost. Let $\bar{r}$ be the row after 
$r'$. Since $r'$ is the last row in segment $S_{t'}$, it must be
an end row. Row $\bar{r}$ represents the first configuration
in the next iteration of the Outer Loop. Since the intervals do not change
in the first step of the Outer Loop, $X(\bar{r}) - X(r') = 0$.
If $\bar{r}$ is also a no-cost row, then it is the first
row in a big segment, which means that $\bar{r}$ is the first row in $S_t$
and $\bar{r} = r''$. In this case, $f_t = 0$ and
$X(r'') - X(r') = 0$.
If $\bar{r}$ has positive cost, we can apply Lemma \ref{lem-bigSeg}
with $r_a = \bar{r}$ and $r'' = r_b$ to get that
$$X(r'') - X(r') = X(r'') - X(\bar{r}) \le 8f_t[ t^2 + 4tF + t + 2F + 8f_t].$$

{\bf Case 2:}
$r'$ has positive cost. We can apply Lemma \ref{lem-bigSeg}
with $r_a = r'$ and $r'' = r_b$ to get that
$$X(r'') - X(r') = \le 8f_t[ t^2 + 4tF + t + 2F + 8f_t].$$

We now use Lemma \ref{lem-bigSeg2} to bound $X(r') - X(r)$.
If the row before $r$ has zero cost, then it must be
an end row, and
$X(r') - X(r)  \le 1 + w(r')(w(r')-1)$.
If the row before $r$ has positive cost, then $X(r') - X(r)  \le 1 + w(r')(w(r')-1) + 2w(r')$.
Note that if the row before $r$ has positive
cost, then $F > 0$. Therefore, we can combine the bounds to get:
$$X(r') - X(r) \le 1 + w(r')(w(r')-1) + 2Fw(r')$$
By  
 Claim \ref{cl-invalidBound}, $w(r') \le 1 + t' + 2F \le 2F + t$.
Therefore 
$$X(r') - X(r) \le (2F+t)^2 - (2F+t) + 1 + 2F(2F+t) \le (t^2 -t +1) + 8F^2 + 6Ft$$
The bound for $\Delta_t = X(r'') - X(r)$ comes from
adding the bounds for $X(r'') - X(r')$ 
and for $X(r') - X(r)$, to get

\begin{equation}
\label{eq-Delta}
\Delta_t \le \left[ 8f_t( t^2 + 4tF + t + 2F + 8f_t) \right] + \left[ (t^2 -t +1) + 8F^2 + 6Ft \right].
\end{equation}
The bound for $\Delta_t$ in \ref{eq-Delta} applies to the
case where $S_t$ is a big segment that is not the first segment in Layer $1$.
If $S_t$ is the first big segment, then 
by Lemma \ref{lem-firstBig},
$$X_t = \Delta_t \le (18(f_t)^2 + 18 f_t + 3).$$
If $S_t$ is not a big segment, then $\Delta_t = 1$. In either case, using the fact
that $t \ge 2$, one can verify that
$\Delta_t$ is bounded by the expression given in \ref{eq-Delta}.
Therefore
\begin{eqnarray*}
& &  \sum_{t=2}^m X_t   =   \sum_{t=2}^m (m-t+1) \Delta_t\\
&\le &  \sum_{t=2}^m (m-t+1) 
\left[ 8f_t( t^2 + 4tF + t + 2F + 8f_t) + (t^2 -t +1) + 8F^2 + 6Ft \right]\\
& = & \sum_{t=2}^m (m-t+1)[8f_t t^2 + (t^2 -t +1)]\\
& + &  \sum_{t=2}^m (m-t+1)[ 8f_t( 4tF + t + 2F + 8f_t) + 8F^2 + 6Ft  ]
\label{eq-Xbound2}
\end{eqnarray*}
 By replacing $t$ or $m-t+1$ with $m$ and replacing $f_t$ with $F$,
 the second summation can be upper bounded by
 $72 F^2 m  +  18 F m^2$.
The dominant term in the expression above is from the $8t^2 f_t$
term.
Using the fact that $\sum_t f_t \le F$ and $t \le m$, we can bound 
\begin{eqnarray*}
\sum_{t=2}^m 8 (m-t+1) t^2 f_t & \le &  \sum_{t=2}^m 8 t^2 f_t + \max_t 8F (m-t) t^2\\
& \le & 8 F m^2 + 8(4/27) F m^3 \le 8 F m^2 + 2 F m^3
\end{eqnarray*}
The expression in the max function is maximized for t = 2m/3.
Therefore:
\begin{eqnarray*}
\sum_{t=2}^m X_t \le 2Fm^3 + 72 F^2 m  +  26 F m^2 + 
\sum_{t=2}^m (m-t+1)(t^2 - t + 1)
\end{eqnarray*}
Putting this bound together with the bound from (\ref{eq-Xbound}):
$$\sum_{t=1}^m |S_t| \le 2Fm^3 + 83 m F^2 + 27 Fm^2  + \sum_{t=2}^m (m-t+1)(t^2 - t + 1)$$

Now we turn to the fault-free tiling.
For $t \ge 2$, in the first row of $\bar{S}_t$, there are $t-1$ intervals
of sizes $t, t-1, \ldots, 2$. Therefore $X_2 = 3$
and $\Delta_t = 1 + \sum_{j=1}^{t-1} 2j = t^2 - t + 1$.
Letting $m = \mu(N)-12F$,
we want to show that
$$ 2Fm^3 + 83 m F^2 + 27 Fm^2  + 
\sum_{t=2}^m (m-t+1)(t^2 + t + 1)  \le \sum_{t=2}^{m+12F} (m+12F-t+1)(t^2 - t + 1)$$
We will lower bound the difference of the two summations:
\begin{eqnarray*}
& & \sum_{t=2}^{m+12F} (m+12F-t+1)(t^2 - t + 1) - \sum_{t=2}^m (m-t+1)(t^2 - t + 1)\\
& \ge &
\sum_{t=2}^m 12F(t^2 - t + 1) + \sum_{j=1}^{12F} (12F-j+1)[(m+j)^2 - (m+j) + 1]\\
& \ge &
12F\left( \frac{m^3}{3}  - \frac{m}{3} \right)
+ [(m+1)^2 - (m+1) + 1]\sum_{j=1}^{12F} (12F-j+1)\\
& \ge &  4Fm^3 - 4Fm + 72 F^2 (m^2 + m + 1)\\
& \ge & 2Fm^3 + 83 m F^2 + 27 Fm^2 
\end{eqnarray*}

\end{proof}

\subsection{Upper Bound on the Length of a Row}
\label{sec-lengthUB}
\ifshow {\bf (sec:lengthUB)}  \else \fi

Layer $3$ of the construction for Parity Weighted Tiling and Function Weighted Tiling
uses a Turing Machine that sweeps across all of the non-blank symbols.
In order to establish that this Turing Machine completes its work in $N$ steps, we need
an upper bound on the length of a row as a function of $N$ and the number of faults in Layer $1$. In this section we present the upper bound on the length of a row required for 
this analysis. Note that this section is not used in the proof of Gapped Weighted Tiling.
We present these bounds here because they use the definition for the function $X$ which is developed
in the previous section.

\begin{lemma}
\label{lem-XdiffLB}
\ifshow {\bf (lem:XdiffLB)}  \else \fi
{\bf [Lower Bound on the Change in $X$]}
Consider any two consecutive segments which end at rows $r_a$ and $r_b$, respectively.
Then $X(r_b) - X(r_a) \ge -2 l(r_a)$.
If $r_a$ and $r_b$ are no-cost rows (meaning the segment ending in row $r_b$ is a complete segment, then
$X(r_b) - X(r_a) \ge w(r_a)^2$.
\end{lemma}

\begin{proof}
In a complete segment, the size of each segment increases by $1$ and there is a new segment of size $2$ added to the right end. Therefore the sequence $(s_1, s_2, \ldots, s_m)$ becomes $(s_1+1, \ldots, s_m +1, 2)$. Therefore
the value of $X$ goes from
$$\sum_{j=1}^m [ 2j (s_j-1) + 1 ] ~~~~\mbox{to}~~~~\sum_{j=1}^m [ 2j s_j + 1 ] + 2(m+1) + 1.$$
The overall increase is $\sum_{j=1}^m 2j + 2m + 3 \ge (m+1)^2$. The length of row $r_a$ is equal to $m+1$,
so the value of $X$ increase by at least $w(r_a)^2$.

Consider a sequence of of rows that do not contain any invalid squares. This sequence of rows represent correctly
executed steps of the Turing Machine applied to a valid row. Each iteration of the inner loop causes one of the intervals to increase in size and the other intervals to remain the same which can only increase $X$. The only point at which the value
of $X$ can decrease is when the head is moving right as it inserts a new blank symbol into an interval. 
As the rest of the tape symbols are moved to the right, one interval which holds the current location of the head can be
temporarily decreased in size by $1$. This can decrease the value of $X$ by at most $2m$. By the time the head reaches the right end of the tape symbols, the sizes of all the intervals are at least their original size.
This can contribute a decrease of at most $2(w(r_a)-1)$. If the last row in the previous segment has non-zero cost,
There may be tiles that change from one row to the next  outside of a valid computation square.
The worst case is if a tile of weight $2$ changes to a tile of weight $0$. This would cause three consecutive intervals of
sizes $s_j, 1, s_{j+2}$ to be merged into one interval of size $s_j + s_{j+2}+1$. The net effect is that $X$ can increase
by at most four times the length of row $r_a$. 
Therefore, $X(r_b) - X(r_a) \ge -2 (l(r_a) + w(r_a) - 1)$.
\end{proof}

\begin{lemma}
\label{lem-numSegUB}
\ifshow {\bf (lem:numSegUB)}  \else \fi
{\bf [Upper Bound on the Number of Segments]}
If row $r$ is the last row in a tiling in Layer $1$ and $F_1$ is the number of
illegal pairs or squares in Layer $1$ of the tiling, then if $F_1 \le N^{1/4}/40$,
the number of segments in Layer $1$ is at most $4N^{1/4} + 2 F_1$.
\end{lemma}

\begin{proof}
Consider a complete segment. The weight of the row at the end of the segment is exactly one larger
than the weight of the last row of the previous segment.
At the end of the $t^{th}$ segment, there have been $t - 2F_1$ such increases overall.
Furthermore, none of the valid Turing Machine rules decrease the weight of a row, so in any valid computation square the weight of the two top tiles is at least the weight of the two bottom squares. The difference in weight between two vertically aligned tiles is at most two. The number of vertically aligned tiles that are not the same and  are outside of valid computation squares is at most $F_1$. Therefore
the weight of any row after the first $t$ segments 
is at least $\max\{0,t - 4 F_1\}$.

Let $r$ be the row just before the beginning of a complete segment.
Then the value of $X$ increases by at least $w(r)^2$.
This means that by the end of the $t^{th}$ segment, the total increases that occurred during
complete segments is at least
$$\sum_{j=1}^{t - 2F_1} (t-4 F_1)^2 \ge \frac{(t - 6F)}{3}$$
In each incomplete segment, the value of $X$ can decrease by at most $2(l+w+1)$, where $l$ and $w$ are
an upper bounds on the length and weight of a row that has occurred so far. By Lemma \ref{lem-phasebounds},
$l \le 1 + 2t F_1 + t^2/2$ and $w \le 1 + t + 2F_1$.
Therefore the total decrease which has occurred so far in all the incomplete segments is 
at most $4F_1(2 + 2tF_1 + t + 2F_1 + t^2)$.
Therefore, by the end of the $t^{th}$ segment, the value of $X$ is at least
$$\frac{(t - 6F)^3}{3} - 4F(2 + 2tF_1 + t + 2F_1 + t^2)$$.
For $t \ge 3N^{1/4}$ and $F \le N^{1/4}/40$, then this function is at least $.15 t^3$.

If a complete segment begins in row $r$, then by Lemma \ref{lem-valid}, the length of the segment is exactly
$X(r)$. If the total number of segments is larger than $4N^{1/4} + 2 F_1$,
then the total number of complete segments that happen after the first $3 N^{1/4}$ segments is at least
$N^{1/4}$.
Each of these segments lasts for at least $.15 (3N^{1/4})^3 \ge 4 N^{3/4}$ rows. This contradicts the fact that 
there are at most $N$ rows.
\end{proof}

\begin{lemma}
\label{lem-lengthUB}
\ifshow {\bf (lem:lengthUB)}  \else \fi
{\bf [Upper Bound on the Length of a Row]}
If row $r$ is the last row in a tiling in Layer $1$ and $F_1$ is the number of
illegal pairs or squares in Layer $1$ of the tiling and $F_1 \le N^{1/4}/40$, then
$l(r) \le  9 N^{1/2} + 2 N^{1/4}+1$. 
\end{lemma}

\begin{proof}
Lemma \ref{lem-phasebounds} gives an upper bound on the length of a row $r$
at the end of the $t^{th}$ segment as a function of $F_1$ and $t$.
Plugging in the bound from Lemma \ref{lem-numSegUB} on $t$, we get that
\begin{eqnarray*}
l(r) & \le & 2tF_1 + 1 + \sum_{j=1}^{t}j\\
& \le & 1 + 2(4N^{1/4} + 2 F_1)F_1  + \frac{(4N^{1/4} + 2 F_1)(4N^{1/4} + 2 F_1 + 1)}{2}
\end{eqnarray*}
When $F_1 \le N^{1/4}/40$, the expression can lower bounded by $9 N^{1/2} + 2 N^{1/4}+1$.
\end{proof}

\subsection{Clean and Corrupt Intervals}
\label{sec-clean}
\ifshow {\bf (sec:clean)}  \else \fi

In order to track the sizes of the intervals, 
we will need a somewhat sophisticated 
definition, which will recursively designate
each interval in row $r_t$ {\it clean} or {\it corrupt}. 
Intuitively, a clean interval in row $r_t$ is an interval that is unaffected
by an illegal pair or square in rows $r_0$ through $r_{t-1}$.
 Each clean interval in a row is given a tag. The clean
intervals within a row all have unique tags. The tags allow us to track the intervals
while the intervals shift and grow as the computation
progresses (extending upwards in the tiling). Lemma \ref{lem-cleanLB} then uses the lower bound on the number of complete segments from Section \ref{sec-segLB} (Lemma \ref{lem-segLB}) to lower bound the number of clean intervals in the last row of Layer $1$. Tagging the clean intervals will be useful so that we can argue more precisely about how the sizes of those intervals grow.

We are only be able to characterize the sizes of the clean intervals as faults can cause the sizes and number of an intervals to change in unpredictable ways. Consider for example an interval of size $s$ and a fault which causes a $B$ in the middle of that interval to change to an $X$ in the following row. The interval of size $s$ splits into two different intervals whose sizes sum to $s+1$.


To write down the definition of {\it clean} and {\it corrupt} intervals, we 
need a notion of a 
distance between sequences of tiles: 

\begin{definition}{\bf [Row Distance on a Sequence of Locations]}
The distance between two tiling rows $r$ and $r'$ is the
number of locations where $r$ and $r'$ differ and is 
denoted by $d(r, r')$.
If $S$ is a sequence of tile locations in a row of
length $N$, and $r$ and $r'$ are 
two tiling rows, then $d_S(r,r')$ is the number of location
in $S$ where $r$ and $r'$ differ.
\end{definition} 

We will also need the following notation: 
For any valid row $r$, let $\mbox{next}(r)$ denote the unique row
that can be placed above $r$ with no illegal squares. According to
Lemma \ref{lem-nextRow}, $\mbox{next}(r)$
represents the application of one step of the Turing Machine to row $r$. 

Consider two consecutive rows in the tiling $r$ and $r'$. If $r$ is valid, then $\mbox{next}(r)$ is well-defined and we want to assess whether a clean interval in row $r$ remains clean in $r'$ according to whether it is equal to the corresponding interval in $\mbox{next}(r)$. On the other hand if $r$ is not valid, then the row will not necessarily correspond to a valid Turing Machine configuration and there is not a well-defined expectation for what $r'$ should be. In this case, we just compare $r'$ to $r$ in determining whether each clean interval in $r$ remains clean in $r'$.

\begin{definition}
\label{def-cleancorrupt}
{\bf [Clean and Corrupt Intervals and Tags]}
$r_0$ is assumed to consist of all
$\Box$ tiles, so the definition of {\em clean}
or {\em corrupt} starts with $r_1$. 
For $t = 1, \ldots N-1$ the intervals in row $r_t$ are designated as clean or
corrupt, and the clean intervals will be tagged by positive integer numbers (which can be viewed as "names" of the clean intervals). This is done recursively, as follows.

First, a {\it comparison row} $\tilde{r}$ 
is determined,  the intervals
in $\tilde{r}$ are designated as clean or corrupt, and
the clean intervals in $\tilde{r}$ are tagged. This is then used to decide 
about the clean and corrupt intervals of 
$r_t$ as well as the tags for the clean intervals in $r_t$. We consider three
cases:
\begin{enumerate} 
\item {\bf Case 1:} $t = 1$.  
The comparison
row $\tilde{r}$ is set to be the correct starting row: 
$$\Box~\leftb~(q_{e2}/\#)~\{ \# \}^{N-4} \Box$$
Since this row is correct, we already know 
the assignment of its clean and corrupt 
intervals: 
The only interval in $\tilde{r}$ is $\leftb~(q_{e2}/\#)$, which is
designated as clean and assigned the tag $1$.

\item {\bf Case 2:} $t > 1$ and $r_{t-1}$ is invalid.  Then
$\tilde{r} = r_{t-1}$. The intervals in $r_{t-1}$ have recursively been designated as clean or corrupt.
Moreover, the clean intervals have been assigned tags.

\item \label{case3}{\bf Case 3:} $t > 1$ and $r_{t-1}$ is valid.  Then
$\tilde{r} = \mbox{next}(r_{t-1})$. Order the intervals
in $r_{t-1}$ and $\mbox{next}(r_{t-1})$ from left to right. 
The $j^{th}$ interval in $\mbox{next}(r_{t-1})$ is given the same clean/corrupt designation 
as the $j^{th}$ interval in $r_{t-1}$. If the $j^{th}$ interval is clean, then the tag is also the same.
Claim \ref{cl-numIntsNext} below shows that since $r_{t-1}$ is valid, the number of intervals in $r_t$ is either the same or one larger than that of $r_{t-1}$.
If $\mbox{next}(r_{t-1})$ has one more interval than $r_{t-1}$, then the new interval is the rightmost interval.
The new interval  is designated as clean and given the tag $t$.
\end{enumerate} 

We now use the designation of clean 
or corrupt intervals in $\tilde{r}$ to 
determine which intervals are clean and 
corrupt in $r_t$. 
For each interval in $r_t$, let $S$ be the sequence of tiles in that interval.
If $d_S(r_t, \tilde{r}) > 0$, then the interval is corrupt in $r_t$.
If $d_S(r_t, \tilde{r}) = 0$, then the interval  adopts
the same clean/corrupt designation as in $\tilde{r}$.
If the interval is clean then it adopts the same tag as the corresponding
interval in $\tilde{r}$.

\end{definition}

The following Claim shows that 
Case \ref{case3} of the above definition is indeed well defined.

\begin{claim}
\label{cl-numIntsNext}
\ifshow {\bf (cl:numIntsNext)}  \else \fi
{\bf [Change in the Number of Intervals in One TM Step]}
If $r$ is a valid row, then $r$ and $\mbox{next}(r)$
have the same number of intervals, except if $r$
has the form $\Box~\leftb \{B, X, \barx\}^* (q_{e1}/\#) \{ \# \}^*~\Box$
in which case $\mbox{next}(r)$ has a new interval. The new interval is the right-most interval in $\mbox{next}(r)$
and has  size $2$.
\end{claim}
\begin{proof}
By Fact \ref{fact:WeightsInts}, we know that the number of intervals is one less than the weight of the row. 
The only TM rule that increases the total weight of the tiles (namely, the only legal head square whose weight of top two tiles is bigger than that of its bottom two tiles) is $\delta(q_{e1}, \#) = (q_{e2}, X, R )$.
By inspection of Figure \ref{figure-ValidConfigGraph}, the only valid row which has the tile $(q_{e1}/\#)$ as its head tile, is of the form $\Box~\leftb \{B, X, \barx\}^* (q_{e1}/\#) \{ \# \}^*~\Box$. 
Therefore if any other rule is applied, the number of
intervals is the same between the two rows. Note that $(q_{e1}/\#)$ is a heavy
tile so $(q_{e1}/\#)$ is the right end of the right-most
interval in $r$. When the rule is applied, the two tiles
$(q_{e1}/\#), \#$ in $r$ become $X, (q_{e2}/\#)$ in $\mbox{next}(r)$. Since $X$ and $(q_{e2}/\#)$ are both heavy tiles, the interval to the left of the $X$ does not change
and the new interval has size $2$.
\end{proof}

In order to use the tags in Definition \ref{def-cleancorrupt} to track
intervals, we need the following lemma about the tags:

\begin{lemma}
\label{lem-tags}
\ifshow {\bf (lem:tags)}  \else \fi
{\bf [Interval Tags are Unique]}
Within a row $r_t$, all the tags of clean intervals are unique and $\le t$.
If a clean interval in $r_t$ has a tag $j < t$, then there is a clean interval with
tag $j$ in row $r_{t-1}$. 
\end{lemma}

\begin{proof}
By induction on $t$. There is only one possible clean interval in 
row $r_1$. If row $r_1$ has that interval, then its tag is $1$.

For the inductive step, suppose that $r_{t-1}$ is invalid. Then
in Definition \ref{def-cleancorrupt}, $\tilde{r} = r_{t-1}$. 
The only way for an interval to be clean in $r_t$ is for that interval
to be identical to a clean interval in $r_{t-1}$ in which case
the interval has the same tag in $r_t$ as it does in $r_{t-1}$.
If the intervals all have unique tags $\le t-1$ in $r_{t-1}$ then
they also have unique tags $\le t$ in $r_t$.

If $r_{t-1}$ is valid, then $\tilde{r} = \mbox{next}(r_{t-1})$. 
There is a one-to-one correspondence (that preserves tags) 
between the clean intervals in 
$\mbox{next}(r_{t-1})$ and the intervals in $r_{t-1}$, except for the
fact that $\mbox{next}(r_{t-1})$ might contain an additional new interval
with tag $t$. If all the intervals in $r_{t-1}$ are unique and $\le t-1$,
then all the intervals in $\mbox{next}(r_{t-1})$ are unique and $\le t$. 
The only way for an interval to be clean in $r_t$ is for that interval
to be identical to a clean interval in $\mbox{next}(r_{t-1})$. If the tag of
a clean interval in $r_t$ is $\le t-1$, then it corresponds to a clean
interval in $\mbox{next}(r_{t-1})$ whose tag is $\le t-1$,
which in turn corresponds to a clean interval in $r_{t-1}$.
\end{proof}

\begin{lemma}
\label{lem-cleanLost}
\ifshow {\bf (lem:cleanLost)}  \else \fi
{\bf [Upper Bound on the Loss of Clean Intervals]}
The number of tags $j$ such that there is a clean interval with tag $j$ in
$r_{t}$ but there is no clean interval with tag $j$ in  $r_{t+1}$ is at most $12h(r_t) + 6v(r_t)$.
\end{lemma}

\begin{proof}
By Lemma \ref{lem-nextRow}, if $h(r_t) + v(r_t) = 0$, then
$r_{t+1}$ is
the same as $\mbox{next}(r_t)$, which means that $r_{t+1}$ represents a correctly executed step of the Turing Machine on the configuration represented in the valid row $r_t$, and row $r_{t+1}$ is also valid. Thus no clean intervals are lost from $r_t$ to $r_{t+1}$,
although it's possible that a clean interval with tag $t+1$ is added in row $r_{t+1}$.

If $h(r_t) > 0$, then row $r_t$ is not valid.
By Definition \ref{def-cleancorrupt},
$r_{t+1}$ is compared to $r_t$ in determining which intervals in $r_{t+1}$ are
clean. Each clean interval in $r_t$ that is unchanged in $r_{t+1}$ remains clean
in $r_{t+1}$. Therefore if a clean interval with tag $j$ that is present in $r_t$
does not appear in $r_{t+1}$, there must be a tile in the interval where $r_t$ and
$r_{t+1}$ differ.
Each tile in $r_t$ is included in at most $3$ intervals, so the number of
intervals  that are clean in $r_t$ that are no longer clean intervals in $r_{t+1}$
is at most $3d(r_t, r_{t+1})$, which by 
Claim \ref{cl-distUB} is at most
$12h(r_t) + 6v(r_t)$.

Now suppose that $r_t$ is valid ($h(r_t)=0$) but $v(r_t) >0$.
$r_{t+1}$ is compared to $\next(r_t)$ in determining which 
intervals are clean or corrupt. Every clean interval in $r_t$
corresponds to a clean interval in $\next(r_t)$ with the same tag.
Every clean interval in $\next(r_t)$
that fails to appear in $r_{t+1}$ must contain a location where
$\next(r_t)$ and $r_{t+1}$ differ. 
Since each tile participates in at most $3$ intervals in $\next(r_t)$,
the number of clean intervals in $r_t$ that are no longer clean intervals in $r_{t+1}$ is
at most $3 \cdot d(\next(r_t), r_{t+1})$.
We will now argue that $d(\next(r_t), r_{t+1}) \le 2(v_t)$.

$r_t$ and $\next(r_t)$ only differ
in two locations. The other locations are all tape tiles. If one of these tape tiles
in $r_t$ is not the same in $r_{t+1}$, then those two vertically aligned tiles 
are in an illegal computation square. Now consider the legal head square resulting
from putting $\next(r_t)$ on top of $r_t$. If either tile in the top row of this
square differs from the two corresponding tiles in $r_{t+1}$ then the head
square is illegal. Therefore in each location where $\next(r_t)$ and $r_{t+1}$
differ, the two vertically aligned tiles in $r_t$ and $r_{t+1}$ in that location
are contained in an illegal computation square. Since each illegal computation square
contains at most two pairs of vertically aligned tiles, the number of 
locations where $\next(r_t)$ and $r_{t+1}$
differ is at most $2(v_t)$. 
\end{proof}

\begin{lemma}
\label{lem-cleanLB}
\ifshow {\bf (lem:cleanLB)}  \else \fi
{\bf [Lower Bound on the Number of Clean Intervals]}
Let $F_1$ be the number of illegal squares or pairs in Layer $1$ of a tiling of the
$N \times N$ grid. Then the number of clean intervals in the last row of Layer $1$ is at least
$\mu(N) - 26F_1$. 
\end{lemma}

\begin{proof}
By Lemma \ref{lem-segLB}, the number of complete segments in Layer $1$ 
is at least $\mu(N) - 14 F_1$.
Since each
complete segment corresponds to an fault-free iteration of the Outer Loop and since one new clean interval
is added in each iteration of the Outer Loop,  the number of
clean intervals that are created in Layer $1$ is at least $\mu(N) - 14 F_1$.
Each of these clean intervals is tagged with the row number in which it
first appears, which is the last row in that segment. Therefore, each
of the $\mu(N) - 14 F_1$ clean intervals created has a unique tag.
By Lemma \ref{lem-cleanLost}, the number of tags such there is a clean interval
in some row $r_{t-1}$ but no clean interval with tag $j$ in row $r_t$
is at most $12F_1$.
Therefore the number of clean intervals in the last row of Layer $1$
is at least $\mu(N) - 14 F_1 - 12 F_1 = \mu(N) - 26 F_1$.
\end{proof}

\subsection{The Potential Function $A$}
\label{sec-potential}
\ifshow {\bf (sec:potential)}  \else \fi

In addition to giving a lower bound on the number of clean intervals in
a tiling, we will want to show that the sizes of those intervals 
do not deviate too much from the sizes of the intervals in a tiling with no illegal
pairs or squares.
In an fault-free computation, if the number of intervals at the end of an iteration of the Outer Loop is $m$,
then the sizes of those intervals
from left to right will be $m+1, m, m-1, \ldots, 2$.
The function $A$ captures the extent to which the interval sizes differ from this ideal case. 

\begin{definition}
\label{def-A}
Consider a sequence of $m$ positive integers $s_1, \ldots, s_m$.
$$A(s_1, \ldots, s_m) = \sum_{j=1}^{m-1} |s_j
- s_{j+1} - 1| + |s_m - 2|. $$
\end{definition}

Note that $A(m+1, m, \ldots, 2) = 0$.
It will also be convenient to refer to the value of $A$ for a row in Layer $1$ of a tiling.
If the sizes of the intervals in row $r$ (from left to right)
is $s_1, \ldots, s_m$, then $A(r) = A(s_1, \ldots, s_m)$. If $r$ does not have any clean intervals, then $A(r)$ is
defined to be $0$.

The primary goal of the analysis in this section is to prove
Lemma \ref{lem-potential}, which says that if $r$
is the last row of Layer $1$ in a tiling, and $F_1$ is the number of 
illegal pairs or squares in Later $1$, then $A$ is bounded by $3 + O(F_1)$.

\subsubsection{Features of the function $A$}

The analysis of each layer will bound the value of $A$ by a constant times
the number of illegal pairs and squares.
At various points in the analysis, we will need to deduce features about the sequence
of sizes of the clean intervals based on the bound for the value on the value of $A$.
The following three technical lemmas describe the required features and argue that they follow
from the bound on $A$.

\begin{lemma}
\label{lem-seq}
\ifshow {\bf (lem:seq)}  \else \fi
Consider a sequence $s_1, \ldots, s_m$ 
such that 
each $s_i \ge 2$.
Let $S$ denote the set of integers occurring in the sequence
$s_1, \ldots, s_m$ with repetitions removed. 
Then 
$$ \left| \{2, \ldots, r \} - S \right|  \le
\max\{r-m-1,0\} + A(s_1, \ldots, s_m).$$ 
\end{lemma}

\begin{proof}
The value of $A(s_1, \ldots, s_m)$ is minimized if 
the $s_i$'s are sorted in decreasing order, so we will assume
this is the case. Let $S_m$ be the set of integers in the range
from $2$ through $s_m$.
Define $N_{skip} = |S_m - S|$ and  $N_{dup} = m - |S|$.
If the sequence $(s_1,\ldots, s_m)$ is sorted in decreasing order then value of $A(s_1, \ldots, s_m)$ is equal to $N_{skip} + N_{dup}$.
Furthermore $s_m \ge m+1 - N_{dups}$. 
If $r  \le s_m$, then
$$\left| \{2, \ldots, r \} - S \right| \le N_{skip}
\le A(s_1, \ldots, s_m)$$
If $r > s_m$, then
$$\left| \{2, \ldots, r \} - S \right| = r - s_m
+ N_{skip} \le r - m - 1 + N_{dup}  + N_{skip}
\le (r-m-1) + A(s_1, \ldots, s_m)$$
\end{proof}

\begin{lemma}
\label{lem-usingA}
\ifshow {\bf (lem:usingA)}  \else \fi
Consider a sequence $s_1, \ldots, s_m$ whose $A$ value is at most $d$.
Then  at least
$(m - d)/2$ of the $s_j$'s  are at least
at least $(m - d)/2$.
\end{lemma}

\begin{proof}
We will prove that if there less than $(m - d)/2$ of the $s_j$'s are at least
at least $(m - d)/2$, then the value of $A$ is strictly larger than $d$.
The value of $A$ is minimized when the sequence $s_1, \ldots, s_m$ is sorted in non-increasing order,
so we will also assume that the $s_j$'s are sorted accordingly.

If there are less than $(m - d)/2$ of the $s_j$'s in the sequence that are at least
at least $(m - d)/2$, then there are more than $m - (m - d)/2 = (m+d)/2$ of the $s_j$'s that fall in the range from $1$ to $\lfloor (m - d)/2 \rfloor$.
This implies that  there are more than $d$ of the $s_j$'s such that $s_j = s_{j-1}$.  Each such value contributes at least $+1$ to the sum defining $A$. Since all the terms in $A$ are non-negative, the lemma follows.
\end{proof}

\begin{lemma}
\label{lem-remove1}
\ifshow {\bf (lem:remove1)}  \else \fi
Consider a sequence $s_1, \ldots, s_m$. If $s_j = 1$ and $s_j$ is removed from the
sequence then the value of $A$ for the sequence decreases.
If any $s_j$ is removed from the sequence, the value of $A$ increase by at  most $1$.
\end{lemma}

\begin{proof} If $s_1 = 1$ then the first term in the sum defining $A$ is $|1 - s_2 - 1| = s+2$. Since $s_2$ is positive, the value of $A$ increases when $s_1$ is removed.
If $s_m = 1$, the last two terms in the sum defining $A$ are $|s_{m-1} - 1 + 1| + |2 - 1| = s_{m-1}+1$.
If $s_m$ is removed, these two terms are placed by $|s_{m-1}-2|$, which for positive $s_{m-1}$ is less than $s_{m-1}+1$.

Finally, we consider the case that $j$ is neither, $1$ nor $m$. When $s_j = 1$ is removed, the two terms
$|1 - s_{j+1} - 1| + |s_{j-1} - 1 - 1| = s_{j+1} + |s_{j-1}-2|$ are replaced by 
$|s_{j-1} - s_{j+1} - 1|$. If $s_{j-1} = 1$, then the original expression is $s_{j+1}+1$
which is replaced by $s_{j+1}$ and the sum decreases.

If $s_{j-1} > 1$, then $s_{j+1} + |s_{j-1}-2| = s_{j+1} + s_{j-1}-2$.
Meanwhile the new expression $|s_{j-1} - s_{j+1} - 1|$ is maximized when $s_{j+1} < s_j$
in which case $|s_{j-1} - s_{j+1} - 1| < s_{j+1}+1 - s_{j-1}$. We have that
$$s_{j+1} + |s_{j-1}-2| = s_{j+1} + (s_{j-1}-2) \ge s_{j+1} - (s_{j-1} - 2) > s_{j+1} + 1 - s_{j-1} \ge |s_{j-1} - s_{j+1} - 1|.$$
\end{proof}

\subsubsection{Bound on the value of $A$ in a tiling}

The analysis for Layer $1$ must bound the value of $A$ in the last row of Layer $1$
as a function of the number of illegal squares and pairs in Layer $1$. We start by analyzing 
a sequence of rows with no illegal pairs.
Lemma \ref{lem-noCostRows} gives an upper bound on the amount by which $A$ can increase over a sequence of fault-free steps (corresponding to no-cost rows). Then Lemma \ref{lem-potential} gives an upper bound on $A(r)$ as a function of the number of faults in all the rows below $r$  in the tiling. 

\begin{lemma}
\label{lem-noCostRows}
\ifshow {\bf (lem:noCostRows)}  \else \fi
{\bf [Upper Bound on the Increase in A Over a Sequence of No-Cost Rows]}
Consider a sequence of consecutive rows in a tiling of the $N \times N$ grid, beginning  with
row $r_{s}$ and ending on row $r_{t}$.
Suppose that the tiling from 
$r_{s}$ through $r_{t}$  does not contain any illegal pairs or
 squares. Then
$A(r_{t}) \le A(r_{s}) + 6$. If $r_{s}$ is an end row then $A(r_{t}) \le A(r_{s}) + 3$.
\end{lemma}

\begin{proof}
Each row proceeding upwards from $r_{s}$ corresponds to one step in the computation
of the Turing machine. 
All the intervals (clean and corrupt) will be numbered from left to right, and
we will track the changes to the interval sizes by $\delta$ functions. So a $\delta(j) = +1$
would indicate that the $j^{th}$  interval as increased in size by $1$. Note that a change to interval $j$ will only
affect the value of $A$ if $j$ is a clean interval.
We will use $m$ to denote the number of clean intervals and $M$ to denote the total number of intervals.
Note that the number of intervals does not change in one iteration of the Outer Loop, except in the
last  step. 

At the start of an iteration of the Inner Loop, the state is $q_{IS}$.
 We will call a row an $\inner[k]$ row, if the state is $q_{IS}$ and the head is in interval $k$.
 Since the intervals overlap, the head could possibly be in two intervals at the same
 time, such as $(q_{IS}/X)$. In this case, we say that the head is in the interval to the left.
 In analyzing the sizes of the intervals, it doesn't matter where in interval $k$ that the head begins.
 When the TM is in state $q_{IS}$, the head moves to the right past any $B$ symbols until
 it reaches an $X$ or a $\rightb$ tile at which point it changes state. Thus, in any 
 consecutive sequence of rows in which the state is $q_{IS}$,
 the head remains in the same interval and the interval sizes do not change.
 
 We will consider different subsequences of rows corresponding to different portions of 
an iteration of the Outer Loop.

\begin{enumerate}
    \item {\bf The subsequence starts with an $\inner[k]$ row and ends with an $\inner[k+l]$ row.} 
    After the $l-1$ iterations of the inner loop,
    the intervals $k$ through $k+l-1$ will have increased in size by $1$ and all other interval sizes will stay the same:
    $\delta(k) = \delta(k+1) = \cdots \delta(k+l-1) = +1$.
    \item {\bf The subsequence starts  with an $\inner[k]$ row and ends before an $\inner[k+1]$ row is reached.}
    If the head never reaches the right end of interval $k$, the intervals do not change.
    If the head sweeps past the right end of interval $k$, then  the size of interval $k$ increases by $1$.  ($\delta(k) = +1$.)
    If the sequence ends while the head is in interval $j$ before reaching the $\rightb$ tile, then the left end of interval $j$ has been moved over but the right end has not been moved over.
    The effect is that interval $j$ has decreased in size  by $1$: $\delta(j) = -1$. If the subsequence ends after the head has reached the $\rightb$ tile, the state will be
    $q_{w \leftb}$ or $q_{left}$. In this case,  all the tiles have been moved over and there is no $\delta(j) = -1$ change.
    \item {\bf The subsequence starts in a row that is not an $\inner$ row and ends when the first $\inner[k]$ is reached.}
        If the initial state in the subsequence is $q_{w \leftb}$ or $q_{left}$, then the head just moves left until it reaches
    a $\barx$ or $\leftb$ symbol, at which point the state transitions to $q_{IS}$.
    The intervals do not change. If the head is in the middle of sweeping right,
    then the state is $q_{w \barx}$, $q_{wX}$ or $q_{wB}$. Suppose that the head starts in  interval $i$. The size of interval $i$
    is increased as the head sweeps right. When the head has finished sweeping to the right, all the tiles have been moved over and the net effect is 
   $\delta(i) = +1$.
    \item {\bf The subsequence starts in the state $q_{OS}$ and stops on or before $\inner[1]$.} The head sweeps left until reaching $\leftb$. The intervals do not change,
    so all $\delta$'s are $0$. 
    \item {\bf The subsequence starts in an $\inner[M]$ row and ends on or before the next end row.} 
    The head moves right until it reaches the $\rightb$ symbol.
    The tile $(q_{IS}/\rightb)~\#$ become $B~ (q_{e1}/\#)$ which means that $\delta(M) = +1$.
    Then $(q_{e1}/\#)~\#$ become $X (q_{e2}/\#)$ which means a new interval of size $2$ has been added. The last step in which
    $X~(q_e2/\#)$ changes to $(q_{OS}/X)~\rightb$ does not change any of the intervals.
\end{enumerate}

Let $c_1, \ldots, c_m$ be the indices of the clean intervals in $r_{s}$.
Let $s_j$ denote the size of interval $c_j$. The value of $A$ is:
$$A = \sum_{i=1}^{m-1} | s_i - s_{i+1} - 1| + |s_m - 2|.$$

We first consider the case in which $r_{s}$ is a end row. 
If the sequence ends before the first Inner Loop begins, the net effect on $A$ is $0$ because none
of the intervals change in size. 
Otherwise, the sequence can then go on to $k$ complete iterations of the inner loop which causes
 the first $k$ intervals to increase in size by $1$. If $a$ is the number of clean
intervals among the first $k$ intervals, then 
$\delta(c_1) = \delta(c_2) = \cdots \delta(c_a) = +1$, which can increase $A$ by at most $1$
because only the $|s_a - s_{a+1} - 1|$ term changes and that term can only change by $1$.
If the sequence goes on to an incomplete iteration of the Inner Loop, the next clean interval could have increased
($\delta(c_{a+1}) = +1)$ and the change to $A$ is still at most $1$. It's also possible that the sequence ends
while the head is sweeping right, in which case $\delta(c_j) = -1$ for some clean interval $c_j$. This can increase
$A$ by an additional $+2$ for a total increase of $+3$.
If the sequence completes all the Inner Loops, then all the clean intervals will have increased in size.
In this case, none of the terms inside the sum in the expression for $A$ change. However the last
term $|s_m-2|$ will increase by $1$.
Finally, if the sequence completes the entire Outer Loop, an additional interval of size $2$ is added to the end,
so there is a new $s_{m+1} = 2$.
The net effect on $A$ from the beginning of the sequence is $0$. 
$$A(r_{t}) - A(r_{s}) =  (|(s_m + 1) - s_{m+1} -1|  + |s_{m+1} - 2|) - |s_m - 2|= 0.$$
Thus if $r_{s}$ is a end row, and the sequence ends before the next end row, $A$ can increase by at most $3$.
Moreover if the sequence is one complete iteration of the outer loop
(i.e. if $r_s$ and $r_t$ are end rows), then $A(r_{t}) - A(r_{s}) = 0$.

Now suppose that the sequence begins in an arbitrary location in the middle of the middle of the Outer Loop.
The worst case, is that the sequence starts during a right sweep in an inner loop. There is possibly a $\delta(i) = +1$
if the head starts out in a clean interval which is increased as the head sweeps right. ($A$ can increase by at most $+2$.)
Then after a one or more complete iterations of the Inner Loop, a consecutive sequence of clean intervals
could have increased in size by $1$. ($A$ can increase again by at most $+2$.) Then the sequence can end 
in the middle of an iteration of the inner loop in which case, there could also be a $\delta(j) = -1$ for the interval
where the head ends up. (Another increase to $A$ of at most $+2$.) The total increase to $A$ is bounded by $6$.
Thus, if the sequence does not contain any end rows (i.e. the sequence is contained within one iteration of the
Outer Loop), then the increase to $A$ is at most $6$.

If the sequence begins at an arbitrary location in the middle of the Outer Loop and completes the Outer Loop (i.e. ends on an 
end row), then there could be a $\delta(i) = +1$
if the head starts out in a clean interval which will be increased as the head sweeps right. ($A$ can increase by at most $+2$.)
In addition, the next iterations of the Inner Loop will start with some interval $k$ and increase all the intervals to the right 
of $k$, this will result in $\delta(c_a) = \delta(c_{a+1}) = \cdots = \delta(c_m) = +1$ and finally, a new clean
interval  of size $2$ will be added to the right end. These changed can increase $A$ by at most $1$ because only
the $|s_{a-1} - s_{a} - 1|$ term increases by $1$. The effect of increasing the last clean interval by $1$ and adding a new
interval of size $2$ cancels out as argued above. Thus if the sequence starts in the middle of an Outer Loop and
runs to the end of the Outer Loop, $A$ increases by at most $3$.

We have argued that if the sequence from $r_{s}$ through $r_{t}$
does not contain any end rows (i.e. the sequence is contained within one iteration of the
Outer Loop), then the increase to $A$ is at most $6$.
If sequence does contain an end row, let $r_a$ be the first end row and $r_b$ be the last end row of the sequence.
Since going from end row to end  row, does not change the value of $A$, $A(r_b) - A(r_a) = 0$.
A sequence that is contained in an Outer Loop and ends on a end row increases $A$ by at most $3$, so
$A(r_a) - A(r_{s}) \le 3$. A sequence that is contained in an Outer Loop and begins with  an end row increases $A$ by at most $3$, so
$A(r_{t}) - A(r_{b}) \le 3$. Putting the inequalities together gives that $A(r_{t}) - A(r_{s}) \le 3$.
\end{proof}

\begin{lemma}
\label{lem-potential}
\ifshow {\bf (lem:potential)}  \else \fi
For any $t \ge 1$,
$$A(r_t) \le 3 + 12[h(r_{t-1}) + v(r_{t-1})] +   \sum_{j=2}^{t-2} 18[h(r_j) +  v(r_j)].$$
\end{lemma}

\begin{proof}
By induction on $t$. Base case: $r=1$.
The only possible clean interval in row $r_1$ is 
$\leftb (q_{e2}/\#)$. If $r_1$ has this clean interval,
then $A(r_1) = |s_1 - 2| = 0$. If $r_1$ does not have
any clean intervals, then $A(r_1) = 0$ by definition.

Inductive step: Suppose the intervals in $r_t$ are determined by comparison
to row $\tilde{r}$. Every clean interval in $r_t$ corresponds to a clean 
interval in $\tilde{r}$. The only way $A$ can change from $\tilde{r}$ to $r_t$
is if a clean interval is removed.
Each location where $r_t$ and $\tilde{r}$ differ can remove at most three clean intervals
from $\tilde{r}$ to row $r_t$. Each removed clean interval can increase $A$ by at most $1$. 
If $s_j$, $s_{j+1}$, and $s_{j+2}$ are the sizes of three consecutive clean intervals in $\tilde{r}$ and the
interval of size $s_{j+1}$ is not present in $r_t$ then:
$$|s_{j+2}-s_j-1| \le |s_{j+2} - s_{j+1} - 1| + |s_{j+1} - s_j| \le |s_{j+2} - s_{j+1} - 1| + |s_{j+1} - s_j - 1| + 1.$$
Therefore $A(r_t) - A(\tilde{r}) \le 3 d(\tilde{r}, r_t)$.

{\bf Case 1:} $r_{t-1}$ is invalid. Then $\tilde{r} = r_{t-1}$.
$$A(r_t) - A(r_{t-1}) =
A(r_t) - A(\tilde{r}) \le 3 d(\tilde{r}, r_t)
\le 6 v(r_{t-1}) + 12h(r_{t-1})$$
The last inequality is due to by Claim \ref{cl-distUB}.

{\bf Case 2:} $r_{t-1}$ is valid and there is an illegal square in rows $r_{t-1}$ and $r_t$. Since $r_{t-1}$ is valid,
$\tilde{r} = \next(r_{t-1})$.
If we were to put row $\next(r_{t-1})$ on top of row $r_{t-1}$, the two rows would not contain any illegal pairs
or squares and would therefore satisfy the conditions of Lemma \ref{lem-noCostRows}. So
$A(\next(r_{t-1})) - A(r_{t-1}) \le 6$.

Every location where $\next(r_{t-1})$ and $r_{t}$ differ must be contained
in at least one illegal square and a square contains two locations.To see why this is true,
consider placing $\next(r_{t-1})$ over $r_{t-1}$. The two rows do not contain any illegal pairs
or squares. Since $r_{t-1}$ is valid, there is exactly one head square that contains the head
tile in the two rows. If $r_t$ differs from $\next(r_{t-1})$ at eiher of those two locations, then
the square must be illegal. In all other locations, $r_{t-1}$ and $\next(r_{t-1})$ contain
the same tape tile. If $r_t$ differs from the $\next(r_{t-1})$ in any of those locations,
then that would correspond to two vertically aligned tape tiles that are not the same. 
Any such vertically aligned pair is contained in two illegal squares.

Therefore 
$$A(r_t) - A(\next(r_{t-1})) = A(r_t) - A(\tilde{r}) \le 3 d(\next(r_{t-1}), r_t)
\le 6 v(r_{t-1}).$$
Putting the two inequalities together and using
the fact that $v(r_{t-1}) \ge 1$, we get that
$$A(r_t) - A(r_{t-1}) 
\le 6 v(r_{t-1}) + 6 \le 12 v(r_{t-1}).$$

{\bf Case 3:} Rows $r_0$ through $r_t$ do not have any illegal pairs or 
squares. Then $r_1$ is an end  row. By Lemma \ref{lem-noCostRows},
$A(r_t) - A(r_1)  \le 3$. Since $A(r_1) = 0$,
then $A(r_t) \le 3$. 

{\bf Case 4:}
$v(r_{t-1}) + h(r_{t-1}) = 0$ and there is a $0 \le p \le t-2$
such that $v(r_{p}) + h(r_{p}) > 0$.
Let $p$ be the largest $p$ such that $p \le t-2$
and $v(r_{p}) + h(r_{p}) > 0$.

The rows $r_{p+1}$ through $r_t$ do not contain any illegal pairs
or squares. By Lemma \ref{lem-noCostRows},
  $A(r_t) - A(r_{p+1}) \le 6$. 
 
By the inductive hypothesis,
$$A(r_{p+1}) \le 3 + 12[h(r_{p}) + v(r_{p})] +   \sum_{j=2}^{p-1} 18[h(r_j) +  v(r_j)]$$
Therefore, using the fact that $v(r_{p}) + h(r_{p}) \ge 1$
and $p \le t-2$,
we have that 
\begin{eqnarray*}
A(r_{t}) & \le & 6 + 3 + 12[h(r_{p}) + v(r_{p})] +   \sum_{j=2}^{p-1} 18[h(r_j) +  v(r_j)]\\
& \le & 3 + 18[h(r_{p}) + v(r_{p})] +   \sum_{j=2}^{p-1} 14[h(r_j) +  v(r_j)]\\
& \le & 3 +   \sum_{j=2}^{t-2} 18[h(r_j) +  v(r_j)]
\end{eqnarray*}

\end{proof}

We can put Lemmas \ref{lem-cleanLB} and \ref{lem-potential}
together to get the following
Lemma which summarizes what is needed from the analysis of Layer $1$
for the next Layer. In particular Lemma \ref{lem-analysisL1} bounds the number of interval sizes
in the range $2$ through $\mu(N)+1$ not represented in the last row of Layer $1$. 

\begin{lemma}
\label{lem-analysisL1}
\ifshow {\bf (lem:analysisL1)}  \else \fi
{\bf [Bound on the Number of Missing Interval Sizes]}
Consider a tiling of the $N \times N$ grid. Let $r$ be the tiling in the
last row of Layer $1$. Let $S$ denote the set of integers such that $s \in S$ if there is a clean interval of size $s$ in row $r$. Let $F_1$ denote
the total number of faults in Layer $1$. Then
$$|\{2,3,\ldots,\mu(N)+1\} - S| \le 44 F_1 + 3$$
\end{lemma}

\begin{proof}
Let $(s_1, \ldots, s_m)$ denote the sizes of the clean intervals,
from left to right, in  row $r$. 
We can remove all the intervals of size $1$ and according to Lemma \ref{lem-remove1}, the value of $A$
for that sequence will only decrease. 
Since the size of the remaining intervals is at least $2$,
we can apply Lemma \ref{lem-seq} with $r = \mu(N)+1$, 
$$ \left| \{2, \ldots, \mu(N)+1 \} - S \right|  \le
\max\{\mu(N)-m,0\} + A(s_1, \ldots, s_m).$$ 
By Lemma \ref{lem-cleanLB}, $m$, the number of clean intervals
in the last row of Layer $1$ is at least $\mu(N) - 26F_1$,
and therefore $\max\{\mu(N)-m,0\} \le 26 F_1$.
Since $F_1 = \sum_{t=0}^{N-1} [h(r_t) + v(r_t)]$,
by Lemma \ref{lem-potential}, the value of $A(s_1, \ldots, s_m)$,
which is the value of $A$ for the last
row of Layer $1$ is at most $18 F_1 + 3$.
\end{proof}

\subsection{Characterizing the intervals in an fault-free tiling}
\label{sec-FF}
\ifshow {\bf (sec:FF)}  \else \fi

In order to compare a tiling with faults to an fault-free tiling,
we will eventually need to characterize the sequence of interval sizes in an fault-free
tiling. 
At the end of every iteration of the outer loop, the sequence of intervals starts at some number
$m$ and decreases by $1$, going from left to right, until the last interval which has size $2$.
In the middle of an iteration of the Outer Loop, the sequence of interval lengths can differ slightly
from this ideal case. The extent of the difference is characterized in the lemma below.

\begin{lemma}
\label{lem-errorFreeSizes}
\ifshow {\bf (lem:errorFreeSizes)}  \else \fi
{\bf [The Sequence of Interval Sizes in a Fault-Free Tiling]}
Consider a row $r$ that represents the configuration of the Turing Machine in an fault-free
execution in which the TM is in the $m^{th}$ iteration of the Outer Loop.
The number of intervals in $r$ is $m$.
Define the set $S$ such that $j \in S$ if and only if there is an interval of size $j$ in row $r$.
\begin{enumerate}
\item The sizes of the intervals form a non-increasing sequence from left to right.
\item There are at most two intervals with the same size and the largest size only appears once.
\item $S \subseteq \{1, 2, \ldots, m+2\}$
\item $|\{2, \ldots, m+2\} - S| \le 2 $
\end{enumerate}
\end{lemma}

\begin{proof}
The Turing Machine starts with one interval of size $2$.
After $m-1$ complete iterations of the Outer Loop, $m-1$ intervals have been added, for a total of $m$ intervals.
The sizes of those intervals is $(m+1, m, \ldots, 2)$. This sequence satisfies properties $1$ through $4$.
Now consider the next execution of the Outer Loop. After the $j-1$ complete iterations of the Inner Loop,
the leftmost $j-1$ intervals have increased in size by $1$, so the sequence is $(m+2, m+1, \ldots, m-j+4, m-j+2, m-j+1, \ldots 2) $. This sequence also satisfies $1$ through $4$. At this point $s_{j-1} = m-j+4$ and $s_j = m-j+2$.

In the middle of the $j^{th}$ iteration of the Inner Loop, the intervals stay the same as the head sweeps left. 
When the head  reaches the left end of interval $j$, it sweeps right and increases the size of that interval by $1$
when it reaches the right end of the interval. As the head sweeps to the right moving all the tape symbols over by $1$, one of the intervals to the right of interval $j$ is decreased temporarily by $1$. So if
$\vec{s} = (m+2, m+1, \ldots, m-j+4, m-j+3, m-j+1, \ldots 2)$ and $\vec{t}$ is the current sequence, then
$s_i = t_i$, except for one $k \in {j+1, \ldots, m}$ where $t_k = s_k-1$.  The sizes are non-increasing from left to right (property $1$). The only two intervals with the same size are $t_k$ and $t_{k+1}$ (property $2$).
All numbers are in the range $\{1, 2, \ldots, m+2\}$ (property $3$), and
the two numbers in the range $\{1, 2, \ldots, m+2\}$ which are not in $S$ are $m-j+2$ and $s_k$ (property $4$).

After $m$ iterations of the inner loop, the sequence of interval sizes is $(m+2, m+1, \ldots, 3)$. The next iteration of the Outer Loop begins when an interval of size $2$ is added to the right end.
\end{proof}

\begin{lemma}
\label{lem-muBounds}
\ifshow {\bf (lem:muBounds)}  \else \fi
{\bf [Bounds on $\mu(N)$]}
The function $\mu(N) \ge N^{1/4}/2$. In addition $\mu(N)$ is $O(N^{1/4})$ and the exact value of $\mu(N)$ can be 
computed in time that is poly-logarithmic in $N$.
\end{lemma}

\begin{proof}
In an fault-free tiling, the number of intervals increases by $1$
whenever an end row is reached. Row $r_1$ is an end row which has one interval.
After the $t^{th}$ end row, there are $t$ intervals, of sizes $t+1, \ldots, 2$.
By Lemma \ref{lem-valid}, the number of rows until the next end row
is $\sum_{j=1}^t [2j (s_j-1) + 1]$, where the sizes of the intervals
$s_1, \ldots, s_t$ are numbered from left to right. 
Plugging in $s_j = t - j + 2$, for $j = 1, \ldots, t$ gives
$\sum_{j=1}^t [2j (t - j +1) + 1]$.
Thus $\mu(N)$ is the largest value of $m$ such that 
\begin{equation}
\label{ineq:mu}
    1 + \sum_{t = 1}^{m-1} \sum_{j=1}^t [2j (t - j +1) + 1] \le N-2.
\end{equation}
Let $f(m)$ be defined to be 
$$f(m) = \sum_{t = 1}^{m-1} \sum_{j=1}^t [2j (t - j +1) + 1] = \sum_{t = 1}^{m-1} \left[ \frac 1 3 t^3 +  t^2 + \frac 5 6 t \right].$$
Note that $f(m) = \Theta(m^4)$. Therefore there is a constant $c$ such that if $m \ge c N^{1/4}$, them $f(m) \ge N$
which means that $\mu(N) = O(N^{1/4}$.
To get the more precise lower bound on  $\mu(N)$:
$$f(m) = \sum_{t = 1}^{m-1} \left[ \frac 1 3 t^3 +  t^2 + \frac 5 6 t \right]
\le \frac {13}{6}(m-1)^4 = \frac {13}{6}[m^4 - 4 m^3 + 6 m^2 - 4m +1]
\le \frac {13}{2}m^4 - 3.$$
If $\mu(N) = N^{1/4}/2$, then $f(m) + 1 < N-2$, which means that $\mu(N) \ge N^{1/4}/2$.
Finally, since $f(m)$ has a closed form that is a degree $4$ polynomial in $m$,
it is possible to binary search on the range $N^{1/4}/2\le m \le cN^{1/4}$
to find the larges value $m$ for which the Inequality (\ref{ineq:mu}) holds.
The number of iterations is $O(\log N)$ and computing the value of $f(m)$
can be done in time that is polynomial  in $\log N$.
\end{proof}

\section{$\Theta(N^{1/4})$-Gapped Weighted Tiling is $\nexp$-complete} 

Containment is straight-forward since given a number $N$ expression in binary
and a tiling of an $N \times N$ grid, which is exponential in $\log N$,
the size of the input the cost of the tiling can be computed in $O(N^2)$ time and
it can be verified whether the cost of the tiling is $0$ or at least $c N^{1/4}$
for some constant $c$.

To establish the hardness of the Gapped Weighted Tiling Problem, we will show
for and arbitrary  $L \in \nexp$,  a reduction to a translationally invariant 2D Tiling with an $n^{1/4}$ gap.
Specifically, we will show a finite set of tiling rules so that, by a polynomial time computable function
$x \rightarrow f(x) = N$,

\begin{itemize}
\item If $x \in L$, then there is a $0$ cost tiling of an $N \times N$ grid using the tiling rules.
\item If $x \not\in L$, then any tiling of an $N \times N$ grid has cost at least $\Omega(N^{1/4})$.
\end{itemize}

Since $L \in \nexp$, there is an exponential time verifier $V$ that can verify that a string
$x \in L$ given a witness whose size is exponential in $|x|$.
We will assume that on input $x$, if $|x| = n$, then the verifier $V$ runs in time and
space at most $2^{\delta n}$, including the space required for the witness. This can be achieved by padding:

\begin{claim}
\label{claim-pad}
\ifshow {\bf (claim:pad)}  \else \fi
{\bf [Padding Argument]}
If $L \nexp$, then for any constant $\delta$, $L$ is polynomial-time reducible to $L' \in \nexp$
such that the verifier $V'$ for $L'$ uses time and space  $2^{\delta n}$.
\end{claim}
\begin{proof}
Suppose that the verifier for $L$ uses time (and space) at most $2^{cn}$. Define a new language $L'$
$$L' = \{x 0^{dn} | x \in L ~\mbox{and}~ |x| = n \},$$
where $d = c/\delta -1$.
A verifier $V'$ for $L'$ will first make sure there are the correct number of $0$'s at the end of the input, then will
erase the $0$'s (incurring a small
overhead).  Then $V'$ will simulate $V$ on $x$. The running time is close to $2^{cn}$. The length of the input is
$(d+1)n$. So as long as $c = \delta (d+1)$, the running time will be $n^{\delta (d+1)n}$.
\end{proof}

The large cost for $x \not\in L$ is achieved by
$\Theta(N^{1/4})$ independent computations, each of which
will run the verifier for $L$ on the input $x$. 
First we need to establish that there are enough intervals
created in Layer $1$ that are wide enough to simulate the
execution of the verifier. 
 Lemma \ref{lem-L1gapped} shows
that the analysis provided in Section \ref{sec-L1analysis} is
sufficient to establish that there will be at least
$N^{1/4} - O(F_1)$ intervals of size at least $N^{1/4}$.

\begin{lemma}
\label{lem-L1gapped}
{\bf [Results From Layer $1$]}
Consider a tiling of Layer $1$ and let 
$F_1$ be the total number of illegal pairs or squares in the tiling.
Then there are at least $N^{1/4}/4 - 44 F_1 + 3$ clean intervals
of size at least $N^{1/4}/4$.
\end{lemma}

\begin{proof}
Let $(s_1, \ldots, s_m)$ be the sequence of sizes of the clean
intervals in a tiling of Layer $1$. 
Lemma \ref{lem-analysisL1} says that
the number of integers in the range $2, \ldots, \mu(N)$
that are not present in the sequence $(s_1, \ldots, s_m)$
is at most $44F_1 + 3$. By Lemma \ref{lem-muBounds},
$\mu(N) \ge N^{1/4}/2$, so there are at least $N^{1/4}/4$
numbers in the range $2, \ldots, \mu(N)+1$ of value at least $N^{1/4}/4$.
At most $44F_1 + 3$ are missing. Therefore the number of clean
intervals of size at least $N^{1/4}/4$ is at least
$N^{1/4}/4 - 44 F_1 + 3$.
\end{proof}

It now remains to 
give the description of the constructions for Layers $2$ and $3$.

\subsection{Layer 2}
\label{sec-L2}

Since Layer $1$ runs bottom to top and Layer $2$ runs top to bottom,
 the only interaction between Layers 1 and 2 takes place in $r_{N-2}$,
 which is the last row for Layer $1$ and the first row for Layer $2$.
The translation rules from Layer $1$ to Layer $2$
will ensure that the heavy tiles from Layer $1$ are
copied as $X$ tiles in Layer $2$. Thus, the intervals in Layer $2$ will have an $X$
on each end.

The Layer $2$ tiling rules also enforce that
an $X$ tile must go above an $X$ tile and can not be placed above any other tile.
This implies that in a no-cost tiling, the $X$ tiles form columns of $X$'s in Layer 2.
In the vertical strip
of tiles between each column of $X$'s, there will be an independent Turing Machine computation. The tiling
rules will enforce that the head of each Turing Machine stays within it's strip.

The tile types for Layer 2 include  $\Box$, $X$, and $\#$.
Any square in which $X$ is directly above or below a tile that is not $X$ or $\Box$
is illegal. Similarly for $\#$,
any square in which $\#$ is directly above or below a tile that
is not $\#$ or $\Box$ is illegal. Since the $\Box$ tiles are only around
the perimeter of the grid, the $X$'s and the $\#$'s must form columns in the interior
of the grid. An example is shown in Figure \ref{fig-strips}.

\begin{figure}[ht]
  \centering
  \includegraphics[width=3.0in]{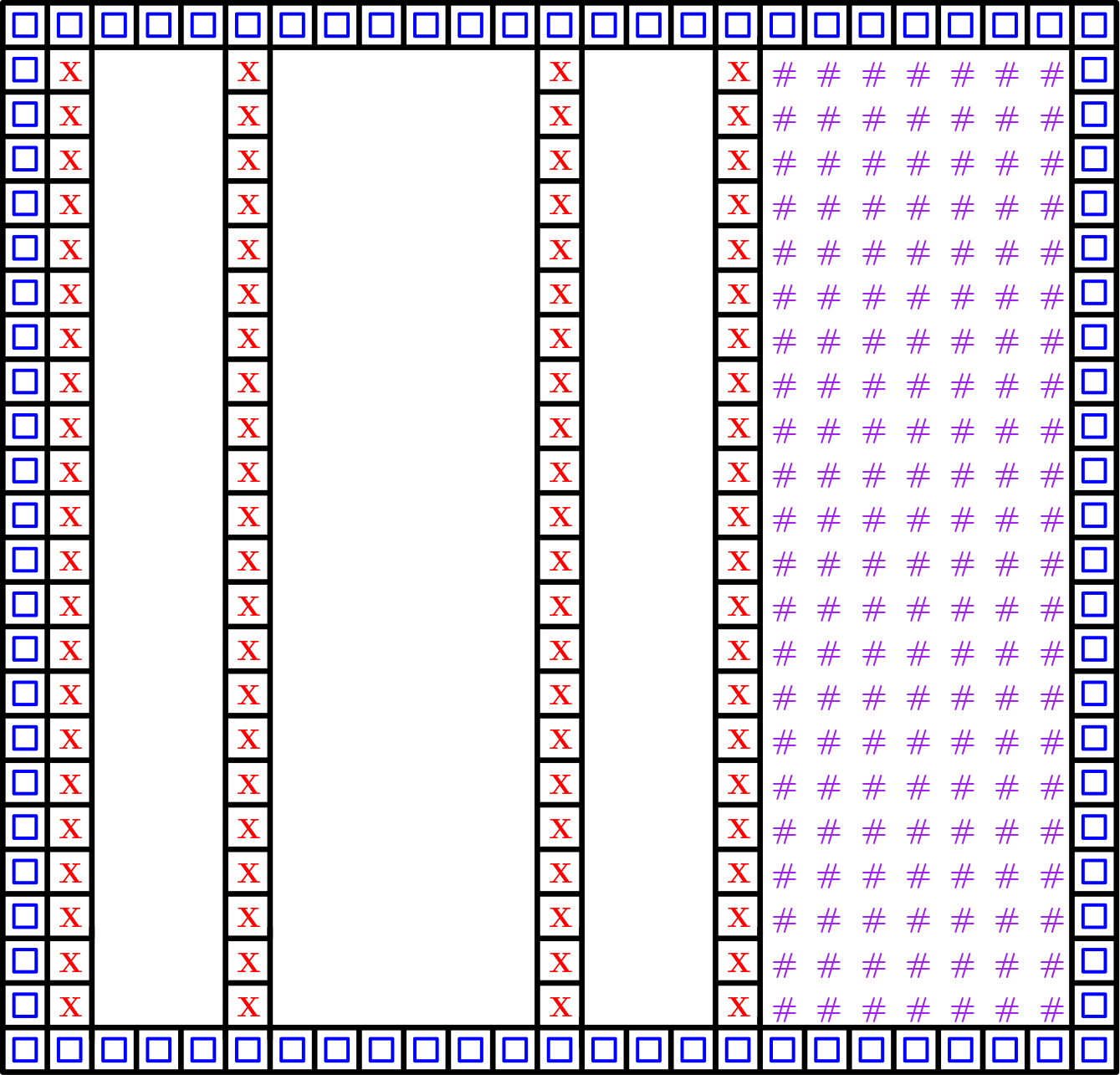}
\caption{A schematic of a Layer 2 tiling. Each strip between the column of X's will contain a tiling
corresponding to an independent Turing Machine computation. In the example shown here, there are three
independent computations.}
\label{fig-strips}
\end{figure}

There will also be tile types representing the execution of a Turing Machine
in between the two columns of $X$ tiles. The tape symbols for the Turing Machine will be $\{S, 0, 1, B, T\}$
and the states will be $\{q_r, q_l\}$. Thus, there will be 
tile types for each of the tape symbols (called {\it tape} tiles) and tile types for state-symbol pairs (called {\it head} tiles) of the form ($q/c$) where $q$ is a state and $c$ is a tape symbol. 

The  relationship  between Layer $1$ and $2$ tiles is summarized below. The  rules will
enforce the condition that for any tile directly below a $\Box$ tile,
if the tile type in Layer $1$ is as indicated on the left,
then the tile type for Layer $2$ must be one of the choices
on the right. 

\begin{eqnarray*}
\# & \rightarrow & \#\\
\mbox{Tile type}~t \neq \#~\mbox{and}~
w(t) > 0 & \rightarrow & X\\
\mbox{Tile type}~t \neq \#~\mbox{and}~
w(t) = 0 & \rightarrow & B~\mbox{or}~(q_l/S)~\mbox{or}~T\\
\end{eqnarray*}

The translation rules enforce that the intervals in the last row of Layer $1$
are preserved in the first row of Layer $1$, except for
intervals of size $1$ which correspond to a tile of weight $2$
in Layer $1$.

In addition, Figure \ref{fig-bottomRowRulesL2} shows the initialization
rules for Layer $2$.
The meaning of the graph is that any square with $\Box~\Box$ in the top
row and two Layer $2$ tiles $t_1~t_2$ in the bottom row
that do not correspond to an edge from $t_1$ to $t_2$
in the graph is illegal.
This resolves the ambiguity in whether a weight-$0$ tile in Layer $1$
gets copied to a $B$, $T$, or a $(q_l/S)$.
If an interval has no illegal initialization squares, the interval
could be $X~X$ or $X~T~X$. These are the only possibilities for intervals
of size $2$ or $3$.
If an interval has has no illegal initialization squares
and has size at least $4$, the the interval
must have the form $X~(q_l/S)~B^*~T~X$.

\begin{figure}[ht]
\centering
  \includegraphics[width=3.0in]{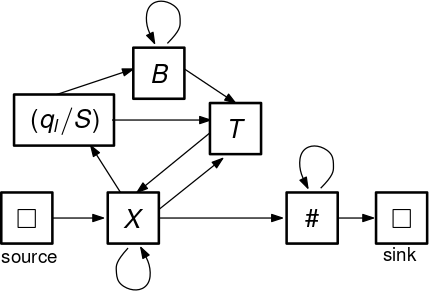}
\caption{These rules constrain the contents of the first row in Layer 2.}
\label{fig-bottomRowRulesL2}
\end{figure}

Figure \ref{fig-correctTranslation} shows and example of a possible last row
for Layer $1$ and its correct translation to the first row of Layer $2$.

\begin{figure}[ht]
\centering
  \includegraphics[width=4.7in]{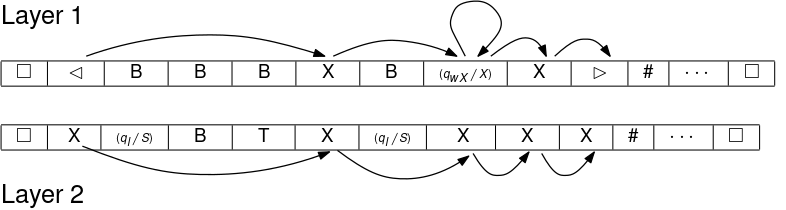}
\caption{An example showing the last row from Layer $1$ and its correct translation to the first row of Layer $2$. The intervals in the two rows are shown with arrows. Every interval in Layer $1$ is translated to an interval of equal size in Layer $2$, except for the interval of size $1$.}
\label{fig-correctTranslation}
\end{figure}

The Turing Machine that is executed within each strip continually increments
a binary counter that is written in reverse on the tape. 
We call this the Binary Counter Turing Machine. The rules are summarized below. 

\begin{eqnarray*}
\delta(q_l, S) & = & (S, q_r, R) \\
\delta(q_r, 1) & = & (0, q_r, R)\\
\delta(q_r, b) & = & (1, q_l, L)~\mbox{for}~ b \in \{0, B\}\\
\delta(q_l, b) & = & (b, q_l, L) ~\mbox{for}~ b \in \{0, 1\}\\
\delta(q, T) & = & (T, q, - )~\mbox{for}~ q \in \{q_r, q_l\}\\
\end{eqnarray*}

In a single iteration, the head starts pointing to the $S$ in state
$q_l$. It transitions to $q_r$ and moves right.
The head then moves right (in state $q_r$) changing $1$'s to $0$'s
until a $0$ or $B$ is encountered. The $0$ or $B$ is overwritten
with $1$ and the head transitions to $q_l$ and moves
left until the $S$ is reached again.
If the head ever reaches the $T$ symbol at the right end of the interval, 
the Turing Machine hits an infinite
loop and never changes state again.

The TM rules are translated into legal and illegal squares
as described in Section \ref{sec-TM2Tile}. The $X$ tile
is treated as a tape symbol. For example, the rule $\delta(q_l, S) = (S, q_r, R)$ would mean that the  square shown is
legal. (Recall that the computation for Layer $2$ goes from top to bottom.)

\vspace{.1in}

\begin{tabular}{|c|c|}
\hline
$X$ & $(q_l/S)$ \\
\hline
$X$ & $S$ \\
\hline
\end{tabular}

\vspace{.1in}

The rule $\delta(q_r, T) = (T, q_r, -)$ would mean that the  square shown is
legal.

\vspace{.1in}

\begin{tabular}{|c|c|}
\hline
  $(q_r/T)$ & $X$ \\
\hline
  $(q_r/T)$ & $X$\\
\hline
\end{tabular}

\vspace{.1in}

The tape symbol $S$  prevents the head from trying to
move left into the $X$ on the left and the symbol $T$ prevents the head from moving
into the $X$ on the right.
Thus, if the strip is not wide enough for
the computation, the head reaches the $T$ symbol and the
computation eventually gets stuck and
does not advance. This does not cause any
additional cost, so for every strip in which the computation starts in configuration
$(q_l/S)~B^*~T$,  there is always a unique
way to tile that strip so that it does not contain
any illegal squares.

If the goal is to produce a string $x$ in the last row
of Layer $2$, one could calculate the number $N$ such
that after $N-3$ steps, 
the Turing Machine is in state $(q_l/S)$
and contents of the counter is $x$.
Note that the contents of the counter always ends in $1$,
so in order to produce an arbitrary string $x$, one
can produce $x1$ and then ignore the last bit.
The lemma below implies that the function mapping $x$ to $N$ is polynomial time
computable and is the function used for the reduction.

\begin{lemma}
\label{lem-bctm}
\ifshow {\bf (lem:bctm)}  \else \fi
{\bf [Number of Steps Use by the Binary Counter TM]}
Consider a binary string $x$. Let $x^R$ denote the reverse
of the string $x$. Let $n(1x)$ be the value of the number whose binary representation
is $1x$ and let $w(x1)$ denote the number of $1$'s in the string $x1$.
Then the number of steps required by the Binary Counter Turing Machine to
write the string $x$ and end up with the head pointing to $S$ is
$4n(1x^R) - 2w(x1)$.
\end{lemma}

\begin{proof}
Let $f(n)$ be the number of steps until the Turing Machine ends up
in configuration $(q_l/S)~0^n~1$.
$f(0) = 2$. In order to end up in configuration $(q_l/S)~0^n~1$,
the BCTM must first reach $(q_l/S)~1^n~B$. Then it takes $2(n+1)$ steps to complete the
next increment step and reach
$(q_l/S)~0^n~1$. The number of steps to reach
$(q_l/S)~1^n~B$ is $f(n-1) + f(n-2) + \cdots f(0)$. Therefore the function
$f$ obeys the recurrence:
$$f(n) = \sum_{j=0}^{n-1} f(j) + 2(n+1).$$
The solution to this recurrence relation is $f(n) = 2(2^{n+1}-1)$.
If the bits of $x$ are numbered from left to right $x_1 x_2 \cdots x_{n}1$, then 
the number of steps to reach $(q_l/S)~x$ is
$$2(2^{n+1} - 1) + \sum_{j=1}^{n} x_j \cdot 2(2^j - 1) = 4 \cdot 2^n + \sum_{j=1}^n 4 x_j  \cdot 2^{j-1} -
2 w(1x) = 4 n(1x^R) - 2 w(x1).$$
\end{proof}

\subsubsection{Layer 2 Intervals}

We now need to extend the definition of intervals, as well as clean and 
corrupt intervals to Layer $2$.
In Layer 2, an interval begins with an $X$ tile and extends to the right up to and
including the next $X$ tile. Note that since heavy tiles on Layer 1
get translated to $X$'s on Layer 2, the intervals of size greater than $1$
stay intact if the translation
is done correctly.

The translation is done at the top end of the grid, so the last row for Layer $1$,
which is also the first row for Layer $2$, is row $r_{N-2}$. Row $r_{N-1}$ is the
top row of the grid which contains all $\Box$ tiles.
An interval in the first row of Layer $2$ is clean if the interval does not contain
any illegal translation or initialization
squares spanning rows $r_{N-2}$ and $r_{N-1}$
and the corresponding interval
in the last row of Layer 1 was also clean. Otherwise the interval is corrupt. 
Lower down in the tiling, an interval in  a row of Layer $2$  is clean
if the sequence of tiles is also a clean interval in row above it
and there are no illegal computation squares spanning the two consecutive rows
in that interval.
Since clean intervals do not move or change in size, two clean intervals 
in rows $r_t$ and $r_{t+1}$ have the same tag (i.e. they ``correspond'')
if they occupy the same set of locations in their respective rows.

Let $F_2$ be the number illegal squares in Layer $2$ of a tiling plus the
number of illegal translation and initialization squares between Layers $1$ and $2$.

\begin{lemma}
\label{lem-L2analysisGap}
\ifshow {\bf (lem:L2analysisGap)}  \else \fi
{\bf [Number of Clean Intervals Lost from Layer $1$ to Layer $2$]}
Let $T_1$ be the set of tags corresponding to clean intervals of size at least $2$
in the last row 
of Layer $1$. Let $T_2$ be the set of tags corresponding to clean intervals in the last row 
of Layer $2$. Then $T_2 \subseteq T_1$ and $|T_1 - T_2| \le F_2$.
The size and location of a clean interval with tag $j$ is the same in the last
row of Layer $1$ and any row in Layer $2$. 
\end{lemma}

\begin{proof}
We will account for any changes in the set of clean
intervals from the last row of Layer $1$ to the first row of
Layer $2$ by illegal translation squares in each interval.
Note since two neighboring intervals only overlap on one tile,
an illegal square can be contained in at most one interval.

If an interval of size at least $2$ is clean in the last row of Layer $1$ 
then that sequence of tiles consists of a heavy tile, followed by a sequence
of weight-$0$ tiles, followed by a final heavy tile.
Since heavy tiles are translated into $X$ tiles and weight-$0$ tiles
are translated into non-$X$ tiles, then if the sequence is correctly
translated, it results in an interval in the first row of Layer $2$.
In this case, the two intervals occupy the same locations in the row
and have the same tag.
If a clean interval at the end of Layer $1$ does not correspond
to a clean interval in Layer $2$, then the interval must contain
an illegal translation square.

Similarly, consider a clean interval in row $r_t$ of Layer $2$. The interval
starts with an $X$, is followed be a sequence of non-$X$ tiles, and finally
ends with an $X$ tile. If there are no illegal squares in the interval
spanning rows $r_t$ and $r_{t-1}$, then in row $r_{t-1}$, the
sequence begins and ends with $X$ and only has non-$X$ tiles in between
and therefore corresponds to an interval. 
This follows from the fact
that any square with an $X$ tile above or below a non-$X$ tile is illegal.
The interval is clean in $r_t$ only if the interval is clean in $r_{t-1}$
in which case the two intervals occupy the same set of locations in their respective
rows and have the same tag.
A clean interval in row $r_t$
that does not correspond to a clean interval in row $r_{t-1}$, must
contain an illegal computation square spanning rows $r_t$ and $r_{t-1}$.
\end{proof}

\begin{lemma}
\label{lem-L2contents}
\ifshow {\bf (lem:L2contents)}  \else \fi
{\bf [Contents in Each Clean Interval at the End of Layer $2$]}
Consider a tiling of an $N \times N$, where 
$N = 4n(1x^R) - 2 w(x1) + 3$ for some binary string $x$. Then
every clean interval in the last row of Layer $2$ of size at least $4$
that does not contain
a $(q_r/T)$ or $(q_l/T)$ tile has the form:
$$X~(q_l/S)~x1~B^*~T~X$$
Moreover, every clean interval in the last row of Layer $2$
of size at least $\log N + 5$ does not contain
a tile of the form $(q_r/T)$ or $(q_l/T)$.
\end{lemma}
\begin{proof}
If an interval is clean in the first row of Layer $2$ (row
$r_{N-2}$), then there are no
illegal initialization squares in that interval spanning $r_{N-1}$ and $r_{N-2}$
which means that the interval must correspond to a path in the graph
denoted in Figure \ref{fig-bottomRowRulesL2}.
The interval begins and ends with $X$ and has no $X$ tiles in the middle.
The only possible path in the graph in Figure \ref{fig-bottomRowRulesL2} that begins and
ends with $X$ with no other intervening $X$'s and has length at least $4$,
corresponds to
$X~(q_l/S)~B^*~T~X$.

By induction on $t$,  if the interval is clean in row $r_{N-2-t}$,
then it is clean in rows $r_{N-2}$ through $r_{N-2-t}$ and
 the contents of the interval in row $r_{N-2-t}$
 represents the configuration of
the Turing Machine after $t$ steps,   starting with
$X~(q_l/S)~B^*~T~X$.

If an interval is clean in row $r_1$, then it represents the 
state of the Turing Machine after $N-3$ time steps,   starting with
$X~(q_l/S)~B^*~T~X$.
If at any point in these $N-3$ time steps, the head reached the $T$ at the
right end of the interval, then the head will stay in that position
and the interval will contain a $(q/T)$ tile in row  $r_1$.

If the interval row $r_1$ does not contain a $(q/T)$ tile, then
the head never reached the $T$ during the first $N-3$ time steps.
This implies that the configuration is the same as it would have been
if there had been an infinite sequence of $B$ symbols to the right of
$(q_l/S)$ in the initial time step. By Lemma \ref{lem-bctm},
the contents of the interval
will be $X~(q_l/S)~x1~B^*~T~X$ in row $r_1$.

After $N-3$ time steps, by Lemma \ref{lem-bctm}, the length of the counter is at most $\log N+1$.
Therefore, the number of tape symbols that the head has reached is at most
$\log N + 2$, including the $S$ to the left of the counter. 
If the interval has size at least $\log N + 5$,
then excluding the $X$ on the left end and the $T~X$ on the right end,
there are $\log N + 2$ tiles.  This is enough room for the Turing Machine
to complete $N-3$ steps without reaching the $T$ on the right
end of the interval which means that the interval in row $r_1$ does not contain a $(q/T)$ tile.
\end{proof}

\subsection{Layer $3$ for Gapped Weighted Tiling}
\label{sec-GWT}

In this subsection, we describe the translation rules from Layer $2$
to Layer $3$ and give a high level description of the Turing Machine
that operates within each strip.
For every $t \in \{X, \Box, S, T, B, 0, 1, \# \}$, a $t$ tile is translated
to another $t$ tile from Layer $2$ to Layer $3$.
The translation of the head tiles of the form $(q/c)$, depend on the
tape symbol $c$. For any state $q$ and tape symbol $c$ in the 
Layer $2$ Turing Machine, the following translation rules apply:

\begin{eqnarray*}
(q/T) & \rightarrow & (q_{s1}/T)\\
\mbox{for}~c \neq T, (q/c) & \rightarrow & (q_{s2}/c)
\end{eqnarray*}

To summarize, in translating from Layer $2$ to Layer $3$,
the state information from Layer $2$ is lost and the
new state depends only on whether the head of the Turing Machine
in Layer $2$ reached the $T$ on the right end of the interval.
If the head is pointing to $T$, then the new state on Layer $3$ is
$q_{s1}$, otherwise the new state is $q_{s2}$.
The tape symbols are translated without change.

The Turing Machine that starts in
state $q_{s1}$ 
repeatedly
executes the single move $\delta(q_{s1}, T) = (q_{s1}, T, -)$  until the last row of Layer $3$.
Thus if a computation in Layer $2$ reaches the $T$ at the right end 
of its interval, then
it remains stuck for the rest of Layer $3$. As long
as each step is executed correctly, there are no
additional costs in these small intervals.

Recall that we would like
to show that for any language $L$ in $\nexp$, we will construct a set of tiling rules
and a mapping from any string $x$ to a number $N$ such that
if $x \in L$, then there is a way to tile the $N \times N$ grid with
zero cost (no illegal pairs or squares) and if $x \not\in L$,
then any tiling of the $N \times N$ will require cost that is
$\Omega(N^{1/4})$. 
The Turing Machine that starts in state
$q_{s2}$  will guess a witness $w$ and launch the verifier
Turing Machine for a language $L \in \nexp$ with input $x$ and witness $w$,
where $x1$ is the string written on the tape at the end of Layer $2$.
Note that the Turing Machine in Layer $2$ always produces a string that
ends in $1$, so in order to produce an arbitrary string, the last bit 
of the string produced is ignored.
If at the end of Layer $2$,
the interval is clean and the head has not reached
the $T$ at the right end 
of the interval, then according to Lemma \ref{lem-L2contents},
it has the correct $x$ written
on the tape of the Turing Machine.
The second Turing Machine (that starts in $q_{s2}$)
will also have the rule
$\delta(q, T) = (q, T, -)$ for any $q$.
Thus, intervals which are not wide enough to complete the computation of
$V$ on $(x, w)$ (for any guess $w$) can be tiled without any additional cost.
If the computation is able to complete and accepts, then there is no cost.
Any square that contains a rejection state of the Turing Machine
$V$ will incur a {\em rejection cost}. So an  interval that is clean at the beginning of
Layer $3$ and is wide
enough to perform the computation that ends
up in a rejecting state will contain at
least one illegal square or square with a rejection cost.

{~}

{\bf Costs of tiles and the perimeter tiles}

{~}

Recall that the tile types consist of border tiles $\Box$ or interior
tiles. Each interior tile is specified by it's tile type for each of the
three layers. For any configuration of four tiles arranged in a square,
let $p = 1$ if the bottom two tiles are an illegal pair for Layer $1$
and let $f_i = 1$ if the square is an illegal square for Layer $i$ or
an illegal translation square from Layer $i-1$ to Layer $i$.
Let $r = 1$ if the square contains a rejecting state for Layer $3$.
The values $p$, $r$, and $f_i$ are $0$ otherwise.
The cost for that square is then $p + f_1 + f_2 + f_3 + r$.
If $F_i$ is the number of illegal pairs or squares in Layer $i$ in a tiling,
and $R$ is the number of square on Layer $3$ that contain rejecting states, 
then the cost of that tiling is $F_1 + F_2 + F_3 + R$.

The last technical point that we need to address before proving the
hardness result for Gapped Weighted Tiling is to address the assumption
that the perimeter of the grid consists of $\Box$ tiles and that there
are no $\Box$ tiles on the interior of the grid.
Towards this end, we create four types of $\Box$ tiles: NW, NE, SE, SW.
The designation of a square that contains a $\Box$ tile
as legal or illegal does not depend on the type of the $\Box$ tile.
Let $C = 21$. 
We will adjust the costs for each square by adding the following amount to the cost of a square
if a tile of the given type is in that location of the square:

{~}

\begin{tabular}{|c|c|c|c|c|}
\hline
 & upper left  & upper right & lower right  & lower left\\
 \hline
 \hline
 NW-$\Box$ & -C & & +2C  & \\
 \hline
 NE-$\Box$ & & -C & & +2C  \\
 \hline
 SE-$\Box$ & +2C  & & -C & \\
 \hline
 SW-$\Box$ & & +2C  & & -C  \\
 \hline
 \end{tabular}

\begin{lemma}
\label{lem-border}
{\bf [Validating the Assumption About Perimeter Tiles]}
In any minimum cost tiling, the perimeter of the grid will consist of $\Box$ tiles
and no border tile will be contained in the interior of the grid. Moreover, there is a
way to tile the perimeter with $\Box$ tiles so that the total contribution due to the
benefits and penalties from $\Box$ tiles is $-4C(N-1)$.
\end{lemma}

\begin{proof}
The cost of any square before the adjustments due to the border tiles is at most $5$.
Since each tile participates in at most four squares, changing a tile can cause the
cost to change by at most $20$, ignoring the penalties and benefits due to the $\Box$ tiles.
Consider a tiling of the grid with at least one $\Box$ tile on the interior. If a $\Box$ tile on the interior is changed to a non-$\Box$ tile, the cost will increase by at most $20$
due to changes in legal/illegal squares, one square will lose the $-C$ benefit by having
a $\Box$ tile in one of its corners. This will amount to a total increase of $20+C$.
However at least one square will lose the $2C$ penalty of having a border tile in the wrong 
corner. The change in the
cost of the tiling will be $(20+C)-2C = 20 - C$, which is negative.

If there is a non-$\Box$ tile on the perimeter, then replace that tile with a
type of $\Box$ tile
that will  get the $-C$ benefit to the cost and no $2C$ penalty. The cost will increase by
at most $20$ due to changes in legal/illegal squares, so the total change in cost 
will be $20-C$ which is negative.

These changes  can be continued until there are no $\Box$ tiles on the interior and only
$\Box$ tiles on the perimeter. Each swap decreases the cost of the tiling.

For each location on the border, there is a  type of $\Box$ tile such that
placing that type of $\Box$ tile in that location will result in one square with the
$-C$ benefit and no squares with the $2C$ penalty. Since there are $4(N-1)$ tiles
on the perimeter, the claim follows.
\end{proof}

We are now ready to prove the hardness result for Gapped Weighted Tiling:

\begin{theorem}
\label{th-GWT}
\ifshow {\bf (th:GWT)}  \else \fi
$f(n)$-GWT in $2$-dimensions is $\nexp$-hard for some $f(n)$ that is $\Omega(n^{1/4})$.
\end{theorem}

\vspace{.1in}

\begin{proof}
Given binary string $x$, let $N$ be the number such that after $N-3$ steps
of the binary counter Turing Machine from Layer $2$, the contents of the
tape are $x1$ and the head is pointing to $S$ in state $q_l$.
Note that the string representing the binary counter in the Layer $2$ Turing
Machine always ends in $1$, so to get an arbitrary string $x$, we pick
$N$ to produce $x1$ and then ignore the last $1$. 
According to Lemma \ref{lem-bctm},
the function mapping $x$ to $N$ is polynomial time computable
and $|x| \le \log N$.

If $x \in L$, then there a way to tile the grid with cost $-4C(N-1)$.
This is achieved by first tiling the perimeter with $\Box$ tiles
so that the total benefit from the $\Box$ tiles is $-4C(N-1)$.
For the interior, have every
Turing Machine in every layer  executed without a fault.
This means that every interval at the end of Layer $2$ is a clean interval.
If the interval contains $(q/T)$ at the end of Layer $2$, 
the state is translated to $q_{s1}$ in Layer $3$ which initiates 
the Turing Machine that  stays in the same state,
incurring no cost.
For all the intervals that do not contain $(q/T)$ at the end of Layer $2$,
the state is translated to $q_{s2}$
in Layer $3$.
According to Lemma \ref{lem-L2contents}, these intervals all have the
correct binary string $x$ which is translated to Layer $3$
and serves as the input to the
 second Turing Machine (that starts in $q_{s2}$). If the head hits the right end of the interval in the Layer $3$ computation,
then the Turing Machine stays in the same 
state for all the remaining steps, incurring
no cost.
The remaining intervals that are wide enough to complete a computation
of $V$, will guess the correct witness $w$ and will accept on input
$x$ and $w$. Since accepting computations do not incur a cost,
the overall cost of the tiling is $0$.

Now supposed that $x \not\in L$. Because of Lemma \ref{lem-border},
we can assume that the minimum tiling has only $\Box$ tiles along the perimeter
and no $\Box$ tiles on the interior of the grid. 
The total benefit from the border tiles will be at most $-4C(N-1)$.
We will prove that the cost due to legal/illegal squares or rejecting computations
will be at least
$N^{1/4}/c$ for some constant $c$.

For $i \in \{1, 2, 3\}$, let $F_i$ denote the number of illegal squares and pairs
in Layer $i$. For $i \ge 2$, $F_i$ also includes the number of illegal translation
squares from Layer $i-1$ to Layer $i$.
By Lemma \ref{lem-L1gapped}, at the end of Layer $1$,
there will be at least $N^{1/4}/4 - 44F_1 -3$ clean intervals of size
at least $N^{1/4}/4$.
By Lemma \ref{lem-L2analysisGap}, the number of clean intervals
of size at least $2$ decreases by at most $F_2$ from the end of Layer $1$
to the end of Layer $2$.
Therefore, the number of clean intervals of size at least $N^{1/4}/4$ 
at the end of Layer $2$ 
will be at least $N^{1/4}/4 -44F_1 - F_2 -3$.
According to Lemma \ref{lem-L2contents}, as long as $N^{1/4}/4 \ge \log N + 5$,
these intervals will all have the correct $x$ at the end of Layer $2$.

$|x| = n$ is at most $\log N$, which means that there is a $\delta$ such that 
for large enough $n$, $2^{\delta n} \le N^{1/4}/4$.
Therefore, in any interval of size at least $N^{1/4}$,
the computation in Layer $3$ has enough room to finish.
Thus, if $x \not \in L$, then each interval that is clean at the end of Layer $2$
and has size at least $N^{1/4}/4$ will either contain
an illegal translation square
(from Layer $2$ to Layer $3$), an illegal square in Layer $3$ corresponding
to an incorrect step of the Turing Machine, or a square containing a rejecting
state.
Therefore, the total cost due to rejecting computations is at least
$N^{1/4}/4 -44F_1 - F_2 - F_3-3$. Either $44F_1 + F_2 + F_3 + 3 \ge N^{1/4}/8$
or $R \ge N^{1/4}/8$.
\end{proof}

\section{Function Weighted Tiling is $\fpnexp$-complete}
\label{sec-FWT}

We now turn our attention to the Function Weighted Tiling  Problem,
which is to compute  the cost of the minimum cost tiling for an $N \times N$ grid.
We will show that FWT in $2$-dimensions is complete for $\fpnexp$. 

\subsection{Containment}

We will first argue that FWT is in $\fpnexp$. Consider the decision problem whose input is a a positive integer $N$
and threshold $\tau$. The question is whether the minimum cost tiling for $\tile$ on an $N \times N$ grid
is less than or equal to $\tau$. This problem is in $\nexp$ since a tiling of cost less than or equal to $\tau$
is of size $N^2$ and can be checked in time $O(N^2)$. 
We can therefore use an oracle to $\nexp$ to binary search for the energy  of the optimal tiling. The cost of any tiling is between $c_1 (N-1)^2$
and $c_2(N-1)^2$, where $c_1$ is the smallest cost for any square and $c_2$ is the largest
cost for any square. 
Therefore, the number of queries will be $O(\log N)$ which is polynomial in the size of the input.

\subsection{Hardness}

An outline of the proof is given in Section \ref{sec-introfunction}.
We will reduce from a function $f \in \pnexp$. Let $M$ denote the polynomial-time Turing Machine
that computes $f$. $M$ has access to an oracle for language $L' \in \nexp$.
We will use $V$ to denote the exponential time verifier for $L'$.

We will need to bound the space  and running time used by the verifier $V$
as well as the size of the output $f(x)$
and number of oracle calls made by $M$. This can achieved by a standard padding argument.
We borrow the following version from \cite{AI21} which has the elements we need:

\begin{claim}
\label{claim-pad2}
\ifshow {\bf (claim:pad2)}  \else \fi
{\bf [Padding Argument: Lemma 2.30 from \cite{AI21}]}
If $f \in \fpnexp$, then for any constants, $c_1$ and $c_2$,
$f$ is polynomial time reducible to a function $g \in \fpnexp$ such that $g$ can computed by
polynomial-time Turing Machine $M$
with access to a $\nexp$ oracle for language $L$. The verifier for $L$ is a  Turing Machine $V$.
Moreover, on input $x$ of length $n$, $M$ runs in $O(n)$ time, makes
at most $c_1 n$ queries to the oracle. Also, the length of the queries made to the oracle is at most $c_1 n$ and the running time of $V$ as well as the size of the witness required for $V$ on any query made by $M$ is $O(2^{c_2n})$. In addition the length of the output of $g$ is at most $c_1 n$.
\end{claim}

\subsection{Analysis of Layer $2$ for Weighted Tiling Parity}

In the last row of Layer $2$, the clean intervals can be designated as 
{\em long-form} or {\em short-form}. The long-form intervals are wide enough to complete the computation of the Binary Counter Turing Machine from Layer $2$. In the short-form intervals, the head gets stuck on the right end of the interval. These short-form intervals will not cause any cost to the overall tiling assuming they do not contain
any illegal pairs or sqares.

\begin{definition}
{\bf [Long-form and Short-form Intervals]}
In the last row of Layer $2$, a clean interval is a 
short-form interval if it has size $2$ or $3$
or contains a $(q/T)$ tile. Otherwise, it is a long-form interval.
\end{definition}

\begin{lemma}
\label{lem-L2analysis}
\ifshow {\bf (lem:L2analysis)}  \else \fi
{\bf [Summary of Analysis of Layer $2$]}
Consider a tiling of an $N \times N$, where 
$N = 4n(1x^R) - 2 w(x1) + 3$ for some string $x$.
Let $F_2$ be the number of illegal squares and pairs in Layer $2$. 
Let $F_1$ denote the number of illegal squares and pairs in Layer $1$.
Let $r$ be the last row of Layer $2$ and let $S$ denote
the set of sizes of the clean intervals in row $r$. Then if $F_1 \le N^{1/4}/40$,
\begin{enumerate}
    \item $l(r) \le F_2 + 9 N^{1/2} + 2 N^{1/4}+1$.
    \item $|\{2,3,\ldots,\mu(N)+1\} - S| \le 44 F_1 + F_2 + 3$
    \item Every clean interval of size at least $\log N + 5$ is a long-form interval.
    \item Every long-form interval has the form $X~(q_l/S)~x1~B^*~T~X$.
    \item Every short-form interval has the form $X~X$ or $X~T~X$, or $X~S~0^*~(q/T)~X$.
\end{enumerate}
\end{lemma}

\begin{proof}
Any increase in length from the last row of Layer $1$
to the first row of Layer $2$ must correspond to a $\#$ in Layer 1 that
was not translated to a $\#$ in Layer 2, which would be contained in an illegal
translation square.
Any increase in length from row $r_{t}$ to row $r_{t-1}$
must come from a $\#$ tile in $r_{t}$ that has a non-$\#$ below it in $r_{t-1}$.
Both squares containing the pair of vertically aligned tiles are illegal.
Therefore, an increase in the length from the last row of Layer $1$ to the last
row of Layer $2$ is accounted for by an illegal square in Layer $2$. 
The bound given in Item $1$ on the length of the last row in Layer $2$
is the expression from Lemma \ref{lem-lengthUB} which is an the upper bound on the length of
the last row of Layer $1$  plus
$F_2$. Note that Lemma \ref{lem-lengthUB} requires that $F_1 \le N^{1/4}/40$, which is assumed in this lemma as well.

Let $r'$ be the last row in Layer $1$ and let $S'$ be the set of sizes of the clean intervals in row $r'$. 
By Lemma \ref{lem-analysisL1}, $|\{2,3,\ldots,\mu(N)+1\} - S'| \le 44 F_1  + 3$.
Also by Lemma \ref{lem-L2analysisGap}, the set of the sizes
of the clean intervals in row $r'$ is the same as row $r$, except that
all the intervals of size $1$ are dropped and
at most $F_2$ clean intervals are dropped.  Therefore
$|\{2,3,\ldots,\mu(N)+1\} - S| \le 44 F_1 + F_2 + 3$.

Items $3$ and $4$ follow from Lemma \ref{lem-L2contents}.

To prove item $5$, notice that the Figure \ref{fig-bottomRowRulesL2} shows the initialization
rules for Layer $2$. The only possible interval of size $2$ is $X~X$ and the only
possible interval of size $3$ is $X~T~X$. Since these intervals do not contain a head
tile, they will remain unchanged in the last row of Layer $2$ unless they contain
an illegal square somewhere in Layer $2$.

Finally, if an interval has size at least $4$ and does not contain any
illegal translation or initialization squares (a requirement for being a clean
interval), then it has the form $X~(q_l/S)~B^j~T~X$ in the
first row of Layer $2$, where $j$ is a non-negative integer.
If the interval remains clean until the last row of Layer $2$, then the interval
represents the state of the Binary Counter Turing Machine after $N-3$ steps.
If the interval contains a $(q/T)$ tile, then the head must have reached the $T$ at the
right end of the interval. Since the head always starts moving left when it encouners a $0$, the head can only reach the $T$ if all of the bits of the
counter are $1$. From the state $X~(q_l/S)~1^r~T~X$, the head sweeps right,
changing all the $1$'s to $0$'s. Then when it reaches the $T$, it remains stuck
in that location for the remainder of the computation, resulting in
$X~S~0^j~(q/T)~X$.

\end{proof}

The running time of the Outer Loop in Layer $3$ will be bounded by a function
of the length of the longest binary string in the last row of Layer $2$.
Therefore, we would like to bound
the number of consecutive
tiles that are $0$ or $1$ as a function of $N$ and the number of illegal pairs
and squares. Towards this goal, we will require the following additional lemma.

\begin{lemma}
\label{lem-boundBits}
\ifshow {\bf (lem:boundBits)}  \else \fi
{\bf [Bounding the Length of Binary Strings from Layer $2$]}
Consider a tiling of the $N \times N$ grid.
Let $S$ be a  sequence of $m$ consecutive tiles in the last row of Layer $2$
occupying locations $l$ through $l+m-1$. Then one
of the following must hold:
\begin{enumerate}
    \item The sequence contains a tile that is not $0$, $1$, $(q/0)$ or $(q/1)$.
    \item The sequence has length at most $\log N+3$
    \item There is a row $r_t$ in Layer $2$ such that 
    there is an illegal square in locations $l$ through $l+m-1$ spanning $r_t$ and $r_{t+1}$.
\end{enumerate}
\end{lemma}

\begin{proof}
Any $T$ or $(q/T)$ that is directly above a tile that is not a $T$ or $(q/T)$ the vertically
aligned pair of tiles is contained in an illegal square to the left and to the right.
The same holds for $S$ or $(q/S)$ tiles, as well as $X$ and $\#$
tiles.
Therefore, if locations $l$ through $l+m-1$ contain an $\#$, $X$, $(q/S)$ or $T$ tile
in the first row of Layer $2$ and those locations consist of only $0$ and $1$ tiles
in the last row of Layer $2$, there must be an illegal square in those locations
somewhere in Layer $2$. The only other possibility is for locations $l$ through $l+m-1$
to consist entirely of $B$ tiles in the first row of Layer $2$.

Let $h_1, \ldots, h_r$ be the indices of the rows in which locations $l$ through $l+m-1$
do not contain a head tile. We will prove by induction on $j$ that
the tiles in locations $l$ through $l+m-1$ are $x~B^{m-s}$, where $x \in \{0, 1\}^s$
and $N-2 - h_j \ge n(x^R)$. 
Recall that $n(x^R)$ is the numerical value of the binary
number represented by the reverse of string $R$.
Since $n(x^R) \ge 2^{s-1}$, this means that $s = |x| \le \log N + 1$.
As long as $m \ge \log N + 3$, there will be at least two $B$ tiles.
As we will argue below,
the appearance of a head tile can cause the number of $B$ symbols in the sequence
to decrease by at most $1$, so if the last row of Layer $2$ contains a head tile,
there will still be at least one $B$ tile in locations $l$ through $l+m-1$.

The first row in Layer $2$ is $r_{N-2}$ and the tiles in locations $l$ through $l+m-1$
are $B^m$, so $h_1 = N-2$ and the claim holds.

Now consider the sequence of rows
$r_{h_{j-1}}$ through $r_{h_j}$. In the first and last rows of this sequence,
locations $l$ through $l+m-1$ do not contain a head tile. In the other rows,
locations $l$ through $l+m-1$ do  contain a head tile.
By induction, the tiles in locations $l$ through $l+m-1$ are $x~B^{m-s}$, where $x \in \{0, 1\}^s$
and $N-2 - h_{j-1} \ge n(x^R)$.
Therefore $|x| \le \log N + 1$ and as long as $m \ge \log N + 3$, there are at least
two $B$'s at the right end of the sequence of tiles $x~B^{m-s}$.
If the head appears at the right end of the sequence in state $q_r$, it will change the
$B$ to a $1$, but then will transition to $(q_l/B)$ and get stuck. 
If the head appears on the right end in state $q_l$, it gets stuck at the first step.

Therefore, the head must appear at the left end of the sequence of tiles.
If the head appears in state $q_l$ it will leave the sequence in the next step
without changing $x$. If the head appears in state $q_r$ at the left end of the sequence,
assuming the sequence of tiles does not contain any illegal squares, the head will sweep right,
increment $x$ and then sweep left and leave the sequence of tiles.
The value of $n(x^R)$ goes up by at most $1$, the sequence of tiles is $x'~B^{m-s}$,
where $x'$ is the reverse of the string representing $n(x^R)+1$.
Since $h_{j-1} \ge h_{j}+1$, we have:
$$N-2 - h_{j} \ge N-2 - (h_{j-1}-1) \ge n(x^R) + 1 .$$

\end{proof}

\subsection{Layer 3 for Weighted Tiling Parity}
\label{sec-FEPlayer3}

At the end of Layer $2$, each long-form clean interval contains a string $x$
that is deterministically computed from $N$ the size of the grid.
When then string $x$ is translated to Layer $3$,
each interval will also contain a non-deterministically chosen binary string $z$.
The strings $x$ and $z$ will be co-located in the same $|x|$ tiles
in the form a string $y$ over $\{0, 1, 2, 3\}$.
The string $z$ will represent a set of guesses for the oracle responses
when the Turing Machine $M$ is run on input $x$.
The role of Layer $3$ is to ensure that each interval makes the same
non-deterministic guess in each interval.
Thus, Layer $3$ will represent a global computation (across all the intervals)
that penalizes tilings that do not have a consensus for the guess $z$
over all the intervals.

\subsubsection{Translation from Layer 2 to Layer 3}

This subsection will describe the translation of tiles from the last row of Layer 2 to the first row
of Layer 3.

The translation rules for Layer $2$ in combination with the initialization rules
for Layer $3$ ensure that long-form and short-form intervals are translated
differently from Layer $2$ to Layer $3$.

The translation rules from Layer $2$ to Layer $3$
are summarized below. The  rules will
enforce the condition that for any tile directly above a $\Box$ tile,
if the tile type in Layer $2$ is as indicated on the left,
then the tile type for Layer $3$ must be one of the choices
on the right. The state $q$ represents any state $q$ in the Turing Machine
for Layer $2$. 

\begin{eqnarray*}
X & \rightarrow & \leftb~\mbox{or}~(q_r/\rightb)~\mbox{or}~X\\
0 & \rightarrow & 0~\mbox{or}~1~\mbox{or}~+\\
(q/0)& \rightarrow & 0~\mbox{or}~1\\
1~\mbox{or}~(q/1) & \rightarrow & 2~\mbox{or}~3\\
S & \rightarrow & S~\mbox{or}~+\\
(q/S)& \rightarrow & S\\
(q/B)~\mbox{or}~ B & \rightarrow & B\\
T & \rightarrow & T\\
(q/T) & \rightarrow & +\\
\# & \rightarrow & \#\\
\end{eqnarray*}

The intervals on Layer $3$ begin and end with a tile from the
set $\{X, \leftb, \rightb\}$ or any head tile $(q/c)$ where
$c \in \{X, \leftb, \rightb\}$. So if the last row of Layer $2$ is correctly
translated to Layer $3$, then the intervals are preserved. 
The Turing Machine for layer $3$ will also have the property that it does not
change the intervals, if executed correctly.

Therefore, the definition of clean and corrupt intervals can be naturally
extended from Layer $2$ to Layer $3$. In the first row of Layer $3$,
an interval is clean  if there are no illegal squares (translation or initialization)
in rows $r_1$ and $r_2$ and the corresponding interval in the
last row of Layer $2$ was clean.
An interval in a higher row $r_t$ of Layer $3$ is clean if there are no illegal squares
in that interval in Layer $3$ spanning rows $r_{t-1}$ and $r_t$
and the interval was clean in row $r_{t-1}$.
The tag for each interval also remains unchanged. If an interval is clean,
it adopts the tag for the corresponding interval (occupying the same locations)
in the previous row. Clean intervals also adopt the same long/short-form
designation of the corresponding clean interval in the previous row.

There is still a high degree of flexibility in a legal translation of
the last row of Layer $2$ to the first row of Layer $3$.
The initialization rules for Layer $3$, summarized in 
Figure \ref{fig-L3_init}, introduces some additional constraints.

\begin{figure}[ht]
\centering
  \includegraphics[width=3.5in]{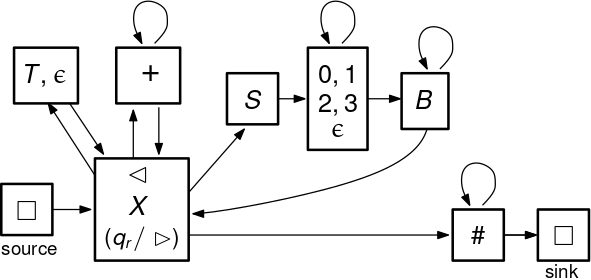}
\caption{These rules constrain the contents of the first row in Layer 3.}
\label{fig-L3_init}
\end{figure}

\begin{definition}
{\bf [Functions $f$ and $g$ mapping digits base $4$ to two bits]}
We will use $\cald$ to denote the set $\{0, 1, 2, 3\}$. The digits in $\cald$ are used
to encode two separate bits. We will think of a string $y \in \cald^n$ as mapping
to two different binary strings based on the first and second bits in the binary
encoding of each digit.
Thus $f_1(y) = x$, where $x_i = 0$ if $y_i = 0$ or $1$, and $x_i = 1$
if $y_i = 2$ or $3$.
Also, $f_2(y) = z$, where $z_i = 0$ if $y_i = 0$ or $2$, and $z_i = 1$
if $y_i = 1$ or $3$.
\end{definition}

\begin{lemma}
\label{lem-longShort}
\ifshow {\bf (lem:longShort)}  \else \fi
{\bf [Form of the Intervals Translated to Layer $3$]}
Consider a tiling of an $N \times N$ grid, where
$N = 4n(1x^R) - 2w(x1)+3$, for some binary string $x$.
Let $t$ and $t'$ represent two tiles from the set $\{X, \leftb, (q_r/\rightb)\}$.
In the first row of Layer $3$, every clean short-form interval has the form 
$$t~t'~~\mbox{or}~~t~T~t'~~\mbox{or}~~t~+^*~t'.$$
Every clean long-form interval has the form:
$$t~S~y ~B^*~T^*~t'$$
where $y \in \cald^n$ and $f_1(y) = x1$.
\end{lemma}

\begin{proof}
According to Lemma \ref{lem-L2analysis}, a clean short-form interval
at the end of Layer $2$ has the form $X~X$, $X~T~X$ or 
$X~S~0^* ~(q/T)~X$.
According the the tanslation rule, each $X$ must be translated
to a tile from $\{X, \leftb, (q_r/\rightb)\}$. Also the $T$ must be translated to $T$.
So if an interval at the end of Layer $2$ has the form $X~X$ or $X~T~X$
and it does not contain any illegal translation squares from Layer $2$ to Layer $3$,
then it has the form $t~t'$ or $t~T~t'$, where $t$ and $t'$ are from the set
$\{X, \leftb, (q_r/\rightb)\}$.

Alternatively, if the interval has the form $X~S~0^* ~(q/T)~X$, the $(q/T)$
tile must be translated to a $+$ tile. The $S$ could be translated to $S$ or $+$
and the $0$ tiles could be translated to $0$, $1$, or $+$. 
However, the initialization rules shown in Figure \ref{fig-L3_init},
show that the only tiles that can go to the left of a $+$ is either a $+$ or a tile
from $\{X, \leftb, (q_r/\rightb)\}$. Therefore, the only way to translate the tiles
in the interval so that there are in illegal initialization or translation squares
is to translate the interval to $t~+^*~t'$.

According to Lemma \ref{lem-L2analysis}, a long form interval will look like
$X~(q_l/S)~x1~B^*~T~X$ in the last row of Layer $2$.
The $(q_l/S)$ tile must be translated to an $S$ tile.
This means that the translated interval can not contain any $+$ tiles because
the only tiles that can be next to a $+$ tile are another $+$ tile or a tile
from $\{X, \leftb, (q_r/\rightb)\}$.
Thus, the bits in $x1$ are translated non-determisitically, so that
the resulting string $y$ has $f_1(y) = x1$. Note that $0$ must go to $0$ or $1$,
and $1$ must go to $2$ or $3$.
The $B$ tiles are translated to $B$ tiles and the $T$ tile to a $T$ tile.
Thus, the resulting interval, if there are no illegal translation or initialization
squares, will have the form $t~S~y ~B^*~T^*~t'$, where $f(y) = x1$
and $t$ and $t'$ are from $\{X, \leftb, (q_r/\rightb)\}$.

\end{proof}

There is still ambiguity in how tiles are translated from Layer $2$ to Layer $3$.
The horizontal rules described in the next subsection will constrain how
an $X$ from Layer $2$ is translated to Layer $3$, since there is an option
of any tile from the set $\{X, \leftb, (q_r/\rightb)\}$ in the translation
rules.
The final ambiguity is whether a $0$ is translated to $0$ or $1$ and whether 
a $1$ is translated to $2$ or $3$. Note that the first bit of $0, 1, 2, 3$
(expressed in binary) is the same as the underlying bit from Layer $2$, but the second
bit can be chosen arbitrarily. 
The goal of the Turing Machine in Layer $3$
is to check that the second bits are
all translated consistently across the intervals. 
There will be a large tiling cost if any of the strings are translated
inconsistently.

\subsubsection{Horizontal Rules for Layer $3$}

Since the Turing Machine in Layer $3$ is a global computation that operates
across the whole row and not just within a strip, 
we will need additional horizontal rules to ensure
that a valid row has exactly one head tile
and that the tape contents are bracketed on the left and right by
$\leftb$ and $\rightb$, respectively.
The tape symbols for the Turing Machine will be:
$$\Gamma = \{\leftb, \rightb, X, 0, 1, 2, 3,  \barzero, \barone, \bartwo, \barthree, B, S, T, + \}$$
We define a subset of the tape tiles $\Gamma' = \Gamma - \{ \leftb, \rightb \}$.
For every $c \in \Gamma'$, there will be a blue version of that tile and a red version
of that tile: $\bluec$ and $\redc$.
Figure \ref{fig-validConfigsL3} summarizes the horizontal rules for Layer $3$.
The set $Q$ is the set of all states for the Turing Machine in Layer $3$.
A pair of tiles $t_1~t_2$ is legal if there is an edge from the vertex containing
$t_1$ to the vertex containing $t_2$ in the graph. 
All other pairs are illegal.
A row of tiles for Layer $3$ is said to be {\em valid} if it does not contain
any pair of adjacent tiles that are an illegal pair.
Otherwise the row is {\em invalid}.

The rules enforce that the $\Box$ tile on the left must be followed by a 
$\leftb$ symbol, followed by a sequence of tape symbols, followed by
$\rightb$, followed by a sequence of $\#$ tiles, followed by the final $\Box$ tile.
In addition, exactly one of the tiles from the $\leftb$ to the $\rightb$ must
be a  head tile. Thus, in a valid row, the intervals all begin and end with an
$X$ tile, except the leftmost interval which begins with $\leftb$ and the rightmost
interval which ends with $\rightb$.

\begin{figure}[ht]
\centering
  \includegraphics[width=4.0in]{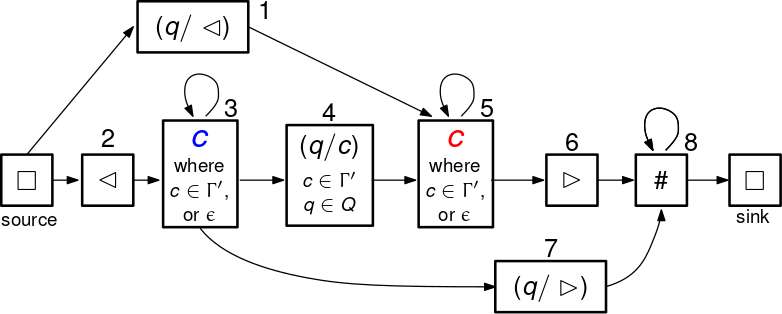}
\caption{This graph shows the horizontal rules for Layer $3$.}
\label{fig-validConfigsL3}
\end{figure}

\begin{lemma}
\label{lem-L3valid}
\ifshow {\bf (lem:L3valid)}  \else \fi
{\bf [Properties of a Valid Row in Layer $3$]}
A row is valid if and only if the row satisfies the following conditions:
\begin{enumerate}
\item The tape contents of the row has the form:
$$\Box~\leftb~(\Gamma')^*~\rightb~\#^*~\Box$$
\item There is exactly one tile is a head tile. 
\item The head is located at one of the tiles
from the $\leftb$ to the $\rightb$.
\item Any tiles from $\Gamma'$ to the left of the head tile are blue and any tiles from $\Gamma'$ to
the right of the head tile are red.
\end{enumerate}
\end{lemma}

\begin{proof}
Any path in the graph in Figure
\ref{fig-validConfigsL3} must pass through vertex $1$, $4$, or $7$ exactly once and therefore
has exactly one head tile. 
In addition, the path goes through $1$ or $2$, then vertices from $3$, $4$, and $5$,
followed by one of vertices $6$ or $7$ and then vertex $8$. Therefore,
the tape contents of each such path must have the form $\Box~\leftb~(\Gamma')^*~\rightb~\#^*~\Box$.
There are no vertices in the graph in which the head is located at a $\Box$ or $\#$ tile,
so the head must point to one of the symbols between the $\leftb$ and $\rightb$ tiles.
If vertex $3$ is reached, then it comes before vertex $4$ or $7$, which means that all the blue tiles
from $\Gamma'$ precede the head tile.
If vertex $5$ is reached, then it comes after vertex $1$ or $4$, which means that all the read tiles
from $\Gamma'$ come after the head tile.

For the converse, consider a row that satisfies all the properties in the lemma.
If the head points to the $\leftb$ symbol, then it can be generated by a path
of the form $1 \rightarrow 5^* \rightarrow 6 \rightarrow 8^*$.
If the head points to the a symbol from $\Gamma'$, then it can be generated by a path
of the form $2 \rightarrow 3^* \rightarrow 4 \rightarrow 5^* \rightarrow 6 \rightarrow 8^*$.
If the head points to the $\rightb$ symbol, then it can be generated by a path
of the form $2 \rightarrow 3^* \rightarrow 7 \rightarrow 8^*$.
\end{proof}

We describe the Turing Machine rules in the next section.
It will be important that there is a unique next step for any Turing Machine
configuration that corresponds to a valid row.

\subsubsection{The Layer 3 Turing Machine}

The Turing Machine in Layer 3 is global in the sense that it
works across all the intervals. The tape alphabet is the set:
$$\Gamma = \{\leftb, \rightb, X, 0, 1, 2, 3,  \barzero, \barone, \bartwo, \barthree, B, S, T, + \}$$
The symbols in $\{ \barzero, \barone, \bartwo, \barthree\}$ are called
 {\em marked} digits. 
For digit $j \in \cald$, marking $j$ corresponds to replacing $j$ with $\barj$
and unmarking $j$ corresponds to replacing $\barj$ with $j$.
Other than marking or unmarking digits,
the Turing Machine never changes the information on the tape.

The program of the Turing Machine consists of two nested loops.
The TM will just continually run the Outer Loop for as many steps as it is
allowed to run.
Consider the row at the beginning of an iteration of the
Outer Loop. 
Let $l_1, \ldots, l_m$ be the locations of the
$S$ tiles in this row. For each such $S$ tile, let $y_i$
be the string of digits immediately to the right of the $S$ tile.
Note that $y_i$ could be the empty string if $S$ is followed 
by a non-digit. If the Outer Loop is executed without error,
there will be a cost 
in an iteration of the Outer Loop
for each $y_k$, where $k > 1$ and $y_k \neq y_1$.

The digits in each string that have already been checked are marked.
In an iteration of the Inner Loop, the Turing Machine reads the first unchecked digit
in $y_1$ and checks that digit against the first unmarked digit in each
of the other  $y_j$'s. When a digit is checked, it becomes marked.

Figure \ref{fig-OuterLoopL3} shows the steps of the Outer Loop in pseudo-code.
Figure \ref{fig-TMrulesL3a} gives a table with all the Turing Machine rules.
An iteration of the Outer Loop begins with the head just one space to the left of $\rightb$.
The Turing Machine sweeps left, unmarking all the digits until  $\leftb$ is reached.
This begins an iteration of the inner loop.
The head sweeps right all the way 
from $\leftb$ to $\rightb$, starting in state $q_{read}$. 
When the first $S$ is encountered, the Turing Machine reads and remembers
the next unmarked digit $j$ and transitions to $q_{1j}$.
This is the next unchecked digit of $y_1$. The digit $j$ is checked.
In state $q_{1j}$ the head is looking for an $S$ which indicates
the beginning of the next $y$. 
After an $S$ is reached, it transitions to $q_{2j}$
indicating that it is looking for the next unmarked digit. When
the next unmarked digit is found, it checks if $j = k$. If $j \neq k$, a cost is incurred. 
The digit $k$
is marked and the head transitions  to $q_{j1}$ in order to look for the next $S$.
The Turing Machine also incurs a cost if in state $q_{2j}$ and a non-digit
is encountered, indicating that the current $y$ being checked is shorter than $y_1$.
When the $\rightb$ is reached, then the Turing Machine transitions to $q_{ret}$ and
sweeps left to the $\leftb$ to begin a new iteration of the inner loop.
The Outer Loop terminates when all the digits in $y_1$ have been checked.
This happens, when the state is $q_{read}$ and the head encounters a non-digit
before an unmarked digit. The head then transitions to $q_{sweep}$ and
sweeps right to the $\rightb$. While the Turing Machine is in state $q_{sweep}$,
any unchecked digit causes the Turing Machine to incur a cost. This happens is one of
the strings $y$ is longer than $y_1$. When the head reaches $\rightb$, 
it transitions to $q_{clear}$ which
begins a new iteration of the Outer Loop.

\begin{figure}[ht]
\noindent
\fbox{\begin{minipage}{\textwidth}
\begin{tabbing}
(1)~~ \= {\sc OuterLoop}:  \\
(3)  \> ~~~~~ \= Sweep left in state $q_{clear}$, unmarking every digit.\\
(4)  \> \>When $\leftb$ is reached, transition to $q_{findS}$ and move right\\
(5)  \> \> Start of {\sc Inner Loop}\\
(6)   \> \> ~~~~~ \=   Move right in state $q_{findS}$ until an $S$ or $\rightb$ is reached\\
(7)   \> \> \>  ~~~~~ \= If $\rightb$ is reached before $S$, transition to $q_{clear}$. Go to (1).\\
(8)   \> \> \>When $S$ is reached, transition to $q_{read}$\\
\\
(9)   \> \> \> Move right in state $q_{read}$ past any marked digits\\
(10) \>\>\> \> If $c \in \{B, X, +, S, T\}$ is reached before a digit, go to (24)\\
(11) \>\>\> \> If $\rightb$ is reached before a digit, go to (26)\\
(12)   \> \>\> If unmarked digit $j$ is reached, mark $j$, transition to $q_{1j}$ .\\ \\
(13)   \> \> \> Move right in state $q_{1j}$ until an $S$ or $\rightb$ is found.\\
(14)   \> \> \> \>If $\rightb$ is reached, go to (22).\\
(15)   \> \> \>If $S$ is reached, transition to $q_{2j}$.\\ \\

(16)   \> \> \> Move right in state $q_{2j}$ past any marked digits.\\
(17)   \> \> \> If an unmarked digit $k$ is reached, mark $k$ and transition to $q_{1j}$.\\
(18)   \> \> \> \>If $j \neq k$, then {\bf Cost.}\\
(19)   \> \> \> If non-digit $c$ is reached for any $c \neq \rightb, S$, transition to  $q_{1j}$. {\bf Cost.}\\
(20)   \> \> \> If $S$ is reached, stay in  $q_{2j}$. {\bf Cost.}\\
(21)   \> \> \> If $\rightb$ is reached, go to (22). {\bf Cost.}\\ \\

(22)   \> \> \> Transition to $q_{ret}$. Move left until $\leftb$ is found.\\
(23)   \> \> \> Transition to $q_{findS}$. Go to (5)\\
(24)  \> \>  Transition to state $q_{sweep}$, move right until $\rightb$ is reached \\
(25)   \> \> \> If an unmarked $j$ is encountered in state $q_{sweep}$, there is a {\bf Cost.}\\
(26)  \>  \> Transition to $q_{clear}$. Go to (1).
\end{tabbing}
\end{minipage}}
\caption{Pseudo-code for an integration of the Outer Loop for the Turing Machine in Layer $3$.}
\label{fig-OuterLoopL3}
\end{figure}

\begin{figure}[ht]
\centering
\begin{tabular}{|c|c|c|c|c|}
\hline
 & & & &  \\
& $q_{findS}$ &  $q_{read}$ & $q_{j1}$ & $q_{j2}$ \\
\hline
$k$ & Right & $(q_{k1}, \bark , R)$ & Right & 
$(q_{j2}, \bark, R)^*$ \\
\hline
$\bark$ & Right & Right & Right & 
Right \\
\hline
$S$ & Right & $(q_{sweep}, S , R)$ & $(q_{j2}, S, R)$ & 
$\mbox{Right}^{\dag}$\\
\hline
$\leftb$ & Right & Right & $(q_{j2}, \leftb, R)$ & 
$(q_{j1}, \leftb, R)$ \\
\hline
$\rightb$ & $(q_{clear}, \rightb, L)$ & $(q_{clear}, \rightb , L)$ & $(q_{ret}, \rightb, L)^{\dag}$ & 
$(q_{ret}, \rightb, R)$ \\
\hline
$c$ & Right & $(q_{sweep}, c , R)$ & Right & 
$\mbox{Right}^{\dag}$\\
\hline
\end{tabular}

\vspace{.2in}

\begin{tabular}{|c|c|c|c|}
\hline
 & & &  \\
& $q_{ret}$ &  $q_{sweep}$ & $q_{clear}$  \\
\hline
$k$ & Left & $\mbox{Right}^{\dag}$ & Left \\
\hline
$\bark$ & Left & Right & $(q_{clear}, k, L)$  \\
\hline
$S$ & Left & Right & Left \\
\hline
$\leftb$ & Right & $(q_{sweep}, \leftb , R)$ & $(q_{findS}, \leftb, R)$ \\
\hline
$\rightb$ & Left & $(q_{clear}, \rightb , L)$ & Left  \\
\hline
$c$ & Left & Right & Left \\
\hline
\end{tabular}

\caption{A summary of the rules for the Layer $3$ Turing Machine.
The word $c$ is any tape character in $\{X, +, B, T\}$. $j$ is any digit from
$\cald$.
{\bf Left} stands for $\delta(q,c) = (q,c,L)$. The word
{\bf Right} stands for $\delta(q,c) = (q,c,R)$. 
 Rule $*$ incurs a cost if $j \neq k$. Rules marked with $\dag$ incur a cost.}
\label{fig-TMrulesL3a}
\end{figure}

The Turing Machine rules are translated into legal and illegal computation squares for
as described in Section \ref{sec-TM2Tile}.
As with Layer $1$, a legal head square is legal as long as any tiles from
$\Gamma'$ to the left of a head tile is blue and any tiles from
$\Gamma'$ to the right of a head tile are red.
Also, a square is illegal if it has a tape tile directly
above another tape tile that are different in any way, including the color.

There is one final type of illegal square which occurs only in Layer $3$.
These introduce the costs indicated in the pseudo-code shown in Figure
\ref{fig-OuterLoopL3}.
Any square which contains a tile of the form
$(q_{2j}/k)$ where $j \neq k$ or $(q_{2j}/c)$, where $c \in \{B, S, T, +, X, \rightb\}$,
is an illegal square. In addition, any square that contains
$(q_{sweep}/j)$ where $j \in \cald$ is also an illegal square.
We will call these {\em illegal verification squares}.

\begin{lemma}
\label{lem-nextRow3}
\ifshow {\bf (lem:nextRow3)}  \else \fi
{\bf [Sequential Rows Represent TM Steps in Layer $3$]}
Consider a row $r$ that is valid in Layer $3$. There is a unique row $r'$ that can be
placed above $r$ such that there are no illegal computation squares that span the two
rows $r$ and $r'$. $r'$ is valid. Moreover, row $r$ corresponds to a Turing Machine 
configuration and $r'$ represents the configuration resulting from executing one step 
in the configuration $r$.
\end{lemma}

\begin{proof}
Consider a valid row $r$. If there is a  row $r'$
such that if $r'$ is placed directly above $r$,
there are no illegal computation squares, then $r'$ must be unique.
The argument is almost identical to the analogous lemma proved for the 
Layer $1$ Turing Machine given in Lemma \ref{lem-nextRow}.

Next we need to show that if $r$ is valid, then there is a valid $r'$
such that there are no illegal computation squares spanning
rows $r$ and $r'$.
Since $r$ is valid, there is exactly one head tile in $r$ and therefore
$r$ corresponds uniquely to a configuration of the Turing Machine.
Let $r'$ be the row resulting from applying one step of the Turing Machine
to the configuration represented by row $r$. 
Color all the $\Gamma'$ tiles to the left of the head tile blue and
all the $\Gamma'$ tiles to the right of the head tile red.

We will first establish that there are no illegal squares spanning rows $r$ and $r'$.
The head square must be legal because it represents one
correctly executed step of the Turing Machine.
All other tiles outside of the head square are tape tiles
and are the same in $r$ and $r'$ because they did not change in the computation
step. Moreover since $r$ is valid, all $\Gamma'$ tiles to the left of the
head square are blue and all $\Gamma'$ tiles to the right of the
head square are red. Therefore any $\Gamma'$ tiles outside of the head square
have the same color in $r$ and $r'$.

To establish that $r'$ is  valid, we will show that $r'$ has all the properties from
Lemma \ref{lem-L3valid}. 
If the head is pointing to a $\leftb$ symbol, it writes a $\leftb$ and moves right.
If the head is pointing to a $\rightb$ symbol, it writes a $\rightb$ and moves left.
If the head is pointing to a symbol from $\Gamma'$, it writes a symbol
from $\Gamma'$ and moves left or right. This guarantees properties $1$ through $3$.
Property $4$ is guaranteed by construction.
\end{proof}

\subsubsection{Analysis of Layer $3$}

We would like to argue that if the number of illegal pairs or squares in Layer $3$ is less
than a certain value, then there will be a complete error-free iteration of the Outer Loop.

\begin{definition}
{\bf [End Row for Layer $3$]}
An {\em end row} for Layer $3$ is a valid row in which the head is 
pointing to $\rightb$ symbol
in state $q_{sweep}, q_{read}, q_{findS}$ or $q_{clear}$.
\end{definition}

Note that since an end row is valid, it corresponds to a valid configuration of the
Turing Machine. In the next
next step of the Turing Machine, the head is just to the left
of the $\rightb$ symbol in state $q_{clear}$, which begins a new iteration
of the Outer Loop.

\begin{lemma}
\label{lem-L3OuterLoop}
\ifshow {\bf (lem:L3OuterLoop)}  \else \fi
{\bf [Number of Steps to Reach an End Row in Layer $3$]}
Consider a valid row $r_s$ in Layer $3$ of a tiling.
Let $y$ be the maximal string of digits to the immediate right of the left-most $S$ tile in $r_s$.
If there is no $S$ tile, then $y$ is empty.
If rows $r_s$ through $r_t$ are all valid and do not contain any illegal computation squares
and $t - s \ge 2 (|y|+2) \cdot l(r_s)$, then one of the rows in $r_{s+1}$
through $r_t$ must be an end row.
\end{lemma}

\begin{proof}
We will argue that if the Turing Machine starts in a valid configuration,
then it will reach an end configuration within $2 (|y|+2) \cdot l(r_s)$ steps.

Each iteration of the Inner Loop begins with the head just to the right of
the $\leftb$ symbol in state $q_{findS}$. 
We will first establish that it takes at most $2l(r_s)$ steps to reach the beginning
of an iteration of the Inner Loop or an end configuration.
If the state of $r_s$ is  $q_{clear}$ or $q_{ret}$, the head moves left
until a $\leftb$ symbol is reached. 
Then the head moves
right into state $q_{findS}$, for a total of at most $l(r_s)$ steps.
If the state of $r_s$ is $q_{sweep}$, $q_{findS}$, $q_{read}$, $q_{1j}$ or $q_{2,j}$
then the head will continue moving right and remain in one of those
states until the $\rightb$ is reached.
If the state is $q_{sweep}$, $q_{findS}$ or $q_{read}$ when the $\rightb$ is reached,
this is an end configuration. 
If the state is $q_{1j}$ or $q_{2,j}$ when the $\rightb$ is reached,
the Turing Machine will move left in state $q_{ret}$ until the $\leftb$ is reached,
at which point it transitions to $q_{findS}$ and moves right. This is the beginning of
an iteration of the Inner Loop. The total number of  steps so far has been $2 l(r_s)$.

If the Turing Machine starts an iteration of the Inner Loop
and the head never reaches an $S$ in state $q_{findS}$, it will reach the $\rightb$ symbol in state $q_{findS}$, which is an end configuration.
In this case, $|y_1| = 0$ and the number of steps spent in the Inner Loop is $l(r_s)$.
Otherwise, the head will eventually reach an $S$ and will transition to $q_{read}$.
If the head reaches a digit before a symbol from
$\{B, S, +, X, T, \rightb\}$ in state $q_{read}$, the digit will be marked and
the head will continue going right in state $q_{1j}$
or $q_{2j}$ until it hits the $\rightb$
symbol, in which case it transitions to $q_{ret}$, sweeps left until the $\leftb$
symbol is reached and then moves right into $q_{findS}$ to start another iteration
of the Inner Loop. 
The marked digit is part of $y$ because at the beginning of an iteration of the Inner Loop,
the state is initially $q_{findS}$ and can only transition to $q_{read}$ when the
first $S$ is reached. The state remains in state $q_{read}$ until a symbol from
$\{B, S, +, X, T, \rightb\}$ or a digit is reached. Thus, if a digit is reached
before a symbol from $\{B, S, +, X, T, \rightb\}$m the current string is still $y$.
Therefore, in each iteration of the inner loop, one additional digit from $y_1$ becomes marked and
there can be  at most $|y_1|$ iterations of the inner loop, each
of which takes at most $2l(r_s)$ steps. The Inner Loop Iterations take
a total of at most $2|y_1|l(r_s)$ steps.

Finally, if the state is $q_{read}$ and the Turing Machine reaches a symbol
from $\{B, X, +, X, T, \rightb\}$, it will sweep right in state $q_{sweep}$
(if it is not already at the $\rightb$ symbol) until the $\rightb$ is reached.
Within an additional $l(r_s)$ steps, the head will reach $\rightb$ in state
$q_{sweep}$ or $q_{read}$, which is and end configuration.
The total number of steps is at most $2l(r_s) + 2|y_1|l(r_s) + l(r_s) \le
2 (|y|+2) \cdot l(r_s)$.
\end{proof}

The set $\cald$ is the set of digits in $\{0, 1,2 ,3\}$.
We can augment $\cald$ to include the marked digits as well:
$\cald' = \cald \cup \{\barzero, \barone, \bartwo, \barthree\}$.
The function $val$ maps strings in $\cald'$ to strings in $\cald$
by changing marked digits to the corresponding unmarked digits,
so $\barj$ is mapped to $j$.
The functions $f_1$ and $f_2$ can be extended to strings from $(\cald')^*$
by first applying the function $val$ and then applying $f_1$ or $f_2$.

\begin{lemma}
\label{lem-L3yUB}
\ifshow {\bf (lem:L3yUB)}  \else \fi
{\bf [Bound on the Length of a String of Digits  in Layer $3$]}
Let $r$ be any row in Layer $3$ of a tiling of an
$N \times N$ grid. Let $y$ be a string of consecutive tiles from $\cald'$ in $r$.
Then $|y| \le (F_2 + F_3) (\log N + 4)$
\end{lemma}

\begin{proof}
Consider a sequence of consecutive tiles in locations $l$ through $l+s-1$
that are all from $\cald'$, where $s = \log N + 4$. If
locations $l$ through $l+s-1$ in the last row in Layer $2$ are
either bits ($0$ or $1$) or head tiles with bits ($(q/0)$ or $(q/1)$)
then by Lemma \ref{lem-boundBits}, there is an illegal square 
contained in locations $l$ through $l+s-1$ spanning two consecutive rows in Layer $2$.
If the last row of Layer $2$, locations $l$ through $l+s-1$ contain a tile
that is not $0$, $1$, $(q/0)$, or $(q/1)$, then either there is an illegal translation
square contained in those locations or those locations contain a tile that is not
from the set $\cald$ in the first row of Layer $3$.
Furthermore, any two vertically aligned tiles in Layer $3$ in which a tile not from $\cald'$
is directly below a tile that is from $\cald$ prime must be contained in illegal squares
on both sides. Therefore, if the locations $l$ through $l+s-1$ contain
tiles from $\cald'$ in a row from Layer $3$, then those locations must contain an illegal
 square somewhere in Layers $2$ or $3$.
Partition the locations of $y$ into consecutive locations of length $\log N + 4$.
There must be an illegal square from Layers $2$ or $3$ contained within each set of locations,
which means that there can be at most $F_2 + F_3$ sets of locations. Therefore the length
if $y$ is at most $(F_2 + F_3) (\log N + 4)$.
\end{proof}

\begin{lemma}
\label{lem-completeOuter}
\ifshow {\bf (lem:completeOuter)}  \else \fi
{\bf [Layer $3$ Completes At Least One Iteration of the Outer Loop]}
As long as $F_1 \le N^{1/4}/40$ and $F_2 + F_3 \le N^{1/4}/10 \log N$, 
there is a sequence of rows in Layer $3$ with no illegal pairs or
computation squares, in which the Turing Machine executes a complete iteration of
the Outer Loop.
\end{lemma}

\begin{proof}
Consider a set of consecutive rows in Layer $3$
with no illegal pairs or computation squares.
Let $r$ be the first row in the sequence or rows.
Let $y$ be the maximal sequence of tiles in row $r$ from the set 
$\cald'$ that are directly to the right of the leftmost $S$ tile in row $r$.
By Lemma \ref{lem-L3OuterLoop}, if the sequence lasts for at least
$2(|y|+2)l(r)$ rows, then the computation represented in the sequence of
rows will reach the end of an iteration of the Outer Loop.
Note that $l(r)$ and $y$ do not change in the course of these rows because
there are no illegal  computation squares and the computation does not change
the contents of the Turing Machine tape other than marking or unmarking digits.
Thus, if the sequence lasts for yet another $2(|y|+2)l(r)$ rows, then
the computation represented in the sequence of rows will include one full
iteration of the Outer Loop.
So as long as the sequence contains at least $4(|y|+2)l(r)$ rows,
it will be guaranteed to contain a complete iteration of the Outer Loop with no computation
errors.

Lemma \ref{lem-L2analysis} gives an upper bound on the length of the last
row of Layer $2$ which for now we will call $L$. 
Any two vertically aligned tiles with a non-$\#$ tile on top of a $\#$ tile
will be contained in illegal computation squares on both sides. So the length of
row $r$ can be at most the upper bound from Lemma \ref{lem-L3analysis} plus $F_3$.
Meanwhile, Lemma \ref{lem-L3yUB} gives an upper bound  of
$(F_2 + F_3)( \log N + 4)$ for $|y|$.
Putting these bounds together means that if we are guaranteed to have a
sequence of at least
$$4 \left[ (F_2 + F_3)( \log N + 4) + 2 \right](L+F_3)$$
consecutive rows with no illegal squares or computation squares, then
there will be a complete iteration of the Outer Loop with no computation errors.

Besides the first and last rows which are all filled with $\Box$ tiles,
there are a total of $N-2$ rows. There are at most $F_3$ rows with 
an illegal pair or square. Therefore there must be at least one sequence of at least
$(N-2-F_3)/(F_3+1)$ rows with no illegal pairs or squares. So as long as the following 
inequality holds:
\begin{eqnarray*}
\frac{N-2-F_3}{F_3+1} & \ge & 4 \left[ (F_2 + F_3)( \log N + 4) + 2 \right](F_3 + L)\\
\end{eqnarray*}
If $F_1 \le N^{1/4}/40$, then the bound from Lemma \ref{lem-L2analysis} says that
$L \le F_2 + 11 N^{1/2} + 1$.
Using the assumption from the Lemma that $F_2 + F_3 \le N^{1/4}/10 \log N$, the inequality above can be verified.
\end{proof}

\begin{lemma}
\label{lem-L3analysis}
\ifshow {\bf (lem:L3analysis)}  \else \fi
{\bf [Summary of Analysis of Layer $3$]}
Consider a tiling of an $N \times N$ grid, where
$N = 4n(1x^R) - 2w(x1)+3$, for some binary string $x$.
Let $F_i$ be the total number of illegal squares
and pairs in Layer $i$, for $i = 1, 2, 3$.
 At the end of
Layer $3$, the following conditions hold:
\begin{enumerate}
\item If $S$ is the set of sizes of clean intervals  in the last row of Layer $3$, then
$|\{2,3,\ldots,\mu(N)+1\} - S| \le 44 F_1 + F_2 + F_3+ 3$.
\item Every clean short-form interval at the end of Layer $3$ has
the form $X~X$, $X~T~X$ or $X~+^*~X$.
\item Any clean interval of size at least $\log N + 5$ is a long-form
interval.
\item If $F_1 \le N^{1/4}/40$ and $F_2 + F_3 \le N^{1/4}/10 \log N$, then
there exists a $y \in \cald^*$ such that
every long-form clean interval of has the form  $X~S ~y_i~B~ B^*~T~X$,
where $y_i \in \cald'$ and $val(y_i) = y$ and $f_1(y) = x$.
\end{enumerate}
\end{lemma}

\begin{proof}
Let $T_2$ be the set of tags corresponding to clean intervals in the last row
of Layer $2$. Let $T_3$ be the set of tags corresponding to clean intervals in the
last row of Layer $3$. Since no clean intervals are created from the last row of Layer $2$
to the last row of Layer $3$ and each clean interval adopts the same tag as
the corresponding clean interval in the preceding row, $T_3 \subseteq T_2$.
If there is a clean interval in the last row of Layer $2$ with tag $j$
and no such clean interval in the first row of Layer $3$, then
the interval contains an illegal translation square.
Similarly, a clean interval with tag $j$ that does not correspond to a clean interval
with tag $j$ in the next row, must contain an illegal square.
Since every lost clean interval corresponds to an illegal translation square or 
illegal computation square in Layer $3$, $|T_2 - T_3| \le F_3$.
If $S'$ is the set of sizes of clean intervals at the end of Layer $2$m then 
by Lemma \ref{lem-L2analysis},
$|\{2,3,\ldots,\mu(N)+1\} - S'| \le 44 F_1 + F_2 + 3$.
At most $F_3$ clean intervals are lost from the end of Layer $2$ to the end of Layer $3$.
Therefore,
$|\{2,3,\ldots,\mu(N)+1\} - S| \le 44 F_1 + F_2 + F_3 + 3$.

Item $2$  follows from the fact that the Turing Machine in Layer $3$ does not
change any tile from one row to another, except for the movement of the head
and marking or unmarking digit tiles. Therefore, if an interval is clean in  Layer $3$,
it has the same form as it did in the first row of Layer $3$.
Lemma \ref{lem-longShort}, indicates what the short-form and long-form intervals look
like in the first row of Layer $3$. 
Also, the intervals do not switch between being long-form or short-form intervals from
the last row of Layer $2$. Therefore, by Lemma \ref{lem-L2analysis}, any clean
interval of size at least $\log N + 5$ must be a long-form interval in the last row
of Layer $3$, which proves item $3$.

Now to prove item $4$. The assumption for this Lemma are the same as the conditions
for Lemma \ref{lem-completeOuter}, so there is a sequence of consecutive rows
in Layer $3$ that do not contain any illegal pairs or computation squares
that correspond to a complete iteration of the Outer Loop.
The $r_s$ be the first row in this sequence.
Find the location of the leftmost $S$ tile in $r_s$ and let $y$ be the maximal
sequence of consecutive tiles from $\cald'$ that are immediately to the right
of the $S$ tile. After the head sweeps left in state $q_{clean}$ at the 
beginning of the Outer Loop, every clean interval has the form
$X~S~y_i~B^*~T~X$, where $y_i \in \cald$. Each $y_i$ corresponds
to the $i^{th}$ clean interval as they are numbered from left to right.
$y$ may or may not be the same as $y_1$.

If there is an $s$ where digit  $s$ of $y_1$ differs from digit $s$ of $y_i$,
then on the $s^{th}$ iteration of the inner loop, there will be an illegal verification square
in interval $i$ when the Turing Machine is in state $q_{2j}$ and the first unmarked
digit in interval $i$ is $k \neq j$. This occurs in line $(18)$ of the pseudo-code
in Figure \ref{fig-OuterLoopL3}. Otherwise, if $y_1 \neq y_i$,
then $y_1$ must be a proper prefix of $y_i$ or $y_i$ is a proper prefix of
$y_1$.  If $y_1$ is  be a proper prefix of $y_i$, in the last iteration of the
Inner Loop, when there are no longer unchecked digits in $y_1$, there will
remain unchecked digits in $y_i$. These will trigger an illegal verification square
when the Turing Machine is in state $q_{sweep}$ and encounters the unchecked digit
in interval $i$. Finally if $y_i$ is a proper prefix of $y_1$, after $|y_i|$ iterations,
the Turing Machine will read digit $|y_i|+1$ of $y_1$ but will encounter no unmarked
digits in interval $S$. The state will transition to $q_{2j}$ when it reaches the $S$
in interval $i$ and it will encounter a non-digit before it encounters an unmarked digit.
Thus triggering a cost in Line $(19)$ in the pseudo-code.

Thus, every interval such that $y_i \neq y$ will no longer be clean after the 
iteration of the Outer Loop. Since every clean interval started with $f_1(y_i) = x1$,
and the string of digits in the interval does not change as long as the interval remains
clean, then $f_1(y_i) = x1$ will remain true after the iteration of the Outer Loop. 
In the remainder of the rows, if the interval remains clean, then it contains no illegal
computation squares and the string of digits remains the same, except perhaps that some
digits become marked, so Item $4$ will still hold
for the last row of Layer $3$.
\end{proof}

\subsection{Layer $4$}

The translation rules from Layer $3$ to Layer $4$ translate any
$j$ or $\barj$ tile to $j$, for  $j \in \cald$.
$S$ is translated to $(q_{s}/S)$, and $t$ is translated to $t$
for any $t \in \{+, \#, X, B, T, \rightb, \leftb \}$.
The states from Layer $3$ are all dropped, so any tile of the form
$(q/c)$ is translated to whatever $c$ would be translated to,
according to the rules above. These translation rules ensure
that every clean interval is translated to a clean interval
as long as the tiles in the interval are translated correctly.
Thus 
every clean long-form interval has the form:
$X~(q_s/S) \cald^* B^*T~X$.

Recall that the functions $f_1$ and $f_2$ map a string of length $n$ over
$\cald$ to two binary strings of length $n$.
The variable $x$ will denote the binary string $f_1(y)$ with the last bit removed
and $z$ will denote $f_2(y)$.
$x$ will be used as the input to the computation and $z$ will be used
as a guess of the answers to the queries to the oracle $L'$.
We will use $\bar{n}$ to denote the number of oracle queries made by $M$ on input $x$,
so we will only be concerned with the first $\bar{n}$ bits of $z$.

The tiling rules in Layer $4$ enforce that an $X$ tile must have a $\Box$
or $X$ tile above and below it. Similarly for $+$ and $\#$,
so the only tiles that change from one row to another
are tiles inside  long-form interval unless there is an illegal
square.
Every clean long-form interval will contain an independent Turing
Machine computation. The Turing Machine rules are translated into legal and illegal
squares as described in Section \ref{sec-TM2Tile}. 
If the head reaches the $T$ at the right end of the interval, then the interval
is not wide enough to complete the computation. In this case, 
 the computation halts and
does not incur any additional cost.

The definitions for clean and corrupt intervals carry over to Layer $4$ as well.
Since clean intervals do not change locations within a row,
two clean intervals will have the same tag 
and long/short-form designation if they occupy the same locations within
their respective rows.
An interval in the first row of Layer $4$ is clean if those tiles corresponded to a clean
interval at the end of Layer $3$ and the interval does not contain any illegal translation
or initialization squares.
In going from one row to the next row in the computation in Layer $4$, an interval is clean if
those tiles corresponded to a clean interval in the previous row and there are no illegal
squares spanning the current and previous rows in the interval.

\begin{lemma}
\label{lem-L4analysis}
\ifshow {\bf (lem:L4analysis)}  \else \fi
{\bf [Bound on the Missing Clean Interval Sizes]}
Let $F_i$ be the total number of illegal squares
and pairs in Layer $i$, for $i = 1, 2, 3, 4$.
If $S$ is the set of clean intervals at the end of Layer $4$ then
$|\{2,3,\ldots,\mu(N)+1\} - S| \le 44 F_1 + F_2 + F_3 + F_4 + 3$.
\end{lemma}

\begin{proof}
Let $T_3$ be the set of tags corresponding to clean intervals in the last row
of Layer $3$. Let $T_4$ be the set of tags corresponding to clean intervals in the
last row of Layer $4$. Since no clean intervals are created from the last row of Layer $3$
to the last row of Layer $4$ and each clean interval adopts the same tag as
the corresponding clean interval in the preceding row, $T_4 \subseteq T_3$.
If there is a clean interval in the last row of Layer $3$ with tag $j$
and no such clean interval in the first row of Layer $4$, then
the interval contains an illegal translation square.
Similarly, a clean interval with tag $j$ that does not correspond to a clean interval
with tag $j$ in the next row, must contain an illegal square.
Since every  clean interval that is lost in Layer $4$
corresponds to an illegal  square  in Layer $4$, $|T_3 - T_4| \ge F_4$.
If $S'$ is the set of sizes of the clean intervals at the end of Layer $3$, then
by Lemma \ref{lem-L3analysis},
$|\{2,3,\ldots,\mu(N)+1\} - S'| \le 44 F_1 + F_2 + F_3 + 3$.
Since at most $F_4$ clean intervals are lost from the end of Layer $3$ to the end of
Layer $4$,
$|\{2,3,\ldots,\mu(N)+1\} - S| \le 44 F_1 + F_2 + F_3 + F_4 + 3$.
\end{proof}

\subsubsection{The Layer $4$ Computation}

Recall that we are reducing from a generic
language  $f \in \fpnexp$ to Function Weighted Tiling.
$M$ denotes the poly-time Turing Machine that computes $f$ with access to a $\nexp$ oracle.
$L'$ is the  the $\nexp$-time language that is the oracle for $M$.
$V$ denotes the $\exp$-time Turing Machine that is the verifier for $L'$.

The reduction  maps a string $x$ to an integer $N$ such that
after $N-3$ steps of the Binary Counter Turing Machine 
described in Section \ref{sec-L2}, the string on the tape
is $x1$ and the head is at the left end of the tape.
According to Lemma \ref{lem-bctm},
$N = 4n(1x^R) - 2w(x1) + 3$, where $x^R$ is the reverse of string $x$,
$n(x)$ is the numerical value of the string $x$ in binary, and $w(x)$ is the
number of $1$'s in $x$.

All the  clean long-form intervals start out with
configuration 
$$(q_{s}/S) ~y~ B \cdots B ~T.$$
The computation that is initiated by state $q_{s}$
proceeds in several stages.

{\bf Stage 1:} The  computation in Stage $1$ "measures" the size of the
interval and writes the size of the interval in binary.
This is accomplished by a counter similar to the one used in \cite{GI}.
The Turing Machine uses a binary and a unary counter. 
The head shuttles back and forth between the two ends of the interval
(using the $S$ and the $T$ tiles to know when it has reached
one of the two ends). In each cycle, the head increments both
counters. When the unary counter has reached the $T$ on the right
end of the interval, it transitions to a new state which
begins the next phase of the computation.
The counting procedure actually counts the number of interior tiles in the interval
and the definition of the size of an interval includes the two endpoint tiles,
so we add $2$ to the final count to get the size of the interval.
We will call this value $r$ for a particular interval.
Note that a clean interval can increase in size by at most $1$ per segment,
so the size of any clean interval is bounded by the number of segments, 
which by Lemma \ref{lem-numSegUB}, is $O(N^{1/4} + F_1)$.
We will argue below that the cost of the minimum tiling, which is at least $F_1$
is $O(N^{1/4})$. 
As long as the size of a clean interval is $O(N^{1/4})$,
this phase of the computation takes time $O(N^{1/2})$.

{\bf Stage 2:} The next stage of the computation uses $x$, $z$, and $r$
to select a term in the cost function towards which it will contribute. Recall that $\bar{n}$ is an upper
bound on the number of oracle queries made by $M$ on an input of length $n$.
The $i^{th}$ bit of $z$ will be denoted by $z_i$.
The goal will be to have $\mbox{check}_k (z)$ intervals checking the $k^{th}$ bit of $z$, where
$$\mbox{check}_k (z) = 2^{\barn+5} [(1 - z_j) \cdot 2^{\barn-j} + z_j \cdot 2^{\barn}]$$

From the string $x$, the size of the grid $N$ is computed,
as is $\mu(N)$ the  number of intervals
after $N-3$ steps of a correct
computation of the Layer  $1$ Turing Machine.
The time and space complexity of the computation in Stage $2$ is bounded
by a polynomial in $n$, which is polylogarithmic in $N$.
By Lemmas \ref{lem-bctm} and \ref{lem-muBounds}, given $x$ of length $n$, the values of $N$
and $\mu(N)$ can be computed in $poly(n)$ time, which is polylogarithmic in $N$.

The output of the function $f$ on input $x$ with oracle responses $z$ is denoted by $f(x,z)$ which
is also computed by the Turing Machine.
Note that in the idealized case in which the interval sizes go from $\mu(N)+1$ down to $2$,
the value $I = \mu(N)+2-r$ is an almost unique identifier 
for each interval going from $1$ up to $\mu(N)$ from left to right.
Note that even in a fault-free computation, if the computation in Layer $1$ finishes in the middle of an iteration
of the Outer Loop, the sequence of interval sizes will deviate slightly from the idealized case.
The computation in each interval will perform different tasks, depending on the value of
$I$.

{~}

\renewcommand{\arraystretch}{1.4}
\begin{tabular}{|c|c|}
\hline
Value of  $I = \mu(N)+2-r $    & Action Taken \\
\hline
\hline
 $I \le 0$    & Transition to $q_{acc}$ and halt \\
    \hline
    & Compute the bit to check $k(r)$\\
 $1 \le I \le \sum_{k=1}^{\barn} \mbox{check}_k (z)$    & as described below\\
   & Go to Stage $3$\\
    \hline
 $\sum_{k=1}^{\barn} \mbox{check}_k (z)+1 \le I \le \sum_{k=1}^{\barn} \mbox{check}_k (z) + 2^3 f(x,z)$    & Transition to $q_{rej}$ and halt\\
 \hline
 $ \sum_{k=1}^{\barn} \mbox{check}_k (z) + 2^3 f(x,z) < I$    & Transition to $q_{acc}$ and halt\\
 \hline
\end{tabular}

{~}
\vspace{.1in}

There is a cost of $+1$ for any computation square that enters a $q_{rej}$ state, so rejecting computations
incur a cost of exactly $1$. Accepting computations do not incur any cost.
The value of $k(r)$ is defined to be
the smallest index $k$ such that
$$\mu(N)- r + 2\le  2^{\barn+5} \sum_{j=1}^k \mbox{check}_k (z) $$
In the next stage, the computation (if it did not stop in Stage $2$)
will check the $k(r)^{th}$ bit of $z$.

{\bf Stage 3:}
If $z_{k(r)} = 0$, then the TM transitions to $q_{rej}$
and halts.
If $z_{k(r)} = 1$, then the TM simulates the Turing Machine $M$
until the point of the ${k(r)}^{th}$ query, using $z_1, \ldots, z_{{k(r)}-1}$
as the oracle responses for the first ${k(r)}-1$ oracle queries.
Let $s$ be the input to the ${k(r)}^{th}$ oracle query.
The computation now simulates the verifier $V$ on input $s$
using a witness that is guessed. If $V$ accepts, the cost is $0$ and if $V$ rejects
the cost is $1$.
If $s \in L'$, there is a witness
that causes $V$ to accept $s$, which implies that there is a $0$-cost
tiling of that strip in Layer $4$.
If $s \not\in L'$, then any tiling will either have an illegal square
because the TM was not correctly executed or will have a cost of $1$
when $V$ terminates.

\subsection{Putting the Layers Together}

The costs for computations are realized by having 
any legal computation square which transitions to $q_{rej}$
have a cost of $1$. We will call these {\em rejecting}  squares
in order to distinguish them from illegal squares which in general have a higher cost.
We are finally ready to define the cost of a square over all four layers of the tiling.
Consider a square of four tiles where each tile is described by it's tile type for each later.
Let $f_i$ denote the indicator variable that is $1$ if the Layer $i$ tile types are an illegal
square for Layer $i$, and is $0$ otherwise. Let $p_1$ and $p_3$ designate if the square
has an illegal pair in its bottom two tiles for Layers $1$ and $3$.
Let $r$ be an indicator variable
denoting whether the square is a rejecting square in Layer $4$.
The cost of the square is:
$$r + 48 (f_1 +p_1) + 5(f_2 + p_3 + f_3 + f_4)$$
Thus if $F_i$ is the number of illegal pairs or squares in Layer $i$, and if $R$ is the number of rejecting squares in Layer $4$,
the total cost of a tiling is
$$R + 48 F_1 + 5(F_2 + F_3 + F_4)$$
$48 F_1 + 5(F_2 + F_3 + F_4)$ is the cost from illegal pairs and squares. 
We will call $R$  the {\em rejection cost} for the tiling.

We need to establish that regardless of whether $x \in L$,
the minimum cost
tiling has no illegal pairs or squares and the choice of $z$ corresponds the correct
oracle responses for the queries to language $L'$. 
We first establish an upper bound on the minimum cost 
tiling of an $N \times N$ grid.

\begin{lemma}
\label{lem-ub}
\ifshow {\bf (lem:ub)}  \else \fi
{\bf [Upper Bound on the Minimum Cost of a Tiling]}
There is a tiling of the $N\times N$ grid whose cost is at most $N^{1/4}/4 \log N$.
\end{lemma}

\begin{proof}
Consider a tiling with no illegal pairs or squares in which the string $z$ used for the oracle output bits
is all $0$'s. 
According to Lemma \ref{lem-errorFreeSizes},
the sizes of the intervals are contained in $\{\mu(N)+2, \ldots, 1\}$
with at most one duplicate. 
Let $T= \sum_{k=1}^{\barn} \mbox{check}_k (z) + 2^3 f(x,z)$. 
The leftmost $T$ intervals will incur a cost of $1$
since all the query responses are assumed to be $0$.
The other intervals will not incur any cost.
Since there can be at most one duplicate in the range $\mu(N)+1, \ldots, \mu(N)-(T-2)$,
the total cost will be at most $T+1$.
The  highest order bit of $\mbox{check}_1 (z)$ is in location $\barn - 1$.
Therefore, the highest order bit of $T$ is $2 \barn +4$ and the
value of $T+1$ is at most $2^{2 \barn+5}$, which by Claim \ref{claim-pad2} for large enough $N$ can assumed to
be  at  most $ N^{1/4}/4 \log N$.
\end{proof}

\begin{lemma}
\label{lem-oneCleanInt}
\ifshow {\bf (lem:oneCleanInt)}  \else \fi
For sufficiently large $N$,
the minimum cost tiling has at least one clean long-form interval at the end of Layer $4$.
\end{lemma}

\begin{proof}
If $S$ is the set of sizes of clean intervals 
in the last row of Layer $4$, then 
Lemma \ref{lem-L4analysis} says that $|\{2,3,\ldots,\mu(N)+1\} - S| \le 44 F_1 + F_2 + F_3 + F_4 + 3$.
By Lemma \ref{lem-muBounds}, the value of $\mu(N)$ is at least $ N^{1/4}/2$. If the tiling has a minimum cost 
then by Lemma \ref{lem-ub}, 
$$44 F_1 + F_2 +  F_3 +  F_4 \le \frac{ N^{1/4}}{4 \log N}.$$
Therefore 
$|\{2,3,\ldots, N^{1/4}/2+1\} - S| \le N^{1/4}/4 \log N + 3$. By Lemma \ref{lem-L3analysis}, any clean
interval of size
at least $\log N + 5$ is a long-form interval. Therefore for large enough $N$,
there is at least one long-form interval.
\end{proof}

Any tiling in which $F_3$, the number of illegal pairs or squares in Layer $3$, is greater than
$ N^{1/4}/10 \log N$ will have cost at least $\frac 5 {10}  N^{1/4}/\log N$ since each illegal
pair or square in Layer $3$ contributes $5$ to the overall cost. We know from Lemma \ref{lem-ub} that such
a tiling will not be a minimum cost tiling. Therefore,  we can ignore those tilings and  assume that the condition in Item $4$ of Lemma \ref{lem-L3analysis} are met.
This implies that there is a single $y \in \cald^*$
in every clean long-form interval  at the end of Layer $3$.
Define $Cost(y, c)$ to be the cost of the minimum cost tiling whose clean large-form intervals all
have string $y$ and that have $c$ illegal pairs or squares.

The following Lemma says that we can ignore tilings with illegal pairs or squares.

\begin{lemma}
\label{lem-ignoreIllegal}
\ifshow {\bf (lem:ignoreIllegal)}  \else \fi
For every $y$, and every $c > 0$, $Cost(y, 0) \le Cost(y, c)$.
\end{lemma}

\begin{proof}
Fix the string $y \in \cald^n$. Let $T_c$ be a minimum cost tiling with string $y$ and $c$ illegal squares.
Similarly for $T_0$.
Let $C_i$ denote the number of illegal squares on Layer $i$ in $T_c$.
We only know that at least one of the $C_i$'s is positive.
The cost of $T_c$ from illegal squares is $48 F_1 + 5(F_2+F_3+F_4)$.
The cost of $T_0$ from illegal squares is $0$.
Let $R(T)$ denote the rejection cost of tiling $T$.
We will prove that $R(T_0) - R(T_c) \le 48 F_1 + 5(F_2+F_3+F_4)$.

Let $s_1, \ldots, s_m$ be the the lengths of the clean intervals in the last row of
Layer $4$ for $T_c$.
Any clean interval of size $s$ in $T_c$ incurs exactly the same cost as a clean interval
of size $s$ in $T_0$. 
Because the intervals are clean, the computations inside those intervals is the same
and correct. Since the size of the interval is the same, the value of $y$ and $r$ for the
two intervals is the same. Therefore the bit $k(r)$ of $z$ that is checked
in the interval is the same. Since $T_c$ and $T_0$ are both assumed to be minimum cost
tilings, the best witness for each computation in Layer $4$ is chosen. In other words, if there is a $0$-cost tiling for that interval, it will be used in both $T_c$ and $T_0$.

If $T_0$ has $\mu(N)$ intervals, those intervals will all have sizes in the range 
$1$ through $\mu(N)+2$. 
Furthermore, according to Lemma \ref{lem-errorFreeSizes},
there is only one number in that range such that $T_0$ has two intervals of that size.
This one extra interval will contribute a cost of at most $1$ to the overall cost.
If $S$ is the set of sizes of the clean intervals at the end of Layer $4$ in $T_c$, then
the difference in rejection costs between $T_0$ and $T_c$ is at most
$|\{2,3,\ldots,\mu(N)+1\} - S |+1$.

According to Lemma \ref{lem-L4analysis}, 
$|\{2,3,\ldots,\mu(N)+1\} - S| \le 44 F_1 + F_2 + F_3 + F_4 + 3$, which means that
$$R(T_0) - R(T_c) \le  44 F_1 + F_2 + F_3 + F_4 + 4$$
Since at least one $F_i$ is positive,
$$44 F_1 + F_2 + F_3 + F_4 + 4 \le 48 F_1 + 5(F_2+F_3+F_4)$$
\end{proof}

We can now focus on tilings that have no illegal squares. 

\begin{lemma}
\label{lem-wideEnough}
\ifshow {\bf (lem:wideEnough)}  \else \fi
In any tiling with no illegal squares, 
the largest $T$ intervals will be wide enough to complete their computations,
where $T= \sum_{k=1}^{\barn} \mbox{check}_k (z) + 2^3 f(x,z)$.
\end{lemma}

\begin{proof}
Suppose input string $x$ maps to the number $N$ in the reduction.
According to Lemma \ref{lem-muBounds},
the number of intervals at the end of Layer $1$ is at least $N^{1/4}/2$.
The largest $N^{1/4}/4$ of these intervals have size at least $N^{1/4}/4$.
By Claim \ref{claim-pad2}, any of these intervals will be large enough
to complete a computation of the verifier $V$.
All the other computations in Layer $4$ as well as the other layers
are polynomial in $n$ and therefore polylogarithmic in $N$.

We need to establish that $T \le N^{1/4}/4$ so that the intervals used in the computation
are among the $N^{1/4}/4$ largest. The value of $T$ is maximized if the string $z$ is all $1$'s.
In this case the value of $T$ is at most $2^{2 \barn + 5}$. By Claim \ref{claim-pad2},
we can assume that
$2^{2 \barn + 5} \le  N^{1/4}/4$.
\end{proof}

Let $\mbox{num}(k)$ be the number of intervals that check the $k^{th}$ bit of $z$.
The goal is to have $\mbox{num}(k) = \mbox{check}_k (z)$.
Let $\mbox{num}(f)$ be the number of intervals whose value $\mu(N)-r+2$
is in the range $\sum_{k=1}^{\barn} \mbox{check}_k (z) + 1$ through $\sum_{k=1}^{\barn} \mbox{check}_k (z) + 2^3 f(x,z)$.
The goal is to have $\mbox{num}(f) = 2^3 f(x,z)$.
The following lemma shows that actual values for the $\mbox{num}$ functions are
not far from the goal.

\begin{lemma}
\label{lem-blips}
\ifshow {\bf (lem:blips)}  \else \fi
For any $S \subseteq [\bar{n}]$, in a fault-free tiling:
$$\sum_{j \in S} \mbox{check}_j (z) + 2^3 f(x,z) - 1
\le  \mbox{num}(f) + \sum_{j \in S} \mbox{num}(b) \le \sum_{j \in S} \mbox{check}_j (z) + 2^3 f(x,z)
+ 2$$
\end{lemma}

\begin{proof}
The set of interval sizes $r$ that check the $k^{th}$ bit are exactly those for which
$\mu(N)+2-r$ is in the range 
$$\sum_{j =1}^{k-1} \mbox{check}_j (z) +1
~~~~\mbox{through}~~~~\sum_{j =1}^{k} \mbox{check}_j (z).$$
The value of $\mbox{num}(f)$ are the number of intervals with size $r$ such that
$\mu(N)+2-r$ is in the range
$$\sum_{j =1}^{\barn} \mbox{check}_j (z) +1
~~~~\mbox{through}~~~~\sum_{j =1}^{\barn} \mbox{check}_j (z) + 2^3 f(x,z).$$
Note that all of the above ranges are disjoint and contained in $\{2, \ldots, \mu(N)+2\}$.
If the multi-set of all $\mu(N) - r +2$ for all the intervals
is exactly $1$ through $\mu(N)$, then
for every $k$, $\mbox{num}(k)$ is exactly $\mbox{check}_k (z)$
and $\mbox{num}(f) = 2^3 f(x,z)$. 
According the Lemma \ref{lem-errorFreeSizes}, the multi-set of interval sizes for a correct computation
is contained in $\{2, \ldots, \mu(N)+2\}$. Moreover, there are at most
two integers missing from this range and at most one duplicate. The Lemma follows.
\end{proof}

Let $\bar{y}$ be a  string such that $f(\bar{y}) = x1$ and
$g(\bar{y}) = \bar{z}$, where $\bar{z}$ is the correct answer
to all of the oracle queries made by $M$ to the $L'$ oracle on input $x$.
Note that the string $\bar{y}$ may  not be unique because the number of
oracle calls $\bar{n}$ is less than $n$, the number of bits in $\bar{z}$,
so bits $\bar{n}+1$ through $n$ of $\bar{z}$ can be arbitrary.

\begin{lemma}
\label{lem-rightY}
\ifshow {\bf (lem:rightY)}  \else \fi
Consider a tiling of an $N \times N$ grid, where
$N = 4n(1x^R) - 2w(x1)+3$, for some binary string $x$.
Let $\bar{y}$ be a  string such that $f_1(\bar{y}) = x1$ and
$g(\bar{y}) = \bar{z}$, where the $j^{th}$ bit of $\bar{z}$ is the correct answer
answer to the $j^{th}$ oracle query, for $j = 1, \ldots, \barn$.
For every $y \in \cald^n$, $Cost(\bar{y},0) \le Cost(y, 0)$.
\end{lemma}

\begin{proof}
By Lemma \ref{lem-L3analysis}, in the last row of Layer $3$
every clean long-form interval has a $y$ such that $f_1(y) = x1$.
Since the translation rules preserve the string, then any clean interval in the first
row of Layer $4$ will also have $f_1(y) = x1$.

Let $\barn$ be the number of oracle calls made by Turing Machine $M$ on input $x$.
We need to establish that the minimum is achieved when the first $\barn$ bits of $f_2(y)=z$
are the same as the first $\barn$ bits of $\bar{z}$.

By induction. Assume that we have established that the first $k$ bits of the
minimum $z$ must match $\bar{z}$ in order to achieve the minimum cost. The inputs to the first $k+1$ oracle
queries are now fixed. Call these $s_1, \ldots, s_{k+1}$.
Now suppose that $\bar{z}_{k+1} = 1$. That means $s_{k+1} \in L'$.
Any string that agrees with $\bar{z}$ in the first $k$ bits and has $z_{k+1} = 0$,
will pay a cost of $\mbox{num}(k+1)$ which is at least $\mbox{check}_{k+1}(x,z) - 1 = 2^{2\bar{n}-k+4} - 1$. Suppose instead we use
the string that has $z_{k+1} = 1$ followed by a string
of zeros. The intervals that  are checking bit $k+1$ can be tiled at $0$ cost
because $s_{k+1}$ is in fact in $L$. The intervals that are checking bits $k+2$ through
$\bar{n}$ as well as the intervals implementing the cost $2^3 f(x,z)$
will incur a cost of $\mbox{num}(f) + \sum_{j=k+2}^{\bar{n}} \mbox{num}(j)$ which by Lemma \ref{lem-blips},
is at most 
$$2^3 f(x,z) +  \sum_{j=k+2}^{\bar{n}} \mbox{check}_j (x,z)+2 \le 2^3(2^{2 \barn - k + 1} - 1) + 2  =
2^{2\bar{n}-k+4} - 6$$
Since this is less than the cost of the incorrect guess $z_{k+1}=0$, the incorrect
guess for $z_{k+1}$ can not yield the minimum cost when the correct guess $\bar{z}_{k+1} = 1$.

Now suppose that $\bar{z}_{k+1} = 0$. That means $s_{k+1} \not\in L'$.
Any string that agrees with $\bar{z}$ in the first $k$ bits and has $z_{k+1} = 1$,
will pay a cost of $\mbox{num}(k+1)$ which is at least $2^{2 \barn +5}$. Note that since 
$s_{k+1} \not\in L'$, when the verifier $V$ is  run on input $s_{k+1}$, it must
reject, which means that the intervals that check bit $k+1$ will incur a cost of $1$.
Suppose instead we use $z_{k+1} = 0$ and $0$'s for the remaining bits of $z$.
The cost will be $\mbox{num}(f) + \sum_{j=k+1}^{\bar{n}} \mbox{num}(j)$ which by Lemma \ref{lem-blips},
is at most 
$$2^3 f(x,z) +  \sum_{j=k+1}^{\bar{n}} \mbox{check}_j (x,z) + 2\le 2^3 (2^{2 \barn-k+2} -1) + 2 
=
2^{2\bar{n}-k+5} - 6.$$
Since this is less than the cost of the incorrect guess $z_{k+1}=1$, the incorrect
guess for $z_{k+1}$ can not yield the minimum cost when the correct guess $\bar{z}_{k+1} = 0$.

We have shown that regardless of the true value for $\bar{z}_{k+1}$, an incorrect guess
for $\bar{z}_{k+1}$ will result in  a higher cost than the correct guess.
\end{proof}

We now have all the pieces in place to prove the reduction.

\begin{theorem}
\label{th-WFTfinal}
\ifshow {\bf (WFTfinal)}  \else \fi
Consider a tiling of an $N \times N$ grid, where
$N = 4n(1x^R) - 2w(x1)+3$, for some binary string $x$.
Then the value of $f(x)$ can be recovered from
the cost of the minimum cost tiling of an $N \times N$ grid.
\end{theorem}

\begin{proof}
By Lemmas \ref{lem-ignoreIllegal} and \ref{lem-rightY}, the minimum cost tiling
does not have any illegal pairs or squares and guesses a $y$ that maps to the correct
$x$ for the input and the correct $z$ for the oracle responses.
The overall cost of the tiling will be
$$\sum_{j=1}^{\barn} (1 - z_j) \mbox{num}(j) + \mbox{num}(f)$$
According to Lemma \ref{lem-blips}, the value will be one larger or two smaller than 
$$2^{\barn+5} \sum_{j=1}^{\barn} (1 - z_j) 2^{\barn-j} + 2^3 f(x, \bar{z})$$
By dividing the minimum cost tiling by $8$ and rounding to the nearest integer,
the lowest order $\barn$ bits will be the value of $f(x, \bar{z}) = f(x)$.
\end{proof}

\section{Parity Weighted Tiling}
\label{sec-PWT}

The construction for Parity Weighted Tiling is almost exactly the same as with Function Weighted Tiling.
We slightly modify the translation of intervals from Layer $3$ to Layer $4$ as follows.
The translation rules from Layer $3$ to Layer $4$ translate any
$j$ or $\barj$ tile to $j$, for  $j \in \cald$.
$S$ is translated to $(q_{s1}/S)$ or $(q_{s2}/S)$, and $t$ is translated to $t$
for any $t \in \{+, \#, X, B, T, \rightb, \leftb \}$.
The states from Layer $3$ are all dropped, so any tile of the form
$(q/c)$ is translated to whatever $c$ would be translated to,
according to the rules above. These translation rules ensure
that every clean interval is translated to a clean interval
as long as the tiles in the interval are translated correctly.
Thus 
every clean long-form interval has the form:
$X~(q_s/S) \cald^* B^*T~X$.
The ambiguity in whether $S$ is translated to $(q_{s1}/S)$ or $(q_{s2}/S)$
is resolved by the translation rule that $\leftb~S$ must be translated
to $\leftb~(q{s1}/S)$ and $X~S$ must be translated
to $X~(q{s2}/S)$.
Thus the leftmost interval (if it is a clean long-form interval)
looks like $(q_{s1}/S) \cald^* B^*E$ and every other clean long-form interval
looks like $(q_{s2}/S) \cald^* B^*E$.

We are now reducing from a language $L \in \pnexp$ which is computed by a polynomial time
Turing Machine $M$ with access to a $\nexp$ oracle. 
All intervals except for the leftmost interval behave in exactly the same way as for the function
version, except that they incur a cost of $2$ for entering a rejection state. Note that we can think of the function
for the decision problem as just mapping to a single bit, depending on whether $M$ accepts or rejects.
The computation in the leftmost interval
starts in a different start state which indicates that it will simulate the Turing Machine 
$M$ on input $(x,z)$ and incur a rejection cost of $+1$ depending on whether $M$ accepts or not.

Finally multiply the cost of an illegal pair or square from the Function Weighted Tiling construction
by $2$. Note that all costs, except the cost of the computation in the leftmost interval are all a 
multiple of $2$ times their corresponding cost in the Function Weighted Tiling construction,
so the analysis in Section \ref{sec-FWT} comparing costs of different tilings still holds.

\begin{theorem}
\label{th-FEPfinal}
\ifshow {\bf (th:FEPfinal)}  \else \fi
Consider a tiling of an $N \times N$ grid, where
$N = 4n(1x^R) - 2w(x1)+3$, for some binary string $x$.
Then the minimum cost tiling of an $N \times N$ grid is odd if $x \in L$ and is
even if $x \not\in L$.
\end{theorem}

\begin{proof}
By Lemmas \ref{lem-ignoreIllegal} and \ref{lem-rightY}, the minimum cost tiling
does not have any illegal pairs or squares and guesses a $y$ that maps to the correct
$x$ for the input and the correct $z$ for the oracle responses.
Since the string $y$ is the same in all the long-form intervals,
the  the leftmost interval has the correct input $x$ and the correct outputs to
the oracle queries. All costs in the tiling are even, except for the penalty
at the end of the computation of the leftmost interval. This computation will
accept if and only if $x \in L$. If the computation accepts, the cost in this interval is
$0$ and the overall cost of the minimum cost tiling is even.
If the computation rejects, then the cost in this interval is $1$ and the
cost of the entire tiling is odd.
\end{proof}

\section{Acknowledgements}

We are grateful to the Simons Institute for the Theory of Computing, at whose program on the ``The Quantum Wave in Computing'' this collaboration began.

\bibliographystyle{alpha}
\bibliography{references}

\end{document}